\documentclass[11pt]{article}

\usepackage{amssymb,amsmath,amsthm,amsfonts}

\usepackage[T1]{fontenc}
\usepackage[tt=false, type1=true]{libertine}
\usepackage[varqu]{zi4}
\usepackage[libertine]{newtxmath}

\usepackage[margin=1in]{geometry}
\usepackage{enumitem}
\usepackage{mathtools}
\setlist{nosep,topsep=0pt,leftmargin=*}

\usepackage{hyperref}
\hypersetup{
    colorlinks,
    allcolors=myblue
}

\usepackage[sortcites,sorting=nyt,style=alphabetic]{biblatex}
\usepackage[lined, ruled]{algorithm2e} % For algorithms
    
    \SetAlFnt{\small}
    \SetAlCapFnt{\small}
    \SetAlCapNameFnt{\small}
    \SetAlCapHSkip{0pt} 
    \IncMargin{-\parindent}
\usepackage{enumitem}
\setlist{nosep,topsep=0pt,leftmargin=*}

\usepackage{xspace,nicefrac,tcolorbox,xifthen,tikz,subcaption,enumitem,physics,suffix,wrapfig}
\usetikzlibrary{shapes, arrows, positioning}

\usepackage[bb=dsserif]{mathalpha}

\usepackage{xcolor}
    \definecolor{myred}{HTML}{ea4335}
    \definecolor{mygreen}{HTML}{41a756}
    \definecolor{myblue}{HTML}{4285f4}

\usepackage[capitalize]{cleveref}
\renewcommand{\paragraph}[1]{\smallskip\noindent\textbf{#1.}}

\DeclareMathOperator*{\argmax}{\arg\max}
\DeclareMathOperator*{\argmin}{\arg\min}

\usepackage{xifthen}
\newcommand{\Line}[4]{%
    #1&%
    \;\ifthenelse{\isempty{#2}}{\phantom{=}}{#2}\;%
    #3%
    \ifthenelse{\isempty{#4}}{}{&&\qquad\left(#4\right)}%
}

\newcommand{\OMcC}{\texttt{McORA}\xspace}

\newcommand{\GFQ}{\texttt{GFQ}\xspace}

\newcommand{\CR}{\texttt{CR}\xspace}
\newcommand{\OPT}{\texttt{OPT}\xspace}
\newcommand{\ALG}{\texttt{ALG}\xspace}
\newcommand{\PRB}{\textsc{Threshold}\xspace}
\newcommand{\U}{\texttt{U}\xspace}
\newcommand{\Q}{\texttt{Q}\xspace}

\newcommand{\PF}{\texttt{PF}\xspace}
\newcommand{\NSW}{\texttt{NSW}\xspace}
\newcommand{\HMF}{\texttt{HMF}\xspace}
\newcommand{\MMF}{\texttt{MM}\xspace}

\newcommand{\GBF}{\texttt{GBF}\xspace}

\newcommand{\TTL}{\texttt{TTL}\xspace}
\newcommand{\LRU}{\texttt{LRU}\xspace}
\newcommand{\FIFO}{\texttt{FIFO}\xspace}

\newcommand{\SAMT}{\texttt{SAM}-\textsc{Threshold}\xspace}

\newtheorem{theorem}{Theorem}
\newtheorem{corollary}{Corollary}

\newtheorem{remark}{Remark}

\newtheorem{assumption}{Assumption}

\newtheorem{lettereddef}{Definition}

\newtheorem{letteredcor}{Corollary}

\newtheorem{letteredprop}{Proposition}

\newtheorem{letteredlemma}{Lemma}

\theoremstyle{definition}

\newtheorem{definition}{Definition}

\newcommand{\citet}[1]{\textcite{#1}}

\title{Online Allocation with Multi-Class Arrivals:\\ Group Fairness vs Individual Welfare}

\author{
    Faraz Zargari\thanks{University of Alberta. Email: \texttt{fzargari@ualberta.ca}}\\
    \and
    Hossein Nekouyan Jazi\thanks{University of Alberta. Email: \texttt{nekouyan@ualberta.ca}}\\
    \and
    Bo Sun\thanks{University of Waterloo. Email:
    \texttt{bo.sun@uwaterloo.ca}}\\
    \and 
    Xiaoqi Tan\thanks{University of Alberta. 
    Email: \texttt{xiaoqi.tan@ualberta.ca}}
}

\date{\vspace{-25pt}}

\begin{document}

\maketitle

\begin{abstract}
    We introduce and study a multi-class online resource allocation problem with group fairness guarantees. The problem involves allocating a fixed amount of resources to a sequence of agents, each belonging to a specific group. The primary objective is to ensure fairness across different groups in an online setting. We focus on three fairness notions: one based on quantity and two based on utility. To achieve fair allocations, we develop two threshold-based online algorithms, proving their optimality under two fairness notions and near-optimality for the more challenging one. Additionally, we demonstrate a fundamental trade-off between group fairness and individual welfare using a novel representative function-based approach. To address this trade-off, we propose a set-aside multi-threshold algorithm that reserves a portion of the resource to ensure fairness across groups while utilizing the remaining resource to optimize efficiency under utility-based fairness notions. This algorithm is proven to achieve the Pareto-optimal trade-off. We also demonstrate that our problem can model a wide range of real-world applications, including network caching and cloud computing, and empirically evaluate our proposed algorithms in the network caching problem using real datasets.
\end{abstract}

\section{Introduction}\label{section-introduction}
This paper introduces a novel framework for ensuring fairness in online resource allocation with multi-class arrivals, referred to as \textit{Multi-class Online Resource Allocation} (\OMcC). We examine a scenario in which a player allocates a fixed amount of resources to sequentially arriving requests or agents from various classes or groups, with the objective of maximizing social welfare  while maintaining fairness across different classes. A key consideration in this framework is the trade-off between group-level fairness and individual welfare: ensuring fairness at the group level can restrict allocations to agents with higher valuations, potentially reducing individual welfare. The \OMcC framework models a wide range of resource allocation problems with numerous applications, including the network and web caching \cite{10.1016/j.comnet.2016.04.006, 10.1145/3453953.3453973}, cloud computing \cite{joe2013pricing, bonald2015multi, wang2014dominant}, radio access networks~\cite{aslan2024fair, mehmeti2022max}, network bandwidth allocation \cite{panayiotou2023balancing, guo2015fair,cao2022online}, and sustainable energy systems \cite{li2024balancing, jenkins2016energy}.

Research on resource allocation and fairness has been a central focus in computer science, operations research, and economics, resulting in extensive studies on fair allocation involving divisible or indivisible resources (e.g., \cite{steinhaus1948problem, dubins1961cut, 10.1145/2483852.2483870, 10.1016/j.artint.2023.103965}) and single or multiple resource types (e.g., \cite{kelly1998rate, joe2013multiresource}). A common assumption in these studies is that fairness is defined at the individual level. In addition, all these works focus on allocation problems in an offline setting, where all agents are known in advance. In practice, however, resource allocation with sequential arrivals of agents is common, where the decision maker lacks prior knowledge of the number of agents or their valuations.

To address these challenges, fair resource allocation in online settings has recently attracted significant attention. The study of online fair allocation can be broadly categorized into two streams: offline agents with sequentially arriving resources (e.g., \cite{hosseini2023class, banerjee2022proportionally, huang2023online, banerjee2022online,esmaeili2023rawlsian, ma2020group,yang2024online}) and offline resources with sequentially arriving agents (e.g., \cite{sinclair2022sequential, 10.1145/3578338.3593558,lechowicz2023time}). In the first scenario, a fixed number of agents are known in advance, while resources become available sequentially. Upon the arrival of each resource, agents reveal their valuations for the resource, and the decision is to allocate the resource to agents with the objective of ensuring certain notions of fairness (e.g., Nash social welfare~\cite{banerjee2022online,huang2023online}) based on their received utilities. This model has been further extended to consider class fairness when offline agents belong to different classes~\cite{hosseini2023class}. A celebrated example of this scenario is the well-known Adwords problem~\cite{mehta2007adwords, buchbinder2007online}. For more literature on online fair allocation in this scenario, refer to \cite{aleksandrov2020online} for a comprehensive survey.

Conversely, in the second scenario, a fixed amount of resources is available upfront, and agents (or requests) arrive sequentially. Each agent’s valuation of the resource is revealed upon arrival, and the resource must be allocated without knowledge of future agents, including the total number. One example is the allocation of cloud computing resources to online users~\cite{zhang2017optimal}. In this case, ensuring fairness across agents becomes even more challenging, as irrevocable allocations to current agents may result in unfairness to future arrivals with the same valuations, given the fixed resource upfront. Therefore, existing results often rely on prior statistical knowledge about agent arrivals (e.g., distributional information on the number of agents~\cite{sinclair2022sequential}) or aim for weaker notions of fairness (e.g., time fairness~\cite{lechowicz2023time}). In this paper, our \OMcC problem studies the second scenario, without making statistical assumptions. Instead, we consider that online agents belong to a fixed number of groups, with the number of agents in each group being unknown, and aim to ensure fairness across different groups.
We aim to explore the following key questions:

\begin{center} \textit{How can fairness be maintained across different groups in an online setting? \\ What is the trade-off between group fairness and individual welfare in this context?} \end{center}

We address these questions by investigating the \OMcC problem with group fairness guarantees under three distinct fairness notions. The first is a quantity-based fairness metric, referred to as Group Fairness by Quantity (\GFQ), which ensures that a certain amount of the resource is reserved for each group. The other two are utility-based fairness metrics: $\beta$-Proportional Fairness ($\beta$-PF) and a more flexible notion called $(\gamma, \beta)$-Fairness. The latter allows for adjustable fairness objectives through a tunable parameter $\gamma \geq 0$, including Nash Social Welfare (\NSW) and Max-Min Fairness (\MMF) as special instances.

\subsection{Our Contributions and Techniques} 
Our primary contributions can be summarized as follows.

\textbf{Online algorithms with tight group fairness guarantees.} 
We establish three group fairness notions in the context of online resource allocation with multi-class arrivals: \GFQ, $\beta$-\PF, and $(\gamma,\beta)$-fairness. We propose novel threshold-based algorithms, namely the \Q-\PRB and \U-\PRB online algorithms, which employ multi-segmental and class-dependent threshold functions to ensure fairness in an online setting. The key novelty behind these algorithms is the conceptualization of group fairness guarantee as a group-level competitive ratio, which allows us to leverage the online threshold-based algorithms framework. We demonstrate that these algorithms achieve optimal group fairness guarantees under  \GFQ and $\beta$-\PF, and provide tight guarantees for ($\gamma,\beta$)-fairness. To prove the tightness of our guarantee for ($\gamma,\beta$)-fairness, we derive a lower bound for the group fairness guarantee and show that our algorithm achieves this lower bound when $ \gamma = 1 $ and is order-optimal for all $ \gamma \neq 1 $.

\textbf{Pareto-optimal trade-off between competitiveness and group fairness.} 
We develop a novel family of algorithms, termed Set-Aside Multi-Threshold-based algorithms (\SAMT), which leverages both local and global threshold functions to explore the trade-off between fairness and competitiveness (i.e., the efficiency of online algorithms) in online resource allocation. The key novelty of our design lies in the use of two distinct types of threshold functions: one local function for each group to ensure fairness and one global function for all groups to optimize efficiency. By employing tunable parameters, the algorithm achieves a Pareto-optimal trade-off for $\beta$-\PF, allowing for flexible adjustments based on different fairness and efficiency priorities. Furthermore, building on our order-optimal algorithm design for ($\gamma,\beta$)-fairness, we demonstrate that our proposed \SAMT algorithm smoothly balances competitiveness and ($\gamma,\beta$)-fairness for all $ \gamma \geq 1 $, making it highly adaptable to various application scenarios.

Our algorithms build on the recent success of threshold-based algorithms \cite{sun2020competitive, lechowicz2023online} and posted price mechanisms \cite{tan2020mechanism} in online resource allocation. However, our approach differs significantly in the following aspects: \Q-\PRB employs a multi-sectional design where the number of sections dynamically adjusts based on the fairness guarantee required by the problem. Existing algorithms that rely on smooth threshold functions fail to achieve optimality in this context. In addition, \U-\PRB introduces a reservation-based approach that establishes specific threshold functions for each class of arrivals, prioritizing fairness over performance. This feature is unique to our algorithm, as existing algorithms fail to achieve the required group fairness guarantees in the multi-class setting. To derive the lower bound, we construct tailored hard instances for each algorithm and fairness notion. These constructions allow us to develop a novel \textit{representative function}-based approach (we call them utilization functions) to derive the lower bounds, demonstrating that our algorithms are optimal for \GFQ and $\beta$-\PF, and asymptotically optimal for the more general but challenging fairness notion of ($\gamma,\beta$)-fairness. This result provides the tightest known lower bound for these fairness guarantees. We expect that these approaches will be useful in related problems such as online network utility maximization \cite{cao2022online}, online matching \cite{hosseini2023class}, and online Nash welfare maximization \cite{huang2023online, banerjee2022online}.

Furthermore, we empirically evaluate our models on the real-world Wikipedia Clickstream dataset \cite{wikimedia_analytics} using the utility-based time-to-live (\TTL) network caching protocol \cite{dehghan2019utility} to demonstrate the trade-off between group fairness and individual welfare. The results highlight the superiority of utility-based fairness notions over quantity-based fairness, while empirically showing that achieving greater fairness comes at the cost of reduced efficiency.

\section{Related Work}\label{section-related-work}
The \OMcC with group fairness guarantees is closely related to two lines of existing work: online allocation without fairness and online fair allocation. In classic online allocation, the primary focus is on allocating goods to arriving buyers to maximize profits or minimize costs, often without considering fairness across online arrivals. Conversely, research in online fair allocation primarily focuses on achieving fairness among agents, focusing on metrics such as proportional fairness, envy-freeness, and max-min fairness (e.g., \cite{bogomolnaia2022fair, joe2013multiresource, mo2000fair}). In particular, online fair allocation can be further divided into two subclasses: fairness among offline agents and fairness among online agents. In the following, we provide a brief overview of the most relevant problems related to our work.

\textbf{Online allocation without fairness guarantees.} In the online resource allocation literature, \OMcC is closely related to problems such as the online knapsack problem~\cite{chakrabarty2008online}, the online conversion problem~\cite{el2001optimal}, the secretary problem~\cite{gardner1970mathematical}, and prophet inequality~\cite{samuel1984comparison}. In these problems, a player seeks to sell an asset to a sequence of arriving buyers with varying valuations and must make irrevocable decisions for each buyer without knowledge of future buyers. These problems are typically studied within the framework of competitive analysis~\cite{borodin2005online}, where the objective is to minimize the competitive ratio—i.e., the worst-case ratio between the returns of an offline optimal algorithm and those of the online algorithm. A key result in recent research shows that a subset of these problems can be solved (near) optimally using a unified threshold-based algorithm, with a threshold function tailored to specific problems~\cite{Sun_OKP_departure_2022,dutting2021secretaries,lechowicz2023online}. However, these works neglect fairness concerns across online buyers, which are crucial, especially when allocating essential resources like computing, energy or food. 

\textbf{Fairness among offline agents.} Fair resource allocation in online settings has gained significant attention recently, with most research focusing on allocating resources that arrive online to offline agents while maintaining fairness among them. For example, \cite{hosseini2023class} explores class fairness in online bipartite matching through \textit{envy-freeness up to one item}, while \cite{banerjee2022proportionally} addresses proportional fairness in fractional online matching. Other works, such as \cite{huang2023online, banerjee2022online}, maximize Nash social welfare in these settings. Additionally, \cite{esmaeili2023rawlsian} and \cite{ma2020group} apply max-min fairness in online bipartite matching, primarily for applications like online advertising. However, in many real-world scenarios, resources are fixed, and it is the agents who arrive online for these resources, as seen in problems such as the online knapsack problem and prophet inequality.

\textbf{Fairness among online agents.} 
Some recent research on fairness among online agents mainly focus on time fairness, which ensures that agents are treated without discrimination, regardless of their arrival time. For example, in prophet inequality, time fairness is achieved through a static threshold that treats all online agents equally~\cite{chawla2024static}. In the context of online knapsack problems, \cite{lechowicz2023time} demonstrates that no nontrivial online algorithm can guarantee time-independent fairness and proposes a modified threshold-based algorithm that ensures conditional time fairness. However, time fairness essentially treats all online agents as a single group, and, to the best of our knowledge, no existing work has explored fairness from the perspective of group fairness among online agents, i.e., each buyer belongs to a group and the fairness is guaranteed across multiple groups. The proposed \OMcC aims to fill this gap by providing a threshold-based algorithm that guarantees group fairness for online agents.

\section{Problem Formulation and Preliminaries}\label{section-problem-formulation}
In this section, we first show the formulation of multi-class online resource allocation problem and then introduce the fairness metrics we consider in this paper.

\subsection{Problem Formulation} 
A decision maker with an initial resource budget of $B$ aims to allocate the resource to a sequence of agents arriving one by one in an online manner. Upon the arrival of agent $ t \in \{1,2,\dots,T\}$, they send a request with valuation $v_t$; an immediate and irrevocable decision {$ x_t \in [0,r_t]$} must be made regarding the amount of resource to be allocated to this agent, where $r_t$ is the allocation rate limit of agent $t$ and revealed online as well. In \OMcC, each agent \( t \) belongs to a class \( j_t \in [K] \), and the utility of class \( j \) with allocation $\mathbf{x} := [x_1,\dots, x_T]$ is defined as \[ U_j(\mathbf{x}) = \sum\nolimits_{t \in [T]} v_t \cdot x_t \cdot \boldsymbol{1}_{\{j_t = j\}}, \] where \( \boldsymbol{1}_{\{j_t = j\}} \) is an indicator function that equals 1 if \( j_t = j \), and 0 otherwise. An efficient allocation is to maximize the total utility of all groups, while respecting the resource constraint $ \sum_{t\in[T]} x_t \leq B $. Given an arrival instance $ I = \{(v_1, j_1, r_1), \cdots, (v_T, j_T, r_T)\} $, let us denote by $ \OPT(I) $ the optimal total utility achieved in the offline setting when the information of the arrival sequence $ I $ is known beforehand. Mathematically, $ \OPT(I) $ can be obtained by solving the following problem. 
\begin{align}\label{p1}
    \OPT(I) =  \underset{ x_t \in [0,r_t] }{\max}\  \sum\nolimits_{j\in[K]} U_j(\mathbf{x})  \qquad {\rm s.t.}\  \sum\nolimits_{t\in[T]} x_t \leq B. 
\end{align}
On the other hand, in the online setting, the instance \( I \) is revealed sequentially and future arrival information, including the total number of agents  $T$, is unknown. Upon the arrival of agent $ t $, the decision-maker needs to make an immediate and irrevocable decisions \( x_t \). Therefore, to maximize efficiency, one aims to maximize the total utility in an online manner. However, this objective does not account for fairness across different groups, i.e., all resources may be allocated to agents from the same group with high valuations. In \OMcC, our goal is to enforce a fair allocation that can balance the utilities across different groups, and further investigate the trade-off between efficiency and fairness. By explicitly incorporating fairness into allocation decisions, \OMcC ensures an equitable distribution of resources across different classes, thereby mitigating disparities often overlooked by conventional optimization models.

Note that without incorporating fairness considerations, \OMcC reduces to the problem formulation of the online (fractional) knapsack problem in \cite{zhou2008budget}, which focuses solely on maximizing overall utility. However, in many real-world applications, fairness across different classes is critical and cannot be overlooked. Below, we highlight some of these applications.

\subsection{Illustrative Examples}\label{sec-problem-formulation-sub-an-illustrative-example}

\paragraph{Network caching management}
Consider a cache of fixed size $B$ and a sequence of $T$ files. The cache is managed using the time-to-live (TTL) policy~\cite{dehghan2019utility}, which is a general policy and can be used to model a broad class of caching policies including LRU and FIFO. In the TTL policy, each file $t$ is assigned a timer $\tau_t$ upon its arrival, i.e., the first time the file $t$ is requested. The file $t$ is stored in the cache upon a cache miss and remains in the cache for $\tau_t$ time units before being evicted. Previous work~\cite{panigrahy2017hit} shows that the hitting probability $x_t \in [0,1]$ of file $t$ can be derived as a function of $\tau_t$ under Poisson arrival requests. Thus, in the TTL policy, the hitting probability $x_t$ can be equivalently considered as the decision variable upon the arrival of file $t$. The utility of a file $t$ can be modeled by a linear utility function $v_t x_t$, where $v_t$ can represent the valuation or popularity of the file $t$. In addition, the expected number of files cannot exceed the buffer size $B$, which requires $\sum_{t\in[T]} x_t \le B$. Thus, the problem formulation~\eqref{p1} can exactly models the network caching problem to maximize the utilities of all files. 

In network caching problem, the files have different valuations and may belong to different classes. For example, a file can have versions in different languages, and the files in the same language can be considered belonging to the same class. The English version of a popular file may have a large valuation due to frequent access, while the Japanese version of the same file may have a much lower valuation. To ensure fairness in accessing the files across users with different language preferences, caching decisions must consider not only file valuations but also their class.

\paragraph{Cloud computing resource allocation}
Consider a data center with \( B \) units of computing resource, and a sequence of \( T \) job requests. 
Each job requests at most $r_t$ resources and has a valuation $v_t$ for being allocated each unit resource. Upon the arrival of request $t$, the provider determines the allocated resource $x_t \in [0, r_t]$ and the request obtains utility $v_t x_t$. The allocation is subject to a resource capacity constraint, \( \sum_{t \in [T]} x_t \leq B \). Thus, the problem formulation \eqref{p1} can be used to model the job utility maximization problem in cloud computing resource allocation.

In cloud computing, the system may include \( K \) distinct users, each submitting requests over time. These users can have varying capacities and preferences, leading to differing valuations for their requests. For instance, large corporations with access to substantial resources may submit requests with significantly higher valuations compared to individual users, potentially monopolizing the computational resources. To address this, it is essential to ensure fairness in resource allocation across different users by considering the characteristics of the users submitting requests when making allocation decisions.

\subsection{Assumptions and Performance Metrics}\label{sec-problem-formulation-sub-performance-metrics}
As mentioned before, each arriving agent $t$ belongs to a class $j_{t} \in [K]$. We assume that there exists a class-dependent finite support for possible valuation within each class, formally stated in Assumption \ref{assumption_bounded_values} below.

\begin{assumption}\label{assumption_bounded_values}
    We assume that the valuations of agents in each class $ j \in [K] $ are normalized and bounded within $ [1, \theta_j] $, i.e., $ v_t \in [1, \theta_j] $ holds for all $ j \in [K] $ and $ t \in [T]$ if $j_t = j$.
\end{assumption}

We refer to $ \theta_j $ as the \textit{fluctuation ratio} of class $ j $. Intuitively, $ \theta_j \geq 1 $ holds for all $ j\in [K] $ and a smaller fluctuation ratio indicates that arrivals within that class tend to be more homogeneous. We also assume w.l.o.g. that $\theta_1\leq \theta_2 \leq \cdots \leq \theta_K $. In many real-world applications, agents from different groups often exhibit the same minimum valuations for a given resource, reflecting its intrinsic value or production cost. For instance, in web caching, where the valuation is based on the number of impressions, the minimum number of impressions is typically uniform at 1 across categories. Similarly, in energy allocation, the lowest valuation is often tied to the baseline generation cost of energy, which remains consistent irrespective of the consumer group. Notably, the assumption of identical lower valuations is necessary only for \GFQ fairness (i.e., Section \ref{section-GFQ}). For utility-based fairness notions, this assumption can be relaxed, and our results remain valid even when different lower valuations are considered for each class.

We consider the following performance metrics to evaluate the efficiency and fairness of online algorithms for \OMcC.

\subsubsection{\textbf{Efficiency Metrics}} \label{sec-problem-formulation-sub-performance-metrics-subsub-efficiency-metrics}
We follow the \textit{competitive analysis} framework to evaluate the performance of an online algorithm using its \textit{competitive ratio}. Let $ \ALG(I) $ denote the {total utility} achieved by an online algorithm $ \ALG $. Our objective is to develop online algorithms that minimize the worst-case competitive ratio, ensuring performance close to the offline optimal solution. Formally, this is defined as $\CR^* \coloneqq \min_{\ALG} \max_{I \in \Omega} \frac{\OPT(I)}{\ALG(I)}$, where $ \Omega $ represents the set of arrival sequences that satisfy Assumption \ref{assumption_bounded_values}, and $ \OPT(I) $ is the total utility of the optimal offline solution for sequence $ I $. In this paper, efficiency is measured using the competitive ratio at the individual welfare level, where welfare refers to the utility achieved for each individual agent. Recall that we do not assume any knowledge of $T$ in the online setting. As a result, the competitive analysis framework ensures that the online solution remains competitive with the offline solution for any $T$.

\subsubsection{\textbf{Fairness Metrics}} \label{sec-problem-formulation-sub-performance-metrics-subsub-fairness-metrics}
We consider two notions of fairness, namely, quantity-based and utility-based metrics. For the first one, a fixed required amount of resource is guaranteed to be allocated to the arrivals of each class, referred to as \textit{group fairness by quantity} (\GFQ) defined as follows. 
\begin{definition}[\textsc{Group Fairness by Quantity}] \label{def_fairness_numbers}
    The total amount of resource allocated to agents from group $ j\in [K] $ is at least $m_j$ for all $ j \in [K] $, namely, an allocation $ \mathbf{x} $ satisfies \GFQ if $ \sum\nolimits_{t\in [T]} x_t \cdot  \boldsymbol{1}_{ \{j_t = j\}} \geq m_j $ holds for all $ j \in [K] $.
\end{definition}

In \GFQ, the fairness requirement $\boldsymbol{m} := \{m_j\}_{j\in[K]} $ is specified beforehand, and the objective is to design online algorithms that satisfy the given $ \boldsymbol{m} $. Since the decision-maker seeks to maximize total utility, resources are often allocated to agents with high valuations, even when they belong to the same class. \GFQ mitigates this by preventing unfair treatment across different classes.

Although intuitive, \GFQ presents two issues. First, it requires the fairness requirement $ \boldsymbol{m} $ to be predetermined, which may not always be available. Additionally, different agents may perceive the same quantity differently based on their valuations. To address these issues, we introduce two utility-based metrics. 

\begin{definition}[$\beta$-Proportional Fairness] \label{def_proportional_fairness}
    For $\beta \geq 1$, allocation $ \mathbf{x} $ is called $\beta$-proportionally fair if for any other allocation $ \mathbf{w} $, we have $\frac{1}{K} \sum_{j=1}^K \frac{U_j(\mathbf{w})}{U_j(\mathbf{x})} \leq \beta$, where \( U_j(\mathbf{w}) \) and \( U_j(\mathbf{x}) \) are the utilities of each class $j$ from the allocations $\mathbf{w}$ and $\mathbf{x}$, respectively. We say that an online algorithm is $\beta$-proportionally fair if it always produces an $\beta$-proportionally fair allocation.
\end{definition}

\begin{definition}[\textsc{$(\gamma,\beta$)-fairness}]
    \label{def-gamma-beta-fairness}
    For $ \gamma \geq 0$ and $\beta \geq 1$, allocation $\mathbf{x}$ is called $(\gamma,\beta$)-fair if for any other allocation $\mathbf{w}$, we have $\sum_{j\in[K]} f(U_j(\mathbf{w})) \leq \sum_{j\in[K]} f(\beta U_j(\mathbf{x}))$ where $ U_j(\mathbf{x})$ denotes the utility of group $j$ under allocation $\mathbf{x}$ and $f(\cdot)$ is defined as follows:
        \begin{align}
        \label{eq:f-defintion}
            f(x) = \begin{cases}
                \frac{x^{1-\gamma}}{1-\gamma}& \gamma \in [0,1)\cup (1,\infty), \\
                \log(x)& \gamma = 1.
            \end{cases}
        \end{align}
    We say that an online algorithm is $(\gamma,\beta)$-fair if it always produces a $(\gamma,\beta)$-fair allocation.
\end{definition}

The $\beta$-proportional fairness given by Definition \ref{def_proportional_fairness} is a generalization of the conventional proportional fairness (\PF) concept with an approximation ratio $ \beta \geq 1 $. When an allocation follows the 1-\PF (i.e., $ \beta = 1 $), it is commonly referred to as a \textit{proportional fair} allocation. As a utility-based fairness metric, 1-\PF has been successfully used in network resource allocation, particularly in wireless and communication networks, to balance efficiency and fairness (e.g., \cite{kelly1998rate, khan2016fairness, kelly1997charging}).  Due to future uncertainties in online settings, achieving exact 1-\PF is in general impossible. Thus, we focus on a $\beta$-approximation, hence the term $\beta$-\PF.

\begin{figure}
    \centering
    \begin{minipage}{0.46\textwidth}
        \includegraphics[trim=0.5cm 0.5cm 0.5cm 0.5cm,clip,width=0.9\linewidth]{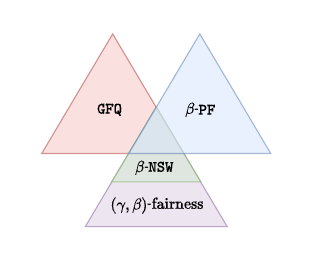}
        \caption{Relationship between \GFQ, $\beta$-\PF, and ($\gamma,\beta$)-fairness. These fairness metrics converge when $\beta = 1$ for $\beta$-\PF and ($\gamma,\beta$)-fairness, and when $m_j = \frac{B}{K}$ for \GFQ. If $ \gamma = 1 $, $ (\gamma, \beta) $-fairness reduces to  $ \beta $-\NSW.
        }
        \label{fig-fairness-metrics}
    \end{minipage}
    \hspace{0.5cm}
    \begin{minipage}{0.46\textwidth}
        \centering
        \includegraphics[trim=1cm 1.1cm 9cm 12cm,clip,width=0.9\linewidth]{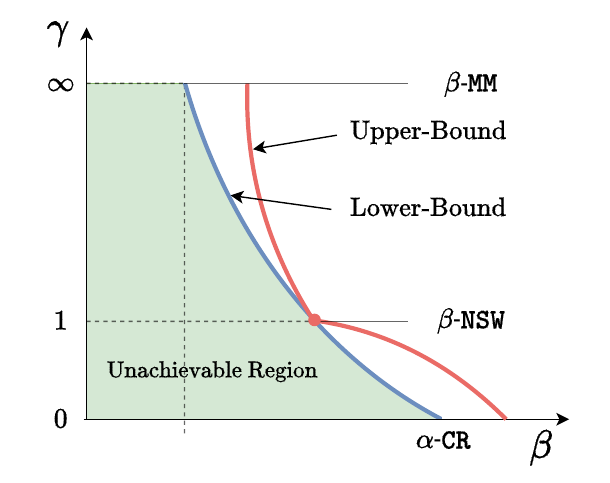}
        \caption{Graphical representation of the fairness guarantees  achieved by our algorithm for ($\gamma,\beta$)-fairness in Theorem \ref{theorem-gamma-beta-fairness} (i.e., upper bound) and the lower-bound results  in Theorem \ref{theorem-gamma-beta-lowerbound}. }
        \label{figure-graphical-representation}
    \end{minipage}
    \vspace{-5pt}
\end{figure}

The ($\gamma,\beta$)-fairness given by Definition \ref{def-gamma-beta-fairness} captures a wide spectrum of fairness measures. It is a tunable metric known as \textit{$\gamma$-fairness} \cite{bertsimas2012efficiency} (In the literature, it is commonly refereed as $\alpha$-fairness), commonly applied in network bandwidth allocation, scheduling, and optimization problems \cite{bertsimas2012efficiency}. \(\gamma\)-fairness can be considered a special case of broader fairness notions. For example, it corresponds to a specific instance of the generalized fairness framework introduced in \cite{lan2011axiomatic}. Additionally, when \(\gamma = 1\), \(\gamma\)-fairness is closely related to the Cobb-Douglas utility function under the scenario where there is a unit weight vector. This is because both frameworks prioritize proportional allocations, with \(\gamma = 1\) yielding logarithmic utility functions, which align with the Cobb-Douglas utility with equal priority of agents.

Similar to the definition of $ \beta $-\PF, achieving exact $\gamma$-fairness is in general impossible in the online setting. Thus, we focus on a $\beta$-approximation, hence the term \textit{$(\gamma, \beta)$-fairness}.  Note that different $\gamma$ values correspond to various well-known fairness criteria. For instance, when $\gamma = 1$, the ($1,\beta$)-fairness is aligned with a $\beta$-approximation of the \textit{Nash Social Welfare} (\NSW). Thus, we will also refer to ($1,\beta$)-fairness as $ \beta $-\NSW hereinafter. When $\gamma = 2$, it reflects \textit{Harmonic Mean Fairness} (\HMF). As $\gamma \to \infty$, it becomes \textit{Max-Min Fairness} (\MMF), which is referred to as $ \beta $-\MMF hereinafter. 

Before leaving this section, we give the following remarks about our efficiency/fairness metrics.

\begin{remark}[Connection between Fairness Metrics]
    Note that when $m_j = B/K$ for all $j \in [K]$, \GFQ simplifies to $1$-\PF, meaning that each class will receive at least a $1/K$ fraction of the total resource. It has also been shown in \cite{banerjee2022proportionally} that for any $ \beta \geq 1 $, an online algorithm that is $\beta$-\PF must also satisfy the $\beta$-\NSW fairness guarantee (i.e., $(1, \beta)$-fairness), as a result of the AM-GM inequality. Furthermore, when $\beta = 1$, a $1$-\PF allocation is equivalent to $1$-\NSW fairness (i.e., $(1,1)$-fairness). Thus, $1$-\NSW  is identical to $1$-\PF, and also to \GFQ when $m_j = B/K$ for all $j \in [K]$. Figure \ref{fig-fairness-metrics} illustrates these relationships.
\end{remark}

\begin{remark}[Fairness Guarantee $ \beta $]
    Due to the online nature of \OMcC, it is generally impossible to achieve $\beta $-\PF and $(\gamma, \beta)$-fairness with $ \beta = 1 $, unless it is in a trivial setting where the fluctuation ratios of all classes are equal (i.e., $ \theta_1 = \theta_2 = \cdots = \theta_K $). Taking $ (\gamma, \beta)$-fairness as an example. For each given $ \gamma \geq 0 $, there is a fundamental lower bound on $ \beta $ that no online algorithm can outperform. Figure \ref{figure-graphical-representation} gives an overview of our main results about the unachievable region due to the fundamental lower bound in Theorem \ref{theorem-gamma-beta-lowerbound} and the fairness guarantee achieved by our algorithm given by Theorem \ref{theorem-gamma-beta-fairness}. 
\end{remark}

\begin{remark}[Implications of Efficiency and Fairness Metrics]\label{remark-metric-implications}
    Efficiency and fairness metrics pursue different objectives. Online algorithms with small competitive ratios can approximately maximize the total utility of all agents, efficiently utilizing the limited resource. In contrast, a fair online algorithm aims to ensure the utilities received by different groups are balanced. For example, in the network caching example mentioned in Section \ref{sec-problem-formulation-sub-an-illustrative-example}, an efficient algorithm focuses on maximizing the overall throughput of the system. However, with \GFQ, the primary goal is to allocate a predefined portion of the cache buffer to each class of files. For \(\beta\)-\PF, the goal is to guarantee that the allocation to each class is proportional to its demand. Finally, in \((\gamma, \beta)\)-fairness, the objective is to balance fairness across each class of files and the efficiency of the system by adjusting the parameter \(\gamma\), thereby reflecting varying levels of fairness importance relative to efficiency.
\end{remark}

\section{\OMcC with \GFQ Guarantees}\label{section-GFQ}
Based on the definition of \GFQ, a reserved allocation $\boldsymbol{m}$ is provided in advance. Therefore we can reformulate the offline version of \OMcC with \GFQ requirement as follows:
\begin{align}\label{eq-gfq}
    \OPT(I) =  \underset{ x_t \in [0,r_t] }{\max}\  \sum\nolimits_{t\in[T]} v_t x_t  \quad {\rm s.t.}\ \underbrace{\sum\nolimits_{t\in[T]} x_t \leq B}_{\text{Budget constraint}},\quad \underbrace{\sum\nolimits_{t\in [T]} x_t \cdot  \boldsymbol{1}_{ \{j_t = j\}} \geq m_j}_{\text{\GFQ requirement}}, \forall j\in [K].
\end{align}

It is evident that if few agents arrive from each class $j\in[K]$, it is impossible to satisfy the \GFQ constraint in the problem formulation \eqref{eq-gfq}. Therefore, to avoid trivial settings, we introduce an additional assumption regarding arrivals from each class.
\begin{assumption}[\GFQ Arrivals]\label{assumption-gfq-arrival}
    For each class \( j \in [K] \), there are at least \( n_j \) arrivals, where \( n_j = \argmin_n \sum\nolimits_{t \in [n]} r_t \cdot \boldsymbol{1}_{\{j_t = j\}} \geq m_j\) is defined as the smallest integer ensuring that Problem \eqref{eq-gfq} remains feasible.
\end{assumption}
Assumption \ref{assumption-gfq-arrival} ensures the feasibility of Problem \eqref{eq-gfq} in both online and offline settings. For example, when $ r_t = 1 $ holds for all $ t \in [T]$, we have $ n_j = \lceil m_j \rceil$, indicating that there are at least $ \lceil m_j \rceil $ arriving agents from class $ j \in [K]$. Without the \GFQ constraint, \OMcC resembles standard online resource allocation problems such as online (fractional) knapsack \cite{sun2020competitive, tan2020mechanism}, where the key challenge lies in balancing two extremes: either allocating resources too quickly or waiting until the end to maximize profit. This issue is tackled by designing a dynamic threshold function that guides allocation decisions \cite{sun2020competitive}. To incorporate the \GFQ constraint, we propose a threshold-based algorithm dubbed \Q-\PRB, detailed in Algorithm \ref{alg-function-daynamic-threshold-K-class}. Upon receiving the first $n_j $ agents from each class $j \in [K]$, Algorithm \ref{alg-function-daynamic-threshold-K-class} ensures the corresponding fairness guarantee for that class, as defined in Definition \ref{def_fairness_numbers}, by automatically accepting {the first $n_j$ agents} regardless of their valuations, until the fairness guarantee for the class is met. As a result, an $M$-portion of the total unit resource is reserved to meet the fairness requirements, where $M=\sum_{j\in[K]} m_{j}$. The remaining $(B - M)$-portion of the total resource is then allocated to the arriving agents based on the threshold function $\phi(u_{t}):[0,B-M] \rightarrow [1,\theta_{K}]$, where $u_{t}$ represents the utilization level of the algorithm from the $(1 - M)$-portion of the resource up to the arrival of agent $t$. In the following subsection, we formally present our design of the optimal threshold function $\phi$.

\begin{algorithm}[t]
    \SetKwInput{KwInput}{Input}
    \SetKwInput{KwInitialization}{Initialization}
    \SetKwComment{Com}{\textcolor{gray}{$\triangleright$} }{}
    \KwInput{ $ (m_j, \theta_{j}), \forall j \in [K] $.}
    \KwInitialization{Initial global utilization, $ u_0 = 0$; Initial utilization of class $j$, $u^{j}_{0} = 0, \forall j \in [K] $.}
    
    \While{agent $ t$ arrives}{
    Obtain the valuation and class information of agent $ t $: $ v_t $ and $j_t$\;
    
    \If(\Com*[f]{\textcolor{gray}{Satisfying \GFQ constraint.}}){$u^{j_t}_{t-1} < m_{j_{t}}$}{
      $y_{t} = \min\{r_t, m_{j_{t}} - u^{j_{t}}_{t-1}$\}.
      
      Update $u^{j_{t}}_{t} = u^{j_{t}}_{t-1} + y_{t}$.
      }
    
    \If(\Com*[f]{\textcolor{gray}{Allocating the remaining resource.}}){$v_t \geq \phi(u_{t-1})$}{
        $ x_t = \min\left\{\argmax_{a\in[0, r_t - y_t]}\left\{a v_t-\int_{u_{t-1}}^{u_{t-1}+a} \phi(\eta)d\eta \right\}, B-M -u_{t-1} \right\}.$}
    Update the cumulative allocation: $ u_t = u_{t-1} + x_t $.
    
    Update the allocation amount of agent $t$: 
            $x_{t} = x_{t} + y_{t}$.
            }
    \caption{Global Threshold-based Fair Allocation by Quantity (\Q-\PRB)}
    \label{alg-function-daynamic-threshold-K-class}
\end{algorithm}

\subsection{Optimal Design of the Threshold Function.} \label{sec-GFQ-subsec-optimal-design}
Let $M = \sum_{j\in[K]} m_i$, $C_j =  B - \sum_{i=1}^{j-1} m_i$ and $D_j = \sum_{i=1}^{j-1} m_i\theta_i$. The following Theorem \ref{theorem-LowerBound-multi-Class} shows the design of the optimal threshold function.

\begin{theorem}[\OMcC with \GFQ Guarantee]%[\GFQ Upper-Bound]
\label{theorem-LowerBound-multi-Class}
    For a given \GFQ requirement $\boldsymbol{m}:=\{m_j\}_{j \in [K]}$, the competitive ratio of Algorithm \ref{alg-function-daynamic-threshold-K-class} can be determined in the following cases.
    \begin{itemize}
        \item When $M \leq  \frac{B}{\alpha_0^*}$, Algorithm \ref{alg-function-daynamic-threshold-K-class} is $\alpha_0^*$-competitive if the threshold function $\phi$ is given by
        \begin{align*}
            \phi(u) = \begin{cases}
                1 & u\in [0, \Gamma^0],\\
                \exp\left(\frac{\alpha_0^*\cdot (u+M) - B - \sum_{i=1}^{j-1}m_i \ln\theta_i}{C_j}\right) & u \in [\Gamma^{j-1}, \Gamma^j], \quad\forall j \in [K],
            \end{cases}
        \end{align*}
        where $\alpha_0^* \coloneqq 1 + \ln\theta_K - \sum_{j=1}^{K-1} \frac{m_j}{B} \ln(\frac{\theta_K}{\theta_j})$ and $\Gamma^j = \frac{B}{\alpha_0^*} - M + \frac{C_{j}}{\alpha_0^*}\ln\theta_j+\frac{1}{\alpha_0^*}\sum_{i=1}^{j-1}m_i \ln\theta_i$. 
        \item When $M \in (\frac{\theta_{j^*-1} \cdot C_{j^*} + D_{j^*}}{\alpha_{j^*}},\frac{\theta_{j^*} \cdot C_{j^*} + D_{j^*}}{\alpha_{j^*}}]$ for some $j^*\in[K]$, Algorithm \ref{alg-function-daynamic-threshold-K-class} is $\alpha_{j^{*}}$-competitive if the threshold function $\phi(u)$ follows the design below:
        \begin{align*}
            \phi(u) = \begin{cases}
                v^* \exp\left(\frac{\alpha_{j^{*}}\cdot u - \sum_{i=j^*}^{j-1}m_i \ln\left(\frac{\theta_i}{v^*}\right)}{C_j}\right) & u \in [\Gamma^{j-1}_{j^*}, \Gamma^{j}_{j^*}] ,\quad \forall j \in \{j^*,\cdots, K\},
            \end{cases}
        \end{align*}
        where $\alpha_{j^*}$ is defined as
        \begin{align*}
           \alpha_{j^*} = \frac{D_{j^*}}{M}+ \frac{C_{j^*}}{B-M}W\left(\frac{\theta_{K}(B-M)}{M} \exp\left(-\frac{X}{C_{j^*}}\right)\exp\left(-\frac{D_{j^*}(B-M)}{C_{j^*} \cdot M}\right)\right),
        \end{align*}
        with $X=\sum_{i=j^*}^{K-1} m_i \ln\left(\frac{\theta_K}{\theta_i}\right)$, $v^* = (\alpha_{j^{*}}\cdot M -D_{j^*})/C_{j^*}$, and $\Gamma^{j}_{j^*} = \frac{C_{j}}{\alpha_{j^{*}}}\ln\frac{\theta_{j}}{v^*}+\frac{1}{\alpha_{j^{*}}}\sum_{i=j^*}^{j-1}m_i \ln\frac{\theta_i}{v^*}$.
    \end{itemize}
\end{theorem}

The threshold function \(\phi\) described in Theorem \ref{theorem-LowerBound-multi-Class} consists of at most \(K + 1\) segments. This structure ensures that Algorithm \ref{alg-function-daynamic-threshold-K-class} achieves a competitive ratio influenced by the sum of the minimum allocation requirements \(\boldsymbol{m}\) dictated by the \GFQ constraint. The proof of this theorem is provided in Appendix \ref{appendix-gfq-mulit-upperbound} and an illustration of the threshold function is given in Figure~\ref{fig:GFQ_alphas}.

For the simple case of \(K = 1\) with \(m_1 = 0\), Theorem \ref{theorem-LowerBound-multi-Class} recovers the well-known optimal competitive ratio \(1 + \ln(\theta_1)\) for online conversion and knapsack problems without fairness concerns \cite{zhou2008budget, tan2020mechanism}. When \(m_1 = B\), the competitive ratio approaches \(\theta_1\), representing the worst-case ratio between the offline and online algorithms. This result stems from the fact that, in the online setting, the first \(n_j\) arrivals from each class must be accepted regardless of their valuations. As a result, an adversary can exploit this by presenting low-valuation items early on, penalizing the algorithm.  

For \(K \geq 2\), however, the situation diverges significantly from the knapsack problem. The optimal threshold function is no longer smooth; instead, it becomes a piecewise function due to the additional complexity introduced by the \GFQ constraint. This constraint also limits the flexibility of the offline algorithm, further differentiating the problem from simpler cases. In the following, for general \(K \geq 2\) with \GFQ requirements, we demonstrate in Theorem \ref{theorem-gfq-lowerbound} that the competitive ratio derived in Theorem \ref{theorem-LowerBound-multi-Class} is indeed optimal.

\begin{theorem}[\GFQ Lower Bound]\label{theorem-gfq-lowerbound}
    No algorithm can achieve a better competitive ratio than algorithm \ref{alg-function-daynamic-threshold-K-class} with the design presented in Theorem \ref{theorem-LowerBound-multi-Class} while maintaining the \GFQ requirement.
\end{theorem}
The proof of this theorem can be found in Appendix \ref{appendix-gfq-lowerbound}. To demonstrate the results of Theorem \ref{theorem-LowerBound-multi-Class} in a more explicit way, we also give a case study of $K=2$ in Appendix \ref{appendix-gfq-case-study}.

\section{\OMcC with $ \beta $-\PF Guarantees}\label{section-proportional-fairness}
In this section, we focus on $\beta$-\PF algorithms. Recall that an allocation $\mathbf{x}$ is $\beta$-\PF if it is feasible (i.e., $ \mathbf{x}\in \mathcal{X} := \{\mathbf{x}|x_t\in [0,r_t], \forall t\in[T]$ and $ \sum_{t\in[T]} x_t \leq B $\}) and for any feasible allocation vector $\mathbf{w}$:
\begin{align}
    \frac{1}{K} \sum\nolimits_{j\in [K]}\frac{U_j(\mathbf{w})}{U_j(\mathbf{x})} \leq \beta, \qquad \forall \mathbf{w} \in \mathcal{X},
\end{align}
where $ U_{j}(\mathbf{w})=\sum\nolimits_{t\in [T]} v_t w_t \cdot  \boldsymbol{1}_{ \{j_t = j\}}, \forall j\in [K] $. Here we consider that the fraction $x/y$ for non-negative $x$ and $y$ is equal to 0 when $x=y=0$, while $x/y=+\infty$ when $y=0$ but $x>0$. 

Intuitively, we cannot adopt a similar approach to Algorithm \ref{alg-function-daynamic-threshold-K-class} to provide a bounded $\beta$-\PF guarantee, as an adversary could send all requests from a single class, depleting the resource. Once the resource is fully allocated, requests from other classes may begin to arrive. To mitigate this, a portion of the resource must be reserved for each class. However, unlike in \GFQ, where these reservations are externally enforced, the reserved portions here must be optimally designed for each class.

\subsection{Fair Allocation with Optimal $\beta$-\PF Guarantees}
We present \U-\PRB, a group-level, threshold-based online algorithm that integrates a utility-based fairness metric for the \OMcC problem. The detailed procedure is outlined in Algorithm \ref{alg-proportional-fairness}. In this algorithm, when each agent arrives, their class is first identified. Then, based on the specific threshold function for that class, the allocation decision $ x_t $ is made by maximizing the pseudo-utility. Specifically, we develop $K$ threshold functions, denoted as $\phi_j(u^j_t): [0, b_j] \to [1, \theta_j]$ for each group $j \in [K]$, where $u^j_t$ represents the utilization level of the algorithm for group $j$, drawn from the $b_j$ portion of the resource up to the arrival of agent $t$. Note that due to resource constraints, the total allocation must satisfy $\sum_{j \in [K]} b_j \leq B$.  Theorem \ref{theorem-proportional-fairness-beta} below gives an explicit design of the threshold function $ \phi_j $ for each class $ j\in [K]$.

\begin{theorem}[$\beta$-\PF Guarantee]
\label{theorem-proportional-fairness-beta}
    For each class $j\in [K]$, if the threshold function $ \phi_j $ is given by
    \begin{align*}
        \phi_j(u) = \begin{cases}
            1& u\in [0, \frac{B}{\sum_{i\in[K]} \alpha_i}],\\
            \exp{({\frac{K \beta u}{B} - 1})} & u\in [\frac{B}{\sum_{i\in[K]} \alpha_i}, \frac{B\cdot \alpha_j}{\sum_{i\in[K]} \alpha_i}],
        \end{cases}
    \end{align*}
    where $\alpha_j = 1 + \ln{\theta_j}$, then Algorithm \ref{alg-proportional-fairness} is $\beta$-\PF with $\beta = \frac{1}{K} \sum_{j\in[K]} \alpha_j$.
\end{theorem}

The proof of the above theorem is given in Appendix \ref{appendix-proportional-fairness-beta}. Here we emphasize that the reserved portion of the resource $b_j$ for each class $j$ is equal to  $\frac{B\cdot \alpha_j}{\sum_{i\in[K]} \alpha_i}$, which implies the proportionality of this algorithm considering the uncertainties of the arrivals from each group. Additionally, it is important to note that each class has its own individual threshold function, unlike Algorithm \ref{alg-function-daynamic-threshold-K-class}, which uses a single global threshold function. This difference arises because, in \GFQ, the fairness guarantee is enforced by the input of the problem, whereas in $\beta$-\PF, fairness is embedded in the design of the threshold functions. Moreover, we show that this approach achieves the best possible $\beta$-\PF in the \OMcC setting.  The following theorem formalizes this result. 

\begin{theorem}[$\beta$-\PF Lower-Bound]\label{theorem-proportional-fairness-lower-bound}
    There exists no $ \beta $-\PF online algorithm for \OMcC with $ \beta < \frac{1}{K} \sum_{j\in[K]} \alpha_j$, where $\alpha_j = 1 + \ln{\theta_j}$.
\end{theorem}

Theorem \ref{theorem-proportional-fairness-lower-bound} shows that Algorithm \ref{alg-proportional-fairness} is indeed optimal by achieving $ (\frac{1}{K} \sum_{j\in[K]} \alpha_j)$-\PF. The proof of the above theorem is provided in Appendix \ref{appendix-proportional-fairness-lower-bound-fairness}.

\begin{algorithm}[t]
    \SetKwInput{KwInput}{Input}
    \SetKwInput{KwInitialization}{Initialization}
    \SetKwComment{Com}{\textcolor{gray}{$\triangleright$} }{}
    \KwInput{$ \theta_{j}, \forall j \in [K] $.}
    \KwInitialization{Initial utilization of class $j$, $ u^{j}_{0} = 0, \forall j \in [K] $.}
    
    \While{agent $ t$ arrives}{
    Obtain the value and class information of agent $ t $: $ v_t $ and $j_t$ \;
    
    \If(\Com*[f]{\textcolor{gray}{Local threshold-based allocation.}}){$v_t \geq \phi_{j_t}(u^{j_t}_{t-1})$}{
      $
        x_t = \argmax_{a\in[0,r_t]} \left\{a v_t-\int_{u^{j_t}_{t-1}}^{u^{j_t}_{t-1}+a} \phi_{j_t}(\eta)d\eta \right\}.
        $
      }
    
    Update the cumulative allocation: $ u^{j_t}_t = u^{j_t}_{t-1} + x_t $. 
            }
    \caption{Local Threshold-based Fair Allocation by Utility (\U-\PRB)}
    \label{alg-proportional-fairness}
\end{algorithm}

\subsection{Pareto-Optimal Efficiency-Fairness Trade-off: $ \alpha $-Competitiveness vs $\beta$-\PF}
It is sensible that fairness and competitiveness are inversely related. The reason for this is that maintaining fairness often requires allocating resources to classes with lower valuations, which leads to a worse competitive ratio compared to allocations that do not prioritize fairness.  To balance efficiency and fairness, we introduce a new class of Set-Aside Multi-Threshold-based (\SAMT) algorithms in Algorithm \ref{alg-proportional-fairness-tradeoff-local-globol}, which incorporates two types of threshold functions. 

At a high level, Algorithm \ref{alg-proportional-fairness-tradeoff-local-globol} works as follows. Upon arrival of each agent, based on its class and valuation information, two allocation decisions are made. The first is based on the agent’s group-specific threshold function, which ensures group fairness. The second is based on a global threshold function, aimed at optimizing individual welfare. This algorithm is named \textit{Set-Aside} in the sense that it essentially sets aside some of the resource for agents to compete for, while reserving a certain portion for each class to maintain fairness across different classes. It is also called \textit{Multi-Threshold} because the allocation is made based on two different types of threshold functions. Specifically, we design $K+1$ threshold functions: one local threshold function for each class, denoted by $\phi_j(u^j_t): [0, b_j] \to [1, \theta_j]$ for each group $j \in [K]$, and one global threshold function, denoted by $\phi^G(u_t): [0, B-\sum_{j\in [K]}b_j] \to [1, \theta_K]$, shared across all classes. Here, we let $b_j$ denote the reserved resource for class $j$ and $B-\sum_{j\in [K]}b_j$ is the set-aside resource. In Theorem \ref{theorem-proportional-fairness-trade-off-design} below, we show that with a well-designed set of threshold functions,  Algorithm \ref{alg-proportional-fairness-tradeoff-local-globol} can smoothly balance $ \alpha $-competitiveness and $\beta$-\PF  for any $ \beta \geq \frac{1}{K} \sum_{j\in[K]} \alpha_j$.

\begin{theorem}[$ \alpha $-Competitiveness vs $\beta$-\PF]\label{theorem-proportional-fairness-trade-off-design}
    For any given $\beta \ge \frac{1}{K} \sum_{j\in[K]} \alpha_j$, Algorithm \ref{alg-proportional-fairness-tradeoff-local-globol} is $\beta$-\PF and $ \alpha_K(1 - \frac{\sum\nolimits_{j\in[K-1]} \alpha_j}{K\beta})^{-1}$-competitive for \OMcC if the threshold functions are designed as follows:
    \begin{align*}
        &\phi_j(u) = \begin{cases}
            1& u\in [0, \frac{b_j}{\alpha_j}],\\
            \exp{({\frac{K \beta u}{B} - 1})} & u\in [\frac{b_j}{\alpha_j}, b_j],
        \end{cases}\quad \forall j\in[K],\\
        &\phi^G(u) = \begin{cases}
            1& u\in \Big[0, \frac{B-\sum_{j\in [K]}b_j}{\alpha_K}\Big],\\
            \exp\bigg({ \frac{\alpha_K}{1 - \frac{\sum\nolimits_{j\in[K]} \alpha_j}{K\cdot \beta}} \cdot \frac{u}{B} - 1}\bigg) & u\in \Big[\frac{B-\sum_{j\in [K]}b_j}{\alpha_K}, B-\sum_{j\in [K]}b_j\Big],
        \end{cases}
    \end{align*}
    where $b_j = \frac{B\cdot \alpha_j}{K\beta}$ and $\alpha_j = 1 + \ln\theta_j$.
\end{theorem}

The proof of this theorem is provided in Appendix \ref{appendix-theorem-proportional-fairness-trade-off-design}. By utilizing the threshold functions defined in Theorem \ref{theorem-proportional-fairness-trade-off-design}, Algorithm \ref{alg-proportional-fairness-tradeoff-local-globol} reserves a budget of $\frac{B \cdot \alpha_j}{K \beta}$ for each class of agents to achieve $\beta$-\PF based on the class-specific threshold function $\phi_j$. Furthermore, the global threshold function $\phi^G$ allocates the remaining fraction of the resource, given by $B - \sum_{i \in [K]} \frac{B \cdot \alpha_i}{K \beta}$, to ensure high levels of efficiency. Notably, as $\beta$ approaches infinity, the resource allocation is determined entirely by the global threshold function, aligning the optimal algorithm without fairness guarantee with the optimal competitive ratio $\alpha_K$ \cite{zhou2008budget}. Conversely, achieving optimal fairness requires an algorithm that omits the global threshold function. This aligns with the design outlined in Theorem \ref{theorem-proportional-fairness-beta}, which achieves $\left(\frac{1}{K} \sum\nolimits_{i \in [K]} \alpha_i\right)$-\PF.

\begin{algorithm}[t]
    \SetKwInput{KwInput}{Input}
    \SetKwInput{KwInitialization}{Initialization}
    \SetKwComment{Com}{\textcolor{gray}{$\triangleright$} }{}
    \KwInput{ $\beta$; $ \theta_{j}, \forall j \in [K] $.}
    \KwInitialization{Initial global utilization, $ u_0 = 0$; Initial utilization of class $j$, $u^{j}_{0} = 0, \forall j \in [K] $.}
    
    \While{agent $ t$ arrives}{
    Obtain the value and class information of agent $ t $: $ v_t $ and $j_t$ \;
    
    \If(\Com*[f]{\textcolor{gray}{Set-Aside threshold-based allocation.}}){$v_t \geq \phi_{j_t}(u^{j_t}_{t-1})$}{
      $x_t^1 = \argmax_{a\in[0,r_t]} \left\{a\cdot v_t-\int_{u^{j_t}_{t-1}}^{u^{j_t}_{t-1}+a} \phi_{j_t}(\eta)d\eta \right\}$.
      }
    \If(\Com*[f]{\textcolor{gray}{Global threshold-based allocation.}}){$v_t \geq \phi^G(u_{t-1})$}{
      $x_t^2 = \argmax_{a \in [0, r_t-x_t^1]} \left\{a\cdot v_t-\int_{u_{t-1}}^{u_{t-1}+a} \phi^G(\eta)d\eta \right\}$.
      }
    
    Update the local cumulative utilization: $ u^{j_t}_t = u^{j_t}_{t-1} + x^1_t $. 
    
    Update the global cumulative utilization: $u_t = u_{t-1} + x^2_t $.
            }
    \caption{Set-Aside Multi-Threshold-based Fair-Efficient Allocation (\SAMT) }
    \label{alg-proportional-fairness-tradeoff-local-globol}
\end{algorithm}

The concept of reserving a portion of resources to ensure fairness is well-established in the literature \cite{banerjee2022proportionally, banerjee2022online}. However, in most cases, the allocation of this reserved portion is done greedily (e.g., reserve half of the total resource for greedy allocation), often leading to suboptimal outcomes within those frameworks. In contrast, we allocate the reserved portion using a threshold function, which offers a significant advantage in worst-case analysis. Moving forward, we prove that the threshold function design in Theorem \ref{theorem-proportional-fairness-trade-off-design} achieves the Pareto-optimal competitiveness-fairness trade-off.

\begin{theorem}[Pareto Optimality]
\label{theorem-proportional-fairness-alpha-lowerbound}
    No \(\beta\)-\PF algorithm can achieve a smaller competitive ratio than $\alpha$, where $\alpha = \alpha_K \cdot \left(1 - \frac{\sum\nolimits_{j\in[K-1]} \alpha_j}{K\beta}\right)^{-1}$. Thus, Algorithm \ref{alg-proportional-fairness-tradeoff-local-globol} attains the Pareto-optimal trade-off between \PF and competitiveness.
\end{theorem}
This trade-off indicates that in the fairest algorithm, the competitive ratio is at most \( \sum_{j \in [K]} \alpha_j \). In contrast, the most competitive algorithm completely disregards fairness, resulting in \( \beta \to \infty \). The proof of this theorem can be found in Appendix \ref{appendix-proportional-fairness-alpha-lowerbound}.

\section{\OMcC with ($\gamma,\beta$)-Fairness Guarantees}\label{section-gamma-beta-fairness}
In this section, we focus on developing $(\gamma, \beta)$-fair online algorithms for \OMcC. By Definition \ref{def-gamma-beta-fairness}, any online algorithm obtaining $(\gamma, \beta)$-fair guarantees needs to produce a feasible allocation $\mathbf{x} \in \mathcal{X}$ such that for every other feasible allocation $\mathbf{w}$:
\begin{align}
    \sum_{j=1}^K f\left(U_j(\mathbf{w})\right) \leq \sum_{j=1}^K f\left(\beta\cdot U_j(\mathbf{x})\right), \qquad \forall \mathbf{w}\in \mathcal{X},
\end{align}
where $f$ is defined in Eq. \eqref{eq:f-defintion} and $ \ U_{j}(\mathbf{x}) =\sum\nolimits_{t\in [T]} v_t x_t \cdot  \boldsymbol{1}_{ \{j_t = j\}}, \forall j\in [K] $. 

Based on Definition \ref{def-gamma-beta-fairness}, the value of $\gamma$ corresponds to the different fairness metrics. To demonstrate this better, we discuss some special cases of this parameter below. 

\paragraph{Competitive Ratio ($\gamma = 0$)} \quad Based on the definition of $(\gamma, \beta)$-fairness, when $\gamma = 0$, the fairness metric simplifies to $\frac{\sum_{j \in [K]} U_j(\mathbf{w})}{\sum_{j \in [K]} U_j(\mathbf{x})} \leq \beta$.  Following this definition, it is easy to see that $(0,\beta)$-fairness reduces to $\beta$-competitiveness without group fairness guarantee, meaning that the optimal algorithm may allocate the entire resource to the agents from the class with the highest valuations. In other words, when $\gamma = 0$, the only criterion that matters is the efficiency of the algorithm, and thus $\gamma = 0$ can be interpreted as a scenario where fairness is not considered at all.

\paragraph{Nash Social Welfare ($\gamma = 1$)} \quad In this case, $(\gamma, \beta)$-fairness simplifies to $\log\left(\prod_{j \in [K]} U_j(\mathbf{w})\right) \leq \log\left(\prod_{j \in [K]} \beta \cdot U_j(\mathbf{x})\right)$, or equivalently, $\frac{\left(\prod_{j \in [K]} U_j(\mathbf{w})\right)^{1/K}}{\left(\prod_{j \in [K]} U_j(\mathbf{x})\right)^{1/K}} \leq \beta$. It is equivalent to a $\beta$-approximation of the \NSW. In the offline setting, it is known that maximizing the \NSW strikes a good balance between efficiency and fairness. In particular, for the allocation of divisible goods, maximizing the \NSW implies Pareto optimality, proportional fairness, and envy-freeness \cite{vazirani2007combinatorial}. This means that no individual prefers the allocation of another, thus achieving both equitable distribution and economic efficiency.  

\paragraph{Max-Min Fairness ($\gamma \to \infty$)} \quad As $\gamma\rightarrow \infty $, $(\gamma, \beta)$-fairness reduces to $\max \min_{j \in [K]} \{ U_j(\mathbf{w}) \} \leq \beta \cdot \max \min_{j \in [K]} \{ U_j(\mathbf{x}) \}$, which is a $\beta$-approximation of the standard max-min fairness. In the offline setting, the objective is to maximize the minimum utility across all groups, ensures that the least well-off group receives the maximum possible allocation, while still ensuring that no other group's allocation is reduced. In the \OMcC problem, \MMF is the same as equalizing the utilities among all the agents. Therefore, in the worst-case allocation $\mathbf{w}$ is such that $U_1(\mathbf{w}) = \dots = U_K(\mathbf{w})$. This approach ensures that the classes with lower valuations receive more resources, while those with higher valuations receive less. Since all utilities are equalized, this can be interpreted as the fairest possible allocation.

\subsection{Fair Allocation with  Tight $(\gamma, \beta)$-Fairness Guarantees}
In this section, we demonstrate that Algorithm \ref{alg-proportional-fairness} can achieve tight ($\gamma, \beta$)-fairness guarantees if the threshold functions for each class are appropriately redesigned. These threshold functions are heavily dependent on the fairness function $f(\cdot)$, as defined in Eq. \eqref{eq:f-defintion}, and generally lack a closed-form expression. In the following, we propose a novel representative function-based approach for analytically deriving threshold functions that ensure tight fairness guarantees.

The key to our approach is the design of a set of non-decreasing utilization functions $\psi_j(\cdot): [1, \theta_j] \to [0, b_j]$ for each class $j \in [K]$, which are closely related to the threshold functions $\phi_j(\cdot): [0, b_j] \to [1, \theta_j]$. Specifically, $\psi_j(v_j)$ represents the utilization level when all agents in class $j$ have valuation $v_j$, and satisfies $v_j = \phi_j(\psi_j(v_j))$. This establishes that $\psi_j$ and $\phi_j$ can be regarded as inverse functions of each other. This relationship implies that the allocation can equivalently be viewed through the lens of the utilization function $\psi_j(\cdot)$ or the threshold function $\phi_j(\cdot)$, as both determine the allocation point in the same manner. Our design in Theorem \ref{theorem-gamma-beta-fairness} leverages this connection by constructing the utilization function $\psi_j(\cdot)$ for hard instances and showing that its inverse can yield the appropriate threshold function $\phi_j(\cdot)$ for Algorithm~\ref{alg-proportional-fairness}. In particular, the threshold function exhibits three different forms depending on the range of $\gamma$: $\gamma \in (0, 1)$, $\gamma = 1$, and $\gamma \in (1, \infty)$. In the following theorem, the notation $[j^-] = [K] \setminus \{j\}$ is used.
\begin{theorem}[($\gamma,\beta$)-Fairness Guarantee]\label{theorem-gamma-beta-fairness}
    For any given $\gamma \geq 0$, Algorithm \ref{alg-proportional-fairness} is ($\gamma,\beta$)-fair if the threshold function $ \phi_j $ is set as the inverse of the following utilization function $ \psi_j $:
    \begin{align*}
        \psi_j(v) = \frac{B }{\beta_j}F_j(v;\gamma), \qquad 1\leq v \leq \theta_j, \forall j\in[K],
    \end{align*}
    where $ \{F_j\}_{\forall j} $ are given as follows:
    \begin{itemize}
        \item if $\gamma\in(1,\infty)$, $ F_j(v;\gamma) $ is given by
        \begin{align*}
            & F_j(v;\gamma) = \frac{1}{\sum_{i\in[j^-]} \left(
            \frac{v}{\theta_i}\right)^{\frac{\gamma-1}{\gamma}} + 1} + \frac{\gamma}{\gamma-1} \ln \left(\frac{v^{\frac{\gamma-1}{\gamma}}\left(\sum_{i\in[j^-]} \left(
            \frac{1}{\theta_i}\right)^{\frac{\gamma-1}{\gamma}} + 1\right)}{\sum_{i\in[j^-]} \left(
            \frac{v}{\theta_i}\right)^{\frac{\gamma-1}{\gamma}} + 1}\right), \quad 1\leq v \leq \theta_j;
        \end{align*}
        \item if $\gamma\in(0,1)$,  $ F_j(v;\gamma) $ is given by
        \begin{align*}
           & F_j(v;\gamma) = \frac{1}{(K-1)\cdot v^{\frac{\gamma-1}{\gamma}} + 1} + \frac{\gamma}{\gamma-1} \ln \left(\frac{K \cdot v^{\frac{\gamma-1}{\gamma}}}{(K-1)\cdot v^{\frac{\gamma-1}{\gamma}} + 1}\right), \quad 1\leq v \leq \theta_j;
        \end{align*}
        \item if $\gamma=1$, $ F_j(v;\gamma)$ is given by
            \begin{align*}
            F_j(v;\gamma) = \frac{1+\ln{v}}{K}, \quad 1\leq v \leq \theta_j;
        \end{align*}
    \end{itemize}
    and the $ \beta_j $'s and $\beta$ are obtained by solving the following optimization problem:
    \begin{itemize}
        \item if $\gamma \neq 1$:
        \begin{align}\label{eq:minimax-1}
            \beta = \min\limits_{\{\beta_j\geq 1\}_{\forall j}}\max\limits_{v_j \in \{1, \theta_j\}, \forall j} \left\{\left(\frac{\sum\nolimits_{j\in[K]} \beta_j^{\gamma-1}\cdot v_j^\frac{{1-\gamma}}{\gamma}}{\sum\nolimits_{j\in[K]} v_j^\frac{{1-\gamma}}{\gamma}}\right)^{\frac{1}{\gamma-1}}\right\}, \quad\ {\rm s.t.}\ \sum\nolimits_{j\in[K]} \psi_j(\theta_j) \leq B;
        \end{align}
        \item if $\gamma = 1$:
        \begin{align}\label{eq:minimax-2}
            \beta = \min\limits_{\{\beta_j\geq 1\}_{\forall j}}\left\{\prod\nolimits_{j\in[K]} \beta_j^{1/K}\right\}\, \quad \ {\rm s.t.}\ \sum\nolimits_{j\in[K]} \psi_j(\theta_j) \leq B.
        \end{align}
    \end{itemize}
\end{theorem}

The proof of Theorem \ref{theorem-gamma-beta-fairness} is provided in Appendix \ref{appendix-theorem-gamma-beta-fairness}. At a high level, Theorem \ref{theorem-gamma-beta-fairness} seeks to satisfy the inequality $U_j(\mathbf{w}) \leq \beta_j \cdot U_j(\mathbf{x})$ for each $ j \in [K] $. This leads to the result $\frac{1}{1-\gamma} \sum_{j\in[K]} U_j^{1-\gamma}(\mathbf{w}) \leq \frac{1}{1-\gamma} \sum_{j\in[K]} \beta_j^{1-\gamma} \cdot U_j^{1-\gamma}(\mathbf{x})$. The main idea of Eqs. \eqref{eq:minimax-1} and \eqref{eq:minimax-2} is to appropriately choose $\beta_j$ for different classes such that $\frac{1}{1-\gamma} \sum_{j\in[K]} \beta_j^{1-\gamma} \cdot U_j^{1-\gamma}(\mathbf{x}) \leq \frac{1}{1-\gamma} \sum_{j\in[K]} \beta^{1-\gamma} \cdot U_j^{1-\gamma}(\mathbf{x})$, which allows us to derive the fairness guarantee $ \beta $.  Note that the solution to the minimax problem in Eq. \eqref{eq:minimax-1} lies either at the boundary points, where \( v_j \in \{1, \theta_j\} \) for all \( j \in [K] \). In this case, the minimax problem is convex in the \(\beta_j\)'s, making it straightforward to solve. Alternatively, the solution occurs when \(\beta_j = \beta_i\) for all \( i, j \in [K]\). The core complexity of Theorem \ref{theorem-gamma-beta-fairness} lies in solving Eq. \eqref{eq:minimax-1} and \eqref{eq:minimax-2}, which are critical to the theorem's results. To provide a more explicit demonstration of the solution to these equations, we include an example for the simplest case, \(K=2\), covering all possible values of \(\gamma\), in Appendix \ref{gamma-beta-fairness-case-study}.

\subsection{Lower Bound of $ \beta $ and Order-Optimality of Theorem \ref{theorem-gamma-beta-fairness}}
In this section, we first prove a lower bound for the ($\gamma,\beta$)-fairness guarantee of any online algorithm and then discuss the order-optimality of the fairness guarantee of Theorem \ref{theorem-gamma-beta-fairness}.

\begin{theorem}[($\gamma,\beta$)-Fairness Lower-Bound]\label{theorem-gamma-beta-lowerbound}
For any $ \gamma \geq 0$ and $ \epsilon > 0 $, there exists no  ($\gamma,\beta^{*}_\gamma - \epsilon$)-fair algorithm for \OMcC,  where $ \beta^{*}_\gamma $ is the optimal objective value to the following optimization problem:
    \begin{align}
        &\beta^{*}_\gamma = \min_{\beta, \{\lambda_j,\rho_{j} \}_{j \in [K]}} \beta  \nonumber\\
        \label{eq:lower-bound-GBF-constraint}
        & \textrm{s.t.}\quad  \beta \ge B^{-1}\left(\sum_{j \in [K]} \frac{1}{\theta_{j}}  \cdot g_j(\theta_{j})^{\frac{1}{1-\gamma}} + \int_{\eta=1}^{\theta_{j}}  \frac{1}{\eta^{2}} \cdot g_j(\eta)^{\frac{1}{1-\gamma}} \cdot d\eta +  \lambda_j\right), 
    \end{align}
    where $ g_j $ is given by
    \begin{align*}
        &g_{j}(v) = B^{1-\gamma}\cdot \left(j-1 + v^{\frac{\gamma-1}{\gamma}} + \sum_{l=j+1}^{k} \theta_{l}^{\frac{\gamma-1}{\gamma}}\right)^\gamma - \sum_{i=1}^{j-1} \rho_{j}^{1-\gamma} - \sum_{i=j+1}^{k} V_{i}^{1-\gamma},
    \end{align*}
    and $V_{j}$ is defined as follows:
    \begin{align*}
        &V_{j} = 
       \left(B^{1-\gamma}\cdot \left(j-1 + \theta_j^{\frac{\gamma-1}{\gamma}} + \sum_{l=j+1}^{k} \theta_{l}^{\frac{\gamma-1}{\gamma}}\right)^\gamma - \sum_{i=1}^{j-1} \rho_{j}^{1-\gamma} - \sum_{i=j+1}^{k} V_{i}^{1-\gamma} \right)^{\frac{1} {1-\gamma}} + \lambda_{j} \cdot \theta_{j}.
    \end{align*}
\end{theorem}
The proof of this theorem is provided in Appendix \ref{appendix-gamma-beta-lowerbound}. At a high level, for the optimization problem in Eq. \eqref{eq:lower-bound-GBF-constraint}, the term $V_j $ depends solely on the decision variables $\lambda_j$ and $\rho_j$. Consequently, the function $g_j(v)$ also depends exclusively on these decision variables. Hence, the constraint defined in Eq. \eqref{eq:lower-bound-GBF-constraint} is determined entirely by  $\rho_j$ and $\lambda_j$. Therefore, by identifying the optimal values for these decision variables, we can derive the optimal value of $\beta_\gamma^*$. Figure \ref{figure-lower-bound-MTF-BTF} illustrates the numerically computed lower bound $ \beta_{\gamma}^* $ for $K = 2$. This figure highlights that even as $\gamma$ shifts, \U-\PRB with ($\gamma,\beta$)-fairness guarantee maintains fairness close to the theoretical lower bound, providing strong evidence of the algorithm's near-optimal performance.

In general, the optimization problem in Eq. \eqref{eq:lower-bound-GBF-constraint} cannot be solved analytically, which complicates the analysis of the asymptotic behavior of the optimal fairness guarantee. However, we prove that for the simple case of $K=2$, Algorithm \ref{alg-proportional-fairness} is asymptotically optimal.

\begin{corollary}[Order-Optimality of \U-\PRB with ($\gamma,\beta$)-fairness guarantee]\label{theorem-gamma-beta-order-optimal} 
    Algorithm \ref{alg-proportional-fairness}, using the threshold functions from Theorem \ref{theorem-gamma-beta-fairness}, achieves order-optimal performance for all $\gamma$ when $K=2$. Specifically: 
    \begin{itemize} 
        \item For $\gamma < 1$, any $(\gamma,\beta)$-fair online algorithm must have $ \beta = \Omega(\alpha_2)$, and the fairness guarantee of Algorithm \ref{alg-proportional-fairness} is $ \mathcal{O}(\alpha_2)$. 
        \item For $\gamma \approx 1$, any $(\gamma,\beta)$-fair online algorithm must have $\beta = \Omega(\sqrt{\alpha_2})$, and the fairness guarantee of Algorithm \ref{alg-proportional-fairness} is $ \mathcal{O}(\sqrt{\alpha_2})$. 
        \item For $\gamma > 1$, any $(\gamma,\beta)$-fair online algorithm must have  $\beta = \Omega(\alpha_1)$, and the fairness guarantee of Algorithm \ref{alg-proportional-fairness} is $ \mathcal{O}(\alpha_1)$.
    \end{itemize} 
    Recall that we have $\alpha_j = 1+\ln\theta_j$ for all $ j \in [K] $.
\end{corollary}

The intuition behind these results is as follows. For $\gamma < 1$, the optimal solution prioritizes the highest valuation class, leading to fairness of ${\Omega}(\alpha_2)$. For $\gamma > 1$, the optimal solution favors lower-value classes, achieving the fairness of ${\Omega}(\alpha_1)$. When $\gamma \approx 1$, resources are more evenly distributed, with a less pronounced increase in fairness achieving the fairness of $\Omega(\sqrt{\alpha_2})$. The proof of Corollary \ref{theorem-gamma-beta-order-optimal} is given in Appendix \ref{appendix-theorem-gamma-beta-order-optimal}.

\begin{figure}[t]
    \centering
    \begin{minipage}{0.46\textwidth}
        \centering
        \includegraphics[trim=0.5cm 0cm 0.5cm 0.5cm,clip,width=\textwidth]{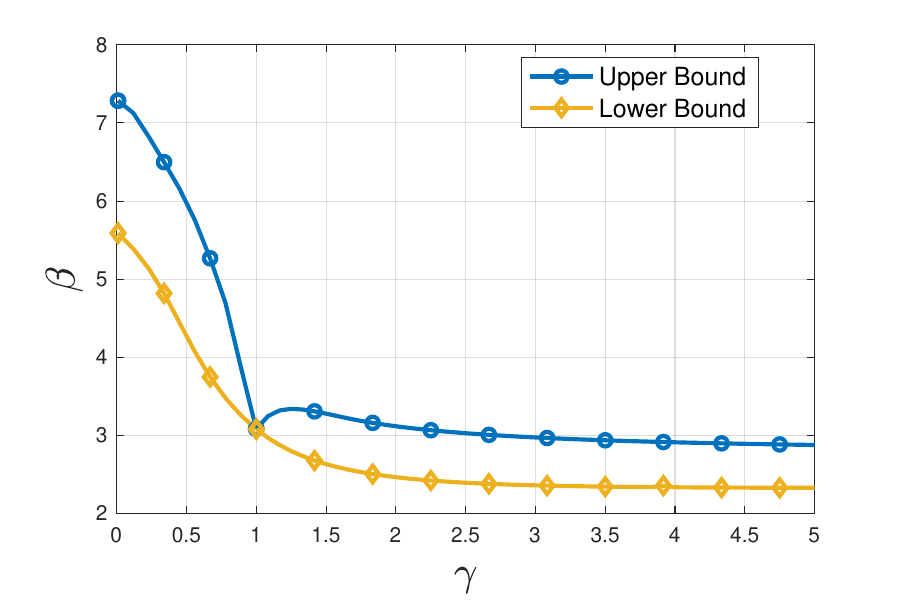}
        \caption{The fairness guarantee of \U-\PRB with ($\gamma,\beta$)-fairness vs lower bound; $\theta_1 = 2$, $\theta_2 = 100$.}
        \label{figure-lower-bound-MTF-BTF}
    \end{minipage}
    \hspace{0.5cm}
    \begin{minipage}{0.46\textwidth}
        \centering
        \includegraphics[trim=0.5cm 0cm 0.5cm 0.5cm,clip,width=\textwidth]{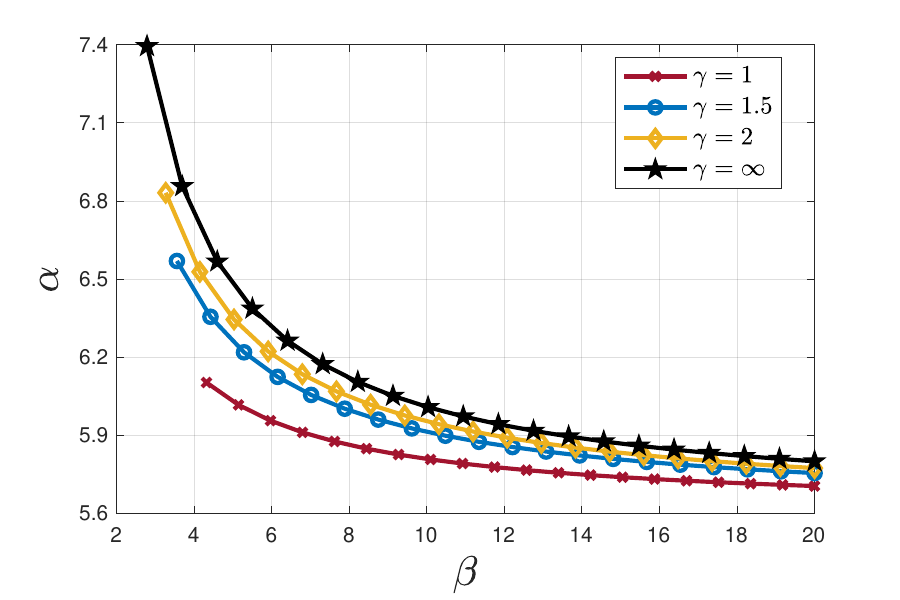}
        \caption{The trade-off between fairness and competitiveness for different values of $\gamma$; $\theta_1 = 2$, $\theta_2 = 100$.}
        \label{figure-trade-off-gamma-beta-alpha}
    \end{minipage}
\end{figure}

\subsection{Efficiency-Fairness Trade-off: $ \alpha $-Competitiveness vs ($\gamma, \beta$)-Fairness}
We are now ready to introduce a family of algorithms that demonstrate the trade-off between fairness and efficiency for  \(\gamma \geq 1\).\footnote{We focus on this domain because when \(\gamma\) is less than one, the definition of \((\gamma, \beta)\)-fairness inherently emphasizes efficiency—specifically, allocating more resources to classes with higher valuations. As a result, discussing the trade-off between fairness and efficiency in this region becomes less meaningful. } We use the same structure of \SAMT in Algorithm \ref{alg-proportional-fairness-tradeoff-local-globol} and show in Theorem \ref{theorem-gamma-beta-fairness-trade-off-design} below that it is possible to redesign the local threshold function for each class to guarantee ($\gamma,\beta$)-fairness and the global threshold function to ensure $\alpha$-competitiveness. 

\begin{theorem}[$\alpha$-Competitiveness vs ($\gamma,\beta$)-Fairness]\label{theorem-gamma-beta-fairness-trade-off-design}
    For a given $\beta $ and $\gamma \ge 1$, Algorithm \ref{alg-proportional-fairness-tradeoff-local-globol} is $\alpha$-competitive and $(\gamma,\beta)$-fair for the \OMcC problem  if the threshold functions $\{\phi_{j}\}_{j \in [k]}$ and $\phi^{G}$ are respectively designed as the inverse of the utilization functions $\{\psi_{j}\}_{j \in [k]}$ and $\psi^{G}$ given as follows: 
    \begin{align*}
        &\psi_j(v) = \frac{B}{\beta_j}F_j(v;\gamma), \qquad 1\leq v \leq \theta_j, \forall j\in[K],\\
        &\psi^G(v) = B\cdot \frac{1 - \sum_{j\in[K-1]} \frac{F_j(\theta_i;\gamma) - F_j(1;\gamma)}{\beta_j}}{1+\ln\theta_K}(1 + \ln v) - \frac{B}{\beta_K}F_K(v;\gamma) - \sum\nolimits_{j\in[K-1]} \frac{B}{\beta_j}F_j(1;\gamma),
    \end{align*}
    where $ \{F_j\}_{\forall j} $ is designed according to Theorem \ref{theorem-gamma-beta-fairness}; $\alpha$ and $ \{\beta_j\}_{\forall j}$ are the solutions to the following optimization problem
    \begin{subequations}\label{eq:alpha-beta-trade-off}
        \begin{align}
        \alpha = & \max_{\{\beta_j\geq 1\}_{\forall j}} \left\{\frac{1+\ln\theta_K}{1 - \sum_{j\in[K-1]} \frac{F_j(\theta_j;\gamma) - F_j(1;\gamma)}{\beta_j}}\right\}, \\
         & \ {\rm s.t.}\ \sum\nolimits_{j\in[K]} \psi_j(\theta_j) \leq B, \qquad   \beta = \begin{cases}
            \max\limits_{v_j \in \{1,\theta_{j}\}, \forall j} \left\{\left(\frac{\sum\nolimits_{j\in[K]} \beta_j^{\gamma-1}\cdot v_j^\frac{{1-\gamma}}{\gamma}}{\sum\nolimits_{j\in[K]} v_j^\frac{{1-\gamma}}{\gamma}}\right)^{\frac{1}{\gamma-1}}\right\}, & \text{if } \gamma >1,  \\
            \\
            \prod\nolimits_{j\in[K]} \beta_j^{1/K}, & \text{if }  \gamma =1.
        \end{cases}
        \end{align}
    \end{subequations}
\end{theorem}

To determine the competitive ratio $ \alpha $ in the theorem above, we need to solve the optimization problem in Eq. \eqref{eq:alpha-beta-trade-off}. This can be done efficiently, as both the objective function and the constraint are convex in the $\beta_j$ variables. The proof of Theorem \ref{theorem-gamma-beta-fairness-trade-off-design} can be found in Appendix \ref{appendix-theorem-gamma-beta-fairness-trade-off-design}. It is worth noting the optimization problem in Eq. \eqref{eq:alpha-beta-trade-off} always admits a solution for any given $\beta$ that is greater than the minimum fairness guarantee derived in Theorem \ref{theorem-gamma-beta-fairness}. For more details, please refer to the proof of Theorem \ref{theorem-gamma-beta-fairness} and Proposition \ref{proposition-modified-gamma-beta-thresholds} in the appendix.

Figure \ref{figure-trade-off-gamma-beta-alpha} presents the trade-off between fairness and competitiveness  described in Theorem \ref{theorem-gamma-beta-fairness-trade-off-design} for $K=2$, evaluated across different $\gamma$ values. As expected, the results demonstrate that for each $\gamma$, an increase in fairness leads to a corresponding decrease in competitiveness.

\section{Numerical Results}\label{section-numerical-result}
In this section, we model the utility-based \TTL caching protocol using \OMcC framework (see the illustrative example in Section~\ref{sec-problem-formulation-sub-an-illustrative-example}) and perform numerical experiments based on the Wikipedia Clickstream dataset \cite{wikimedia_analytics}. We evaluate the performance of \Q-\PRB using the \GFQ metric and evaluate \U-\PRB using the $\beta$-\PF and ($\gamma,\beta$)-fairness metrics. The results are compared against both the offline optimal solution and the optimal online algorithm. We also empirically evaluate the performance of \SAMT to investigate the trade-off between fairness and competitiveness. 

\subsection*{Experiment Setup.} 
Based on the \TTL protocol described in \cite{dehghan2019utility}, the objective is to maximize the cumulative utility over all files, represented as $\sum_{t \in [T]} \mathcal{U}_t(x_t)$, by optimizing the hitting probability $x_t \in [0,1]$ for the file arriving at time $t$. According to \cite{dehghan2019utility}, when $x_t = 0$, the file is not cached, while $x_t = 1$ corresponds to caching the file indefinitely. An average buffer occupancy constraint, $\sum_{t \in [T]} x_t \leq B$, where $B$ denotes the cache size, is imposed to limit the total caching capacity. 

This problem is analogous to the \OMcC problem when the utility function $\mathcal{U}_t$ is assumed to be linear. For our experiments, we employed the Wikipedia Clickstream dataset, which categorizes referral pages on Wikipedia by language and contains over half a billion data points. This dataset includes the language of each page and the number of its referrers across the internet, providing a measure of the page's impression level. We interpret this impression level as the valuation of each data point. This interpretation aligns with real-world scenarios, where highly linked pages are more popular and, consequently, should remain in the cache for longer durations. 

To ensure fairness, a portion of the cache buffer must be allocated to pages from different languages. We model cached page sizes using an exponential distribution with a mean of 2.6 KB and set the total cache size to 100 MB. Our first experiment considers three languages—English, French, and Japanese—as distinct classes, aiming to evaluate utility, resource allocation, and competitive ratio performance. The arrival sequence is modeled by sampling with probabilities of 0.65 for English, 0.25 for French, and 0.10 for Japanese, repeated 200 times, with mean results reported. After removing extreme cases, normalized valuations are set as $\theta_1 = 116$, $\theta_2 = 178$, and $\theta_3 = 253.9$.

\begin{figure}[t]
    \centering
    \begin{subfigure}[b]{\textwidth}
        \centering
        \includegraphics[width=0.9\textwidth]{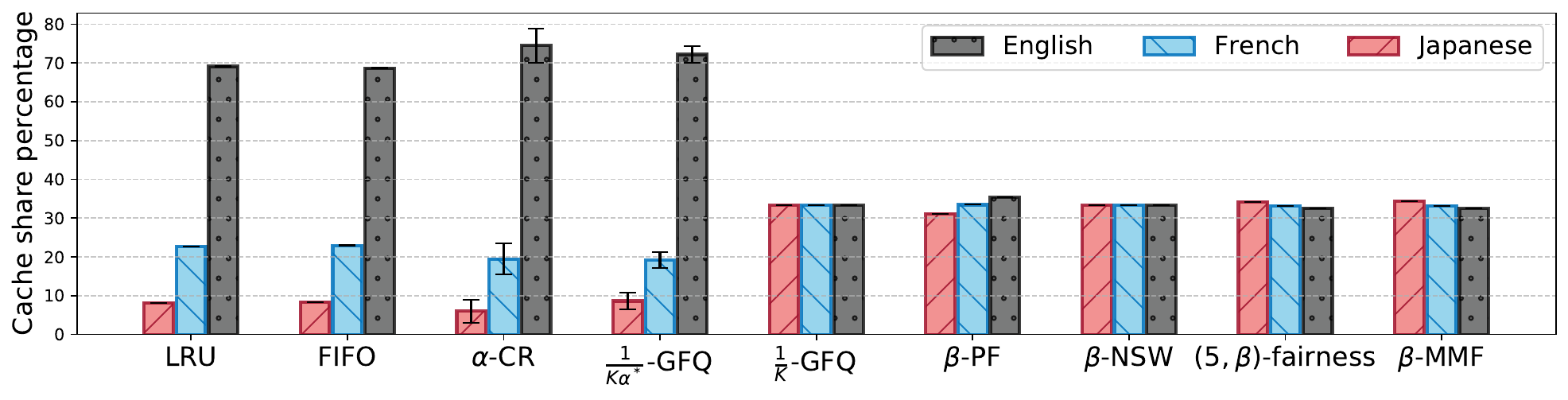}
        \caption{Resource share of each class.}
        \label{fig:sub1}
    \end{subfigure}
    \begin{subfigure}[b]{\textwidth}
        \centering
        \includegraphics[width=0.9\textwidth]{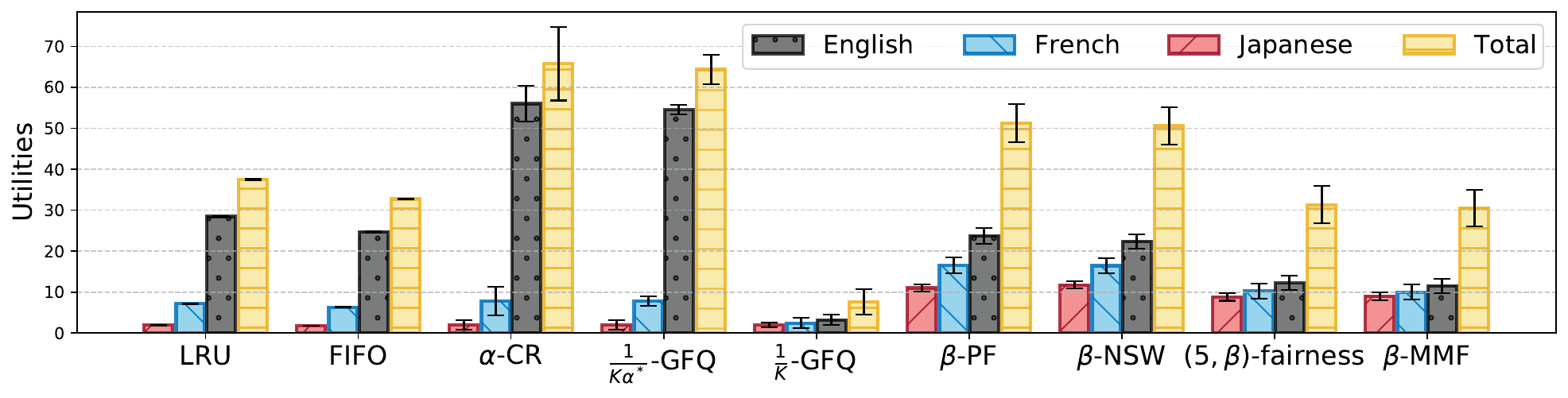}
        \caption{Utility of each class.}
        \label{fig:sub2}
    \end{subfigure}
    \caption{Results of the first experiment; $\theta_1 = 116$, $\theta_2 = 178$ and $\theta_3 = 253.9$.}
    \vspace{-10pt}
    \label{figure-utilites-resource-shares}
\end{figure}

We denote the \TTL optimal online algorithm without fairness constraints as $\alpha$-\CR and compare its performance to several fairness-driven algorithms. These include the \GFQ algorithm with fairness parameters $m_j = 1/3$ and $m_j = 1/(3\alpha_0^*)$ for all classes, where $\alpha_0^*$ is defined in Theorem \ref{theorem-LowerBound-multi-Class}. These variants are referred to as $\frac{1}{K}$-\GFQ and $\frac{1}{K\alpha^*}$-\GFQ, respectively. Additionally, we consider the \U-\PRB algorithm using the proportional fairness metric ($\beta$-\PF) and the \U-\PRB algorithm incorporating the $(\gamma, \beta)$-fairness measure. The latter employs $\gamma$ values of 1, 5, and $\infty$, corresponding to $\beta$-\NSW, $(5, \beta)$-fairness, and $\beta$-\MMF, respectively. Additionally, we evaluate alternative caching protocols, including hindsight \LRU and \FIFO algorithms, to compare with our online algorithms, demonstrating that our algorithms outperform these protocols in terms of fairness, even in a hindsight scenario. Unlike \TTL, these protocols cache all files for a fixed duration. As demonstrated by \cite{dehghan2019utility}, these approaches can also be formulated within a utility-based optimization framework

In the second experiment, we consider two arrival classes, with the goal of examining the impact of the maximum possible valuation for each class. Specifically, we fix the first language class with $\theta_1 = 48.5$ and vary the second language class one at a time. We then compare the empirical competitive ratio for each algorithm. Each experiment is repeated 200 times per algorithm and for each $\theta_2$ value, with the mean performance reported alongside a 95\% confidence interval. The final experiment aims to evaluate the impact of the number of classes on the competitive ratio. We fix English as the maximum valuation across all arrivals, then gradually increase the number of classes $K$ from 3 to 20. As before, each experiment is repeated 200 times for each class configuration, with the mean results reported alongside the 95\% confidence intervals.

\begin{figure}[t]
    \centering
    \begin{subfigure}[b]{0.48 \textwidth}
        \centering
        \includegraphics[width=\textwidth]{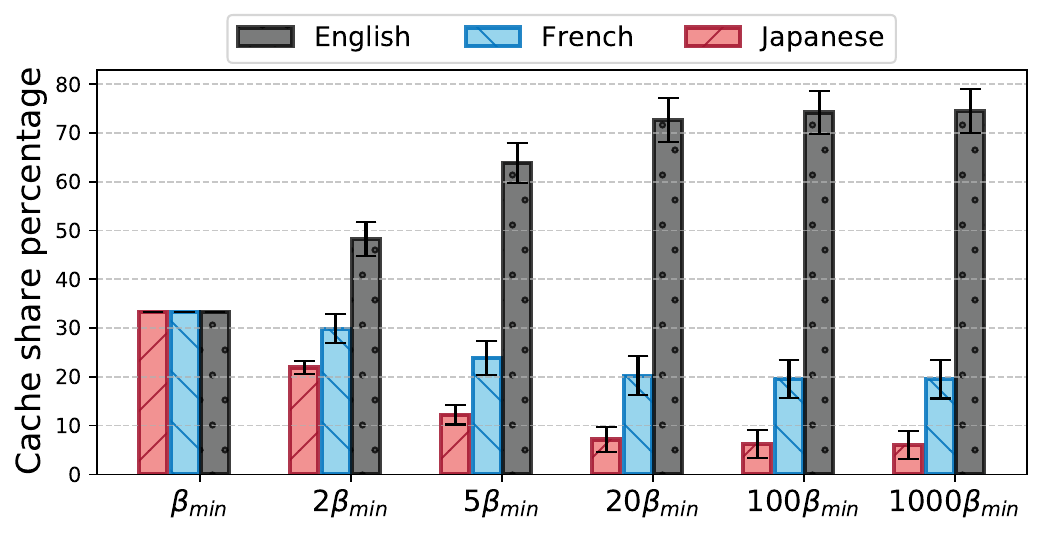}
        \caption{Resource of each class for different values of $\beta$.}
        \label{figure-1}
    \end{subfigure}
    \begin{subfigure}[b]{0.48\textwidth}
        \centering
        \includegraphics[width=\textwidth]{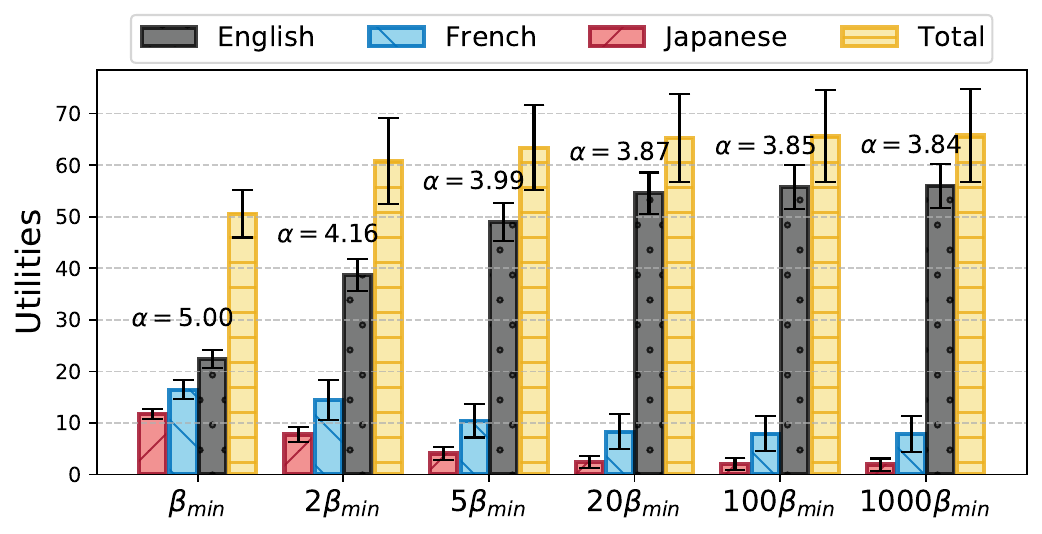}
        \caption{Utility of each class for different values of $\beta$.}
        \label{figure-2}
    \end{subfigure}
    \caption{Fairness-Efficiency trade-off for different values of $\beta$; $\theta_1 = 116$, $\theta_2 = 178$ and $\theta_3 = 253.9$.}
    \vspace{-10pt}
    \label{figure-alpha-beta-empirical-trade-off}
\end{figure}

\subsection*{Experiment Results.}
Figure \ref{figure-utilites-resource-shares} presents the distribution of resource shares across classes and the corresponding utilities received by each class. Utility-based fairness-driven algorithms and $\frac{1}{K}$-\GFQ maintain nearly equal resource allocation among classes, prioritizing fairness. In contrast, \LRU, \FIFO, $\alpha$-\CR, and $\frac{1}{K\alpha^*}$-\GFQ allocate more resources to high-valuation classes, emphasizing efficiency over fairness. Notably, $\beta$-\NSW achieves significantly higher total utility than $\frac{1}{K}$-\GFQ despite similar resource distribution, highlighting limitations of quantity-based fairness measures like \GFQ in utility maximization. Additionally, increasing $\gamma$ in the $(\gamma,\beta)$-fairness measure raises the share and utility of lower-valuation classes, promoting fairness but reducing overall utility, illustrating the trade-off between group fairness and performance. In another experiment using the same inputs, we evaluate the fairness-efficiency trade-off in $\beta$-\NSW by gradually increase $\beta$ from its minimum possible value ($\beta_{\min}$ which is obtained from Theorem \ref{theorem-proportional-fairness-beta}). Based on Figure \ref{figure-alpha-beta-empirical-trade-off}, as $\beta$ increases, resource allocation and utility for high-valuation classes rise while the empirical competitive ratio decreases, aligning with Theorem \ref{theorem-gamma-beta-fairness-trade-off-design} and Figure \ref{figure-trade-off-gamma-beta-alpha}.

Figure \ref{figure-emperical-3class} compares the empirical competitive ratio (ECR) using a CDF plot. The $\frac{1}{K\alpha^*}$-\GFQ and $\alpha$-\CR algorithms exhibit the fastest convergence and superior individual performance, consistent with Section \ref{section-GFQ}, where \GFQ achieved near-optimal competitive ratios under $M = \sum_j m_j \leq \frac{B}{\alpha^*_0}$. Interestingly, the $\frac{1}{K}$-\GFQ algorithm has a higher CDF than others, highlighting the importance of $m_j$ selections in \GFQ's performance. However, despite similar resource distributions to $\beta$-\NSW, $\frac{1}{K}$-\GFQ underperforms in individual welfare, exposing the inefficiency of quantity-based fairness in this context. The figure also illustrates the trade-off between group-level fairness and individual welfare, as equity-focused algorithms tend to achieve lower competitive ratios, underscoring the challenge of balancing fairness and efficiency.

In Figure \ref{figure-theta2-theta1-changes}, we observe that increasing the maximum valuation $\theta_2$ significantly impacts the $\frac{1}{K}$-\GFQ algorithm. To better illustrate this effect, we scaled the corresponding line by $0.5$ in the plot. In contrast, the impact on the $\alpha$-\CR\ algorithm remains minimal. Generally, algorithms that prioritize allocations to classes with lower valuations experience a greater degradation in competitive ratio as $\theta_2$ increases. This is because, in order to maintain fairness, the resource share allocated to the first class increases when $\theta_2$ rises, leading to an increase in the competitive ratio. In contrast, the $\alpha$-\CR\ and $\frac{1}{K\alpha^*}$-\GFQ algorithms are less affected by changes in $\theta_2$, as these algorithms inherently allocate more resources to classes with higher valuations (see Figure \ref{figure-utilites-resource-shares}).

Furthermore, in Figure \ref{figure-number-of-classes-changes}, we also observe that increasing the number of classes has the least effect on the $\alpha$-\CR\ algorithm. This is because $\alpha$-\CR\ focuses predominantly on allocating resources to the highest-valued class, and the introduction of additional classes has only a minor impact. However, as more classes with lower valuations are introduced, the concentration of arrivals with lower valuations increases, leading to a slight reduction in performance for the $\alpha$-\CR\ algorithm. On the other hand, algorithms that emphasize fairness experience a more pronounced increase in competitive ratio as the number of classes increases. This is due to the fact that having more classes inherently reduces the resource share available to the highest valued class. Consequently, the trade-off between group-level fairness and individual welfare becomes more apparent.

\begin{figure}[t]
    \centering
    \includegraphics[width=0.8\textwidth]{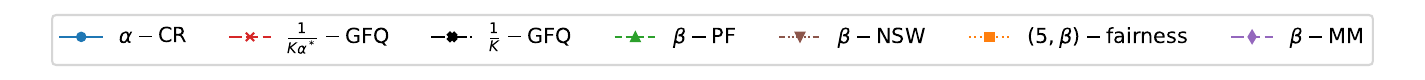}
    \begin{minipage}{0.32\textwidth}
        \centering
        \includegraphics[width=\textwidth]{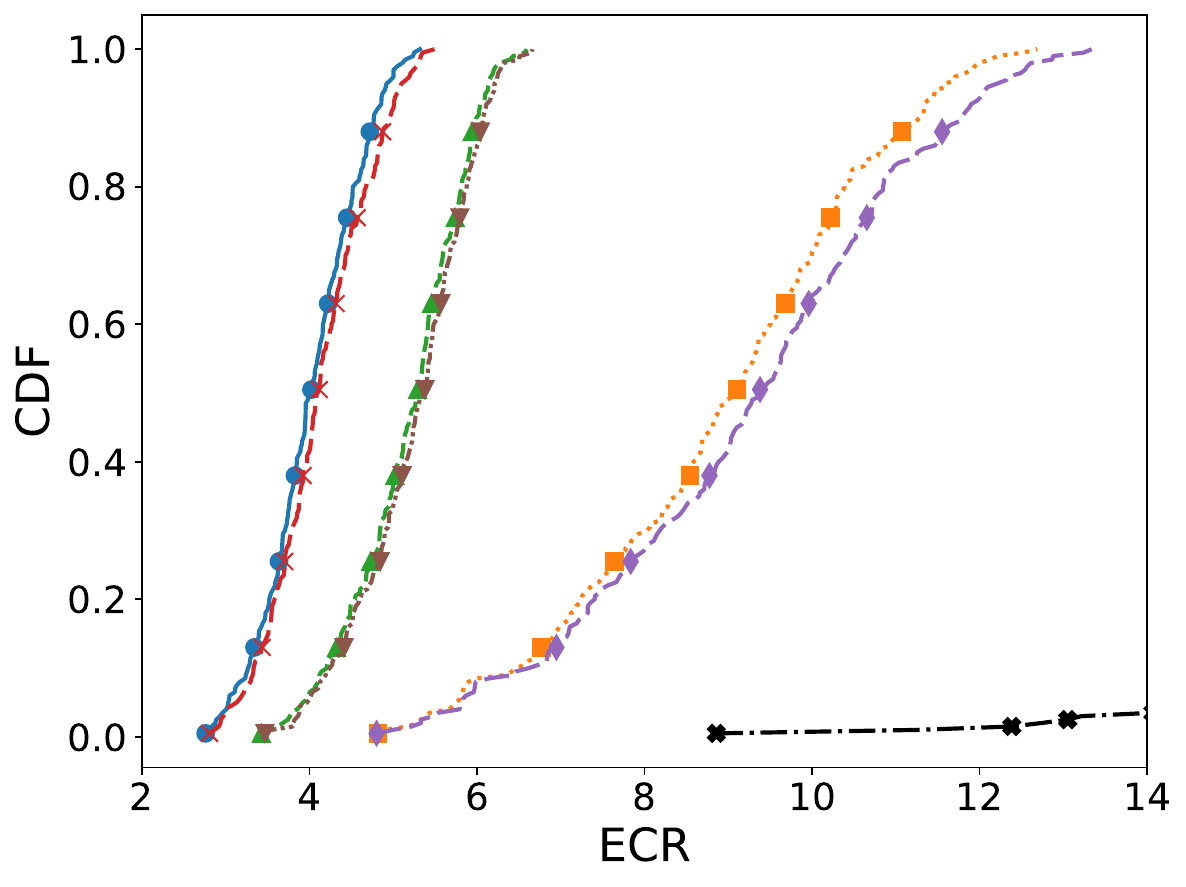}
        \caption{CDF plot of the ECR of algorithms; first experiment.}
        \label{figure-emperical-3class}
    \end{minipage}
    \hspace{0.1cm}
    \begin{minipage}{0.32\textwidth}
        \centering
        \includegraphics[width=\textwidth]{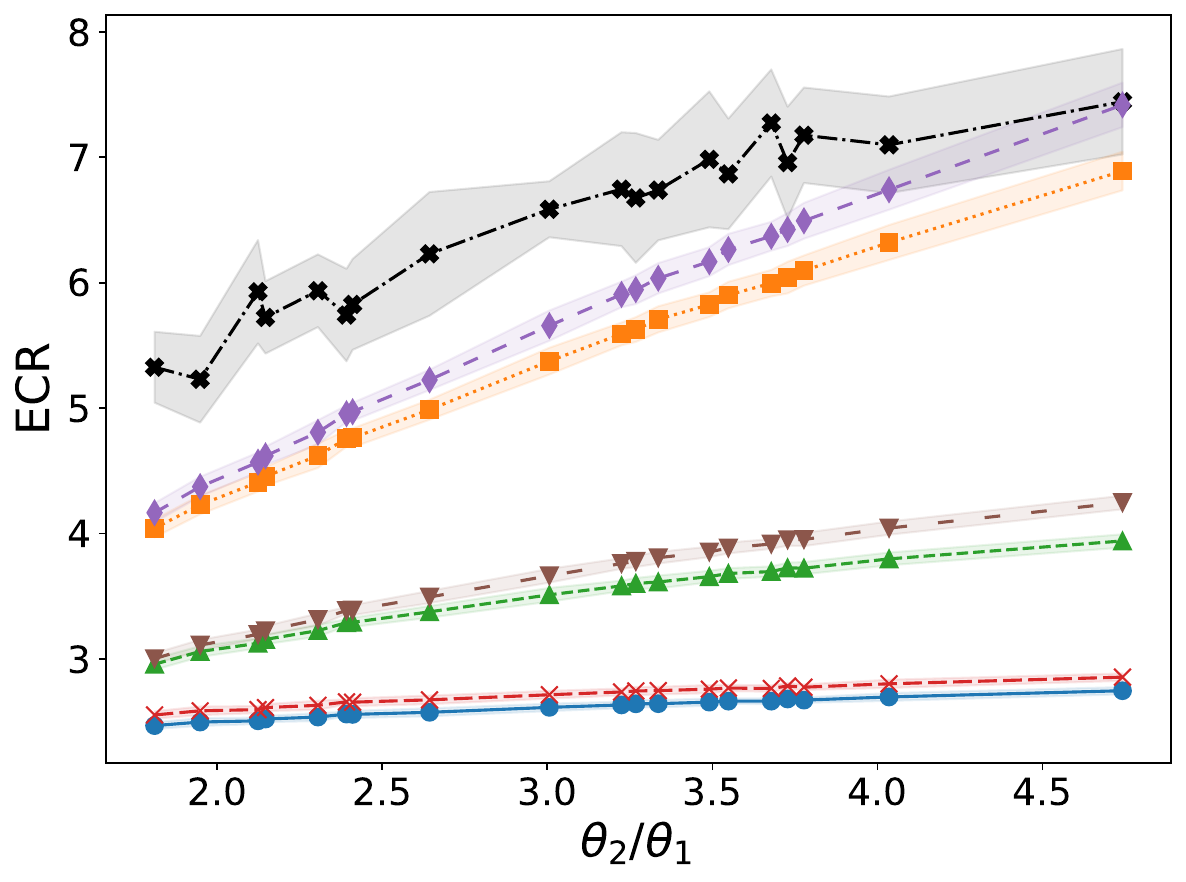}
        \caption{CDF plot of the E$\beta$-\PF of algorithms; first experiment.}
        \label{figure-theta2-theta1-changes}
    \end{minipage}
    \hspace{0.1cm}
    \begin{minipage}{0.32\textwidth}
        \centering
        \includegraphics[width=\textwidth]{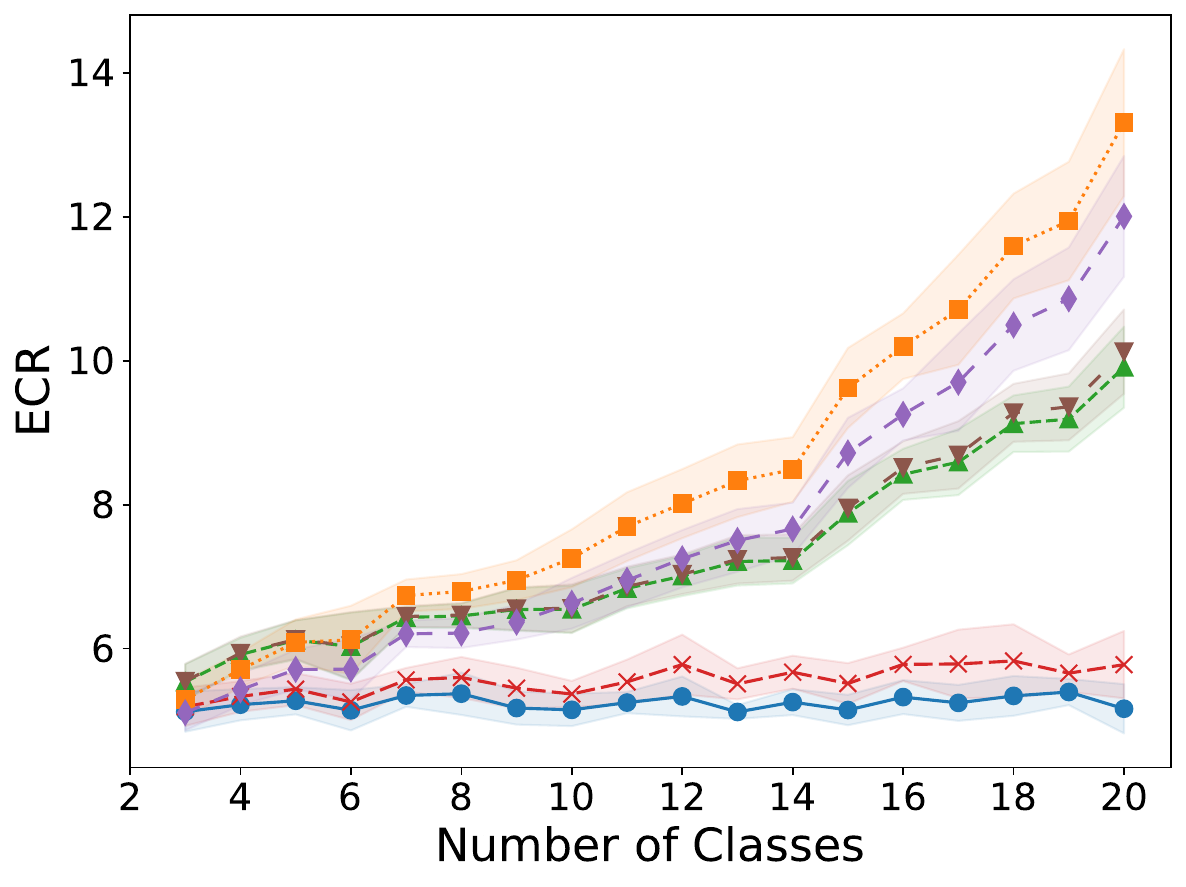}
        \caption{ECR vs the number of classes; {third experiment}.}
        \label{figure-number-of-classes-changes}
    \end{minipage}
    \vspace{-10pt}
\end{figure}

\section{Conclusions and Future Work}\label{section-conclusion}
In this work, we investigated the problem of online fair allocation among sequentially arriving agents, each belonging to a class or group. We proposed three novel threshold-based online algorithms: \Q-\PRB, \U-\PRB, and \SAMT. The \Q-\PRB algorithm employs a multi-segment threshold function that incorporates a quantity-based fairness metric among groups, termed \GFQ. The \U-\PRB algorithm uses a class-dependent threshold function to ensure fair allocation according to utility-based fairness metrics, specifically $\beta$-\PF and $(\gamma,\beta)$-fairness. Finally, \SAMT combines global and class-dependent threshold functions, leading to a family of algorithms that can smoothly balance fairness and efficiency. To achieve these results, we proposed a novel representative function-based approach to derive the lower bounds for online allocation under different group fairness notations, proving that our algorithms are optimal for \GFQ and $\beta$-\PF, and near-optimal for $(\gamma,\beta)$-fairness. Additionally, we empirically evaluated the performance of our model through numerical simulations on a real-world Wikipedia Clickstream dataset to address the caching problem using utility-based \TTL caching protocol, highlighting the trade-offs between group fairness and individual welfare in online settings.

Several open directions remain for extending this study on the trade-off between group fairness and individual welfare. A key unresolved question is whether a Pareto-optimal trade-off can be determined under the ($\gamma, \beta$)-fairness metric. Additionally, extending our model to the multi-resource setting would be a valuable but challenging next step.

\printbibliography{}

% \bibliography{ref}
% \bibliographystyle{plain}

\newpage
\appendix
{\Large{\textbf{Appendix}}}
\begin{table}[h]
	\caption{A summary of key notations.}
	\label{tab:notations}
    \small
	\begin{center}
		\begin{tabular}[P]{|c|l|}
            \hline
            \textbf{Notation} & \textbf{Description} \\
            \hline
            $t \in [T]$ & Time step index \\
            \hline
            $K \in \mathbb{N}$ & Number of arrival classes \\
            \hline
            $\theta_j $ & Fluctuation ratio of the arrivals of class $j$  \\
            \hline
            $B $ & Total resource budget  \\
            \hline
            \hline
            \multicolumn{2}{|c|}{\textbf{Online Inputs and Decisions}}\\
            \hline
            $j_t \in [K]$ & (\textit{Input}) Class of the arriving agent at time $t$ \\
            \hline 
            $v_t\in [1, \theta_j]$& (\textit{Input}) Valuation of the arriving agent at time $t$ when $j_t = j$\\
            \hline 
            $r_t$& (\textit{Input}) Allocation rate limit of arriving agent at time $t$\\
            \hline 
            ${x}_t \in [0,r_t]$ & (\textit{Decision}) Allocation decision at time $t$ \\
            \hline
            \hline
            \multicolumn{2}{|c|}{\textbf{\GFQ Notations}}\\
            \hline
            $m_j$& Minimum allocation requirement to class $j$\\
            \hline 
            $n_j$& Minimum number of arrivals from class $j$\\
            \hline
            $M=\sum\nolimits_{j\in[K]}m_j \in [0, B]$& Cumulative \GFQ constraint of all classes\\
            \hline 
            $C_j =  B - \sum_{i=1}^{j-1} m_i$ & Remaining resource after \GFQ allocation up to class $j-1$\\
            \hline 
            $D_j = \sum_{i=1}^{j-1} m_i\theta_i$& Maximum cumulative utility of \GFQ constraint up to class $j-1$\\
            \hline 
            $\phi(\cdot):[0, B-M] \to [1, \theta_K]$& The multi-segment threshold function shared among all classes\\
            \hline
            \hline
            \multicolumn{2}{|c|}{\textbf{$\beta$-\PF and  ($\gamma,\beta)$-fairness Notations}}\\
            \hline
            $U_j(\mathbf{x})=\sum_{t\in[T]} v_t x_t\cdot \boldsymbol{1}_{\{j_t=t\}}$ & Utility of the class $j$ under allocation $\mathbf{x}$\\
			\hline
            $b_j\in[0,B]$ &The reservation resource for class $j$ \\
            \hline
            $\phi_j(\cdot):[0,b_j] \to [1, \theta_j]$& The Local threshold function of the class $j$ \\
            \hline
            $\phi^G(\cdot):[0,B] \to [1, \theta_K]$ & The Global threshold function shared among all classes\\
            \hline
            $\psi_j(\cdot):[1, \theta_j]\to [0,b_j]$ & The Local utilization function of the class $j$ \\
            \hline
            $\psi^G(\cdot):[1, \theta_K]\to [0,B]$ & The Global utilization function shared among all classes\\
            \hline
            $u^j_t$ & Local utilization level of the algorithm for group $j$ by time $t$\\       
            \hline
            $u_t$ & Global utilization level of the algorithm by time $t$\\ 
            \hline
            $\alpha_j = 1 + \ln\theta_j$ &  The competitive ratio when all the arrivals are from class $j$ \\
            \hline
		\end{tabular}
	\end{center}
\end{table}

Throughout the proofs in this appendix, we assume \( r_t = B \) for all \( t \in [T] \) to simplify the analysis. This assumption does not compromise the validity of our results, as an adversary can always choose sufficiently large \( r_t \) values to penalize the online algorithm. This scenario is particularly relevant when the time horizon is unknown, and \( r_t \) does not affect the analysis outcomes but rather the execution of the algorithm. Importantly, this assumption does not influence the algorithm's design, as established in prior work \cite{sun2020competitive, lechowicz2023online}. As a result, we exclude \( r_t \) from the input sequence of arrivals and instead define the input sequence as \( I = \{(v_1, j_1), (v_2, j_2), \dots, (v_T, j_T)\} \).

\section{Proofs of Section \ref{section-GFQ}}\label{appendix-GFQ}
\subsection{Proof of Theorem \ref{theorem-LowerBound-multi-Class}}\label{appendix-gfq-mulit-upperbound}
Here we aim to determine the dynamics required for the increasing threshold function $\phi(u_{t})$ in Algorithm \ref{alg-function-daynamic-threshold-K-class} to ensure the algorithm's $\alpha$-competitiveness. We define the function $\Upsilon(v)$ for any value $v$ such that $v \in [1,\theta_K]$ as:
\begin{align*}
\Upsilon(v) = \argmax_{a \ge 0}{a\cdot v-\int_{0}^{a} \phi(b)db}.
\end{align*}
Based on Assumption \ref{assumption-gfq-arrival} there are at least one arrival from each class. Therefore, for a given instance $I$, with the maximum received valuation $v$, the objective value of the optimal offline algorithm on give instance $I$, $\OPT(I)$, can be upper-bounded as 
\begin{align*}
\OPT(I) \leq v(B-\sum_j m_j) + v\sum_{i\geq j}m_i + \sum_{i<j}m_i \theta_i = v C_j + D_j.
\end{align*}
Now in the case that $M \leq  \frac{B}{1 + \ln\theta_K - \sum_{j=1}^{K-1} \frac{m_j}{B} \ln(\frac{\theta_K}{\theta_j})}$, the performance of Algorithm \ref{alg-function-daynamic-threshold-K-class}  can be lower-bounded as follows:
\begin{align*}
    &\ALG(I) \geq M +  \Upsilon(1)+ \int_{\Upsilon(1)}^{\Upsilon(v)} \phi(u) du.
\end{align*}
Let $v$ be some value in the range $[\theta_{j-1},\theta_{j}]$ for some $j \in [1,K]$. Then we can see that 
\begin{align*}
\ALG(I) \geq &M + \Upsilon(1) + \sum_{i=1}^{j-1} \int_{\Upsilon(\theta_{i-1})}^{\Upsilon(\theta_{i})} \phi(u) du + \int_{\Upsilon(\theta_{j-1})}^{\Upsilon(v)} \phi(u) du.
\end{align*}
Based on the definition of $\phi(\cdot)$ we have 
\begin{align*}
    \ALG(I) \geq & M + \Upsilon(1)+ \sum_{i=1}^{j-1} \frac{C_i}{\alpha_0^*} \exp\left(\frac{\alpha_0(u + M) - B - \sum_{l=1}^{i-1}m_l \ln \theta_l}{C_i}\right)\bigg\rvert_{\Upsilon(\theta_{i-1})}^{\Upsilon(\theta_{i})}\nonumber\\
    & \qquad \qquad+\frac{C_j}{\alpha_0^*} \exp\left(\frac{\alpha_0(u + M) - B - \sum_{l=1}^{j-1}m_l \ln \theta_l}{C_i}\right)\bigg\rvert_{\Upsilon(\theta_{j-1})}^{\Upsilon(v)}.
\end{align*}
Now using the definition of $\Upsilon(\cdot)$ we can reduce the above inequality to 
\begin{align*}
    \ALG(I) \geq & M + \left(\frac{B}{\alpha_0^*} - M \right)+ \sum_{i=1}^{j-1}\frac{C_{i}}{\alpha_0^*}\left(\theta_i-\theta_{i-1}\right) + \frac{C_{j}}{\alpha_0^*}\left(v-\theta_{j-1}\right).
\end{align*}
Using the telescopic summation we have 
\begin{align*}
    \ALG(I) \geq & \frac{B}{\alpha_0^*}+ \frac{C_j}{\alpha_0^*}v - \frac{C_1}{\alpha_0^*} + \sum_{i=1}^{j-1}\frac{m_i \theta_i}{\alpha_0^*} \\
    =& \frac{C_j}{\alpha_0^*}v + \frac{D_{j}}{\alpha_0^*}.
\end{align*}
Therefore
\begin{align}
    \frac{\OPT(I)}{\ALG(I)}\leq \frac{v\cdot C_j + D_j}{v\frac{C_j}{\alpha_0^*} + \frac{D_{j}}{\alpha_0^*}} = \alpha_0^*.
\end{align}

Additionally, in the case that $M \in (\frac{\theta_{j^*-1} \cdot C_{j^*} + D_{j^*}}{\alpha_{j^*}},\frac{\theta_{j^*} \cdot C_{j^*} + D_{j^*}}{\alpha_{j^*}}]$, the performance of Algorithm \ref{alg-function-daynamic-threshold-K-class} can be lower-bounded as
\begin{align*}
    &\ALG(I) \geq M + \int_{\Upsilon(v^*)}^{\Upsilon(v)} \phi(u) du.
\end{align*}

Let $v$ be some value in the range $[U_{j-1},U_{j}]$ for some $j \in [j^{*},K]$. Then the algorithm performance is
\begin{align*}
    \ALG(I) \geq &M + \int_{\Upsilon(v^*)}^{\Upsilon(\theta_{j^*})} \phi(u) du + \sum_{i=j^*+1}^{j-1} \int_{\Upsilon(\theta_{i-1})}^{\Upsilon(\theta_{i})} \phi(u) du + \int_{\Upsilon(\theta_{j-1})}^{\Upsilon(v)} \phi(u) du.
\end{align*}

Same as before, using the definition of $\phi(\cdot)$ we can obtain
\begin{align*}
    \ALG(I) \geq& M + \frac{C_{j^*}}{\alpha_{j^*}}v^* \exp\left(\frac{\alpha_{j^*}\cdot u}{C_{j^*}}\right)\bigg\rvert_{\Upsilon(v^*)}^{\Upsilon(\theta_{j^*})}\nonumber\\
    & \qquad \qquad + \sum_{i=j^*+1}^{j-1} \frac{C_i}{\alpha_{j^*}}v^* \exp\left(\frac{\alpha_{j^*}\cdot u - \sum_{l=j^*}^{i-1}m_l \ln \frac{\theta_l}{v^*}}{C_i}\right)\bigg\rvert_{\Upsilon(\theta_{i-1})}^{\Upsilon(\theta_{i})}\nonumber \\
    & \qquad \qquad+\frac{C_j}{\alpha_{j^*}}v^* \exp\left(\frac{\alpha_{j^*}\cdot u- \sum_{l=j^*}^{j-1}m_l \ln \frac{\theta_l}{v^*}}{C_i}\right)\bigg\rvert_{\Upsilon(\theta_{j-1})}^{\Upsilon(v)}. 
\end{align*}
Using the definition of $\Upsilon(\cdot)$ we have
\begin{align*}
     \ALG(I) \geq & M + \frac{C_{j^*}}{\alpha_{j^*}}v^*\left(\frac{\theta_{j^*}}{v^*}-1\right) + \sum_{i=j^* + 1}^{j-1}\frac{C_{i}}{\alpha_{j^*}}v^*\left(\frac{\theta_{i}}{v^*}-\frac{\theta_{i-1}}{v^*}\right) + \frac{C_{j}}{\alpha_{j^*}}v^*\left(\frac{v}{v^*}-\frac{\theta_{j-1}}{v^*}\right).
\end{align*}
Now we can further simplify $\ALG(I)$ such that
\begin{align*}
    \ALG(I) \geq & M + \frac{C_j}{\alpha_{j^*}}v - \frac{C_{j^*}}{\alpha_{j^*}}v^* + \sum_{i=j^*+1}^{i-1}\frac{m_i \theta_i}{\alpha_{j^*}} \\
    = &  \frac{C_j}{\alpha_{j^*}}v + \frac{D_{j^*}}{\alpha_{j^*}} + \sum_{i=j^*+1}^{i-1}\frac{m_i \theta_i}{\alpha_{j^*}}\\
    =& \frac{C_j}{\alpha_{j^*}}v + \frac{D_{j}}{\alpha_{j^*}}.
\end{align*}
Therefore,
\begin{align*}
    \frac{\OPT(I)}{\ALG(I)}\leq \frac{v\cdot C_j + D_j}{v\frac{C_j}{\alpha_{j^*}} + \frac{D_{j}}{\alpha_{j^*}}} = \alpha_{j^*}.
\end{align*}
This concludes the proof of the competitive ratio provided in Theorem \ref{theorem-LowerBound-multi-Class}.

\subsection{Proof of Theorem \ref{theorem-gfq-lowerbound}}\label{appendix-gfq-lowerbound}
Let us start with the definition of the hard instance $ I^{\GFQ}$.
\begin{lettereddef}[\GFQ Fairness Guarantee Hard Instance: $I^{\GFQ}$] \label{def_instance_I_p_GFQ}
    Instance $I^{\GFQ}$ is defined as a scenario characterized by a continuous, non-decreasing sequence of valuation arrivals. In this scenario, each valuation is replicated for every class as long as it remains feasible. For some value of $\epsilon$ such that $\epsilon \rightarrow 0$, instance $I^{\GFQ}$ can be shown as follows:

\begin{align*}
    I^{\GFQ} = \Biggl\{&\underbrace{(1, 1), (1,2), \dots , (1, K)}_{K\textit{ number of agents}}, \underbrace{(1+\epsilon, 1), \dots, (1+\epsilon, K)}_{K\textit{ number of agents}}, \dots,\\
    &\underbrace{(\theta_{1}, 1), \dots ,(\theta_{1}, K)}_{K\textit{ number of agents}}, \underbrace{(\theta_{1}+\epsilon, 2), \dots , (\theta_{1}+\epsilon, K)}_{K-1 \textit{ number of agents}}, \dots, \underbrace{(\theta_{k}, K)}_{1\textit{ agent}}  \Biggr\},
\end{align*}
where in above $(v, j)$, $\forall j \in [k]$, corresponds to a buyer with valuation equal to $v$ from class $j$.
\end{lettereddef}
\begin{lettereddef}
     A utilization function $\psi(v):[1, \theta_K] \to [0, B]$ is defined as the final utilization of the budget $B$ after executing the instance $I^{\GFQ}$ by an online algorithm. 
\end{lettereddef}

\begin{letteredprop}
    \label{lemma:Utilization_Function}
    Suppose $1 = \theta_0 \leq \theta_1 \leq \dots \leq \theta_K$. If there exists an $\alpha$-competitive online algorithm, then there must exist a utilization function  $\psi(v):[1,\theta_K] \to [0, B]$ such that $\psi(\cdot)$ is a non-decreasing function and satisfies
    \begin{equation*}
    \label{eq:Utilization_Function}
        \begin{cases}
            \psi(1) + \int_1^v ud\psi(u) \geq \frac{1}{\alpha}\left[v(B-M) + v\sum_{i\geq j}m_i + \sum_{i<j}m_i \theta_i \right], \quad\forall v\in (\theta_{j-1}, \theta_j]\\
            \psi(1) \geq \max\{M, \frac{B}{\alpha}\} , \psi(\theta_K) \leq B.
        \end{cases}
    \end{equation*}
\end{letteredprop}

\begin{proof}
    Since online algorithms operate in real-time, making irrevocable decisions solely based on causal information, $\psi(v)$ exhibits non-decreasing behavior within the interval $[1,\theta_K]$. Considering the maximum effective utilization as $B$, the utilization function must adhere to the boundary condition $\psi(\theta_K)\leq B$. Moreover, by definition, the total value attained by an $\alpha$-competitive online algorithm is at least $\frac{1}{\alpha}$ of the offline optimum for any arrival instances. Hence, within the context of instance $I^\GFQ$, we can conclude that
    \begin{align*}
        &\psi(1)\geq \frac{B}{\alpha}.
    \end{align*}
    Also $\psi(1)$ should always be at least $M$ in order to make sure the minimum allocated amount to each class constraint is satisfied. Thus
    \begin{align*}
        \psi(1)\geq \{M , \frac{B}{\alpha}\}.
    \end{align*}
    Under the conditions of instance $I^{\GFQ}$, for every $v$ falling within the interval $(\theta_{j-1}, \theta_j]$, the optimal algorithm, denoted as \OPT, maximizes the allocation by assigning the entire budget to the highest possible value ($v$). To ensure the minimum requirements of each class are met, \OPT allocates a value of $v$ to all classes $i$ where $\theta_i$ is greater than or equal to $v$, and assigns a value of $\theta_i$ to classes $i$ where $\theta_i$ is less than $v$. This means that \OPT optimally distributes the budget to classes based on their maximum possible value. Thus, the optimal value achievable in the offline setting is:
    \begin{align*}
        \OPT(I^{\GFQ})  \leq v(B-M) + v\sum_{i\geq j}m_i + \sum_{i<j}m_i \theta_i, \quad \forall v\in (\theta_{j-1}, \theta_j]
    \end{align*}
    On the other hand, the performance of any online algorithm $ \ALG $ with utilization function $ \psi $ is
    \begin{align*}
        \ALG(I^{\GFQ}) \geq \psi(1) + \int_1^v ud\psi(u).
    \end{align*}
     Combining the above analysis with the definition of $\alpha$-competitiveness,  the differential equation in Proposition \ref{lemma:Utilization_Function} follows.
\end{proof}

For any value of $\alpha$, let us define the function $\omega_{\alpha}(v):[1,\theta_K] \rightarrow R$ as follows:
\begin{align}
    \omega_{\alpha}(v) = \begin{cases}
        \frac{B}{\alpha} \cdot v & \quad v \in [1,\theta_{1}], \\
        \sum_{j=1}^{K-1} \mathbf{1}_{v \in [\theta_{j},\theta_{j+1}]} \cdot \frac{1}{\alpha}\left[v\cdot C_{j+1} + D_{j+1}\right] & \quad v \in [\theta_{1},\theta_{K}].
    \end{cases}
\end{align}
Let $v^{*} = \omega_{\alpha}^{-1}(L \cdot \max\{\frac{B}{\alpha},M\})$. It is easy to see that $v^{*}$ is well-defined. Below, we will consider the two cases where in the first case $\frac{B}{\alpha} \ge M$ and as a result $v^{*} = 1$, and in the latter $\frac{B}{\alpha} < M$ and as a result $v^{*} > 1$. 

\paragraph{Case 1: $\frac{B}{\alpha} \ge M$}\quad
Based on the inequality in Lemma \ref{lemma:Utilization_Function}, we have
\begin{align*}
    &\psi(1) + \int_1^v ud\psi(u) \geq \omega_{\alpha}(v).
\end{align*}
Then we can rewrite the integration by its parts as
\begin{align*}
    &\psi(1) + [u\psi(u)]\Big|_1^v -\int_1^v \psi(u) du \geq  \omega_{\alpha}(v).
    % \Rightarrow & v\psi(v) -\int_L^v \psi(u)d u \geq \omega_{\alpha}(v)\\
\end{align*}
And consequently:
\begin{align*}
    \psi(v) \ge \frac{\omega_{\alpha}(v)}{v} + \frac{1}{v}\int_1^v \psi(u) du .
\end{align*}
By applying the Gronwall's inequality \cite{mitrinovic2012inequalities} we get
\begin{align*}
    &\psi(v) \geq \frac{\omega_\alpha(v)}{v} + \int_1^v \frac{\omega_\alpha(u)}{u}\cdot\frac{1}{u} du.
\end{align*}
For $v \in [\theta_{j-1}, \theta_j]$ we have:
\begin{align*}
    &\psi(v) \geq \frac{\omega_\alpha(v)}{v} + \int_1^{\theta_1} \frac{\omega_\alpha(u)}{u}\cdot\frac{1}{u} du + \sum_{i=2}^{j-1}\int_{\theta_{i-1}}^{\theta_i} \frac{\omega_\alpha(u)}{u}\cdot\frac{1}{u} du + \int_{\theta_{j-1}}^{v} \frac{\omega_\alpha(u)}{u}\cdot\frac{1}{u} du. 
\end{align*}
Using the the definition of $\omega_\alpha(\cdot)$ we have
\begin{align*}
    &\psi(v)\geq \frac{C_j}{\alpha} + \frac{D_j}{\alpha \cdot v} + \frac{B}{\alpha} \ln\left(\theta_1\right) + \sum_{i=2}^{j-1} \frac{C_i}{\alpha}\ln\left(\frac{\theta_i}{\theta_{i-1}}\right) - \frac{D_i}{\alpha}\left(\frac{1}{\theta_i} - \frac{1}{\theta_{i-1}}\right) \\
    & \hspace{5.35cm} + \frac{C_j}{\alpha}\ln\left(\frac{v}{\theta_{j-1}}\right) - \frac{D_j}{\alpha}\left(\frac{1}{v} - \frac{1}{\theta_{j-1}}\right).\nonumber
\end{align*}
By further simplifying the above inequality we can see
\begin{align*}
    &\psi(v) \geq \frac{C_j}{\alpha} + \frac{D_j}{\alpha\cdot \theta_{j-1}} -\sum_{i=2}^{j-1} \frac{D_i}{\alpha}\left(\frac{1}{\theta_i} - \frac{1}{\theta_{i-1}}\right) + \sum_{i=1}^{j-1} \frac{C_i}{\alpha}\ln\left(\frac{\theta_i}{\theta_{i-1}}\right) +\frac{C_j}{\alpha}\ln\left(\frac{v}{\theta_{j-1}}\right).
\end{align*}
By further simplification we will get
\begin{align*}
    \psi(v) \geq \frac{C_j}{\alpha} + \sum_{i=1}^{j-1} \frac{1}{\alpha\cdot\theta_i}\left(D_{i+1} - D_i\right) + \sum_{i=1}^{j-1} \frac{C_i}{\alpha}\ln\left(\frac{\theta_i}{\theta_{i-1}}\right)+\frac{C_j}{\alpha}\ln\left(\frac{v}{\theta_{j-1}}\right).
\end{align*}
Using the telescopic summation we will get
\begin{align*}
    &\psi(v) \geq \frac{C_j}{\alpha} + \sum_{i=1}^{j-1} \frac{m_i}{ \alpha} + \sum_{i=1}^{j-1} \frac{C_i}{\alpha}\ln\left(\frac{\theta_i}{\theta_{i-1}}\right)+\frac{C_j}{\alpha}\ln\left(\frac{v}{\theta_{j-1}}\right),
\end{align*}
which can be equivalently stated as
\begin{align*}
    \psi(v) \geq \frac{B}{\alpha}+\sum_{i=1}^{j-1} \frac{C_i}{\alpha}\ln\left(\frac{\theta_i}{\theta_{i-1}}\right)+\frac{C_j}{\alpha}\ln\left(\frac{v}{\theta_{j-1}}\right).
\end{align*}
Therefore by setting $v=\theta_K$ and applying the inequality from Proposition \ref{lemma:Utilization_Function}, we obtain::
\begin{align*}
    &B\geq \psi(\theta_K)\geq \frac{B}{\alpha} +\sum_{i=1}^K \frac{C_i}{\alpha}\ln\left(\frac{\theta_i}{\theta_{i-1}}\right).
\end{align*}
And as a result
\begin{align*}
    \alpha \geq 1 + \sum_{i=1}^K \frac{C_i}{B}\ln\left(\frac{\theta_i}{\theta_{i-1}}\right) = 1 + \ln(\theta_K) - \sum_{l=1}^{K-1} \frac{m_l}{B} \ln(\frac{\theta_K}{\theta_l}).
\end{align*}

\paragraph{Case 2: $M> \frac{B}{\alpha}$}\quad Let $v^{*}$ be some value in the range $[\theta_{j^{*}-1},\theta_{j^{*}}]$ for all  $j^{*} \in [1,K]$. Since $\psi(1) \ge M$, and the fact that $\psi(v)$ function is an increasing function, then we will have: 
\begin{align*}
     & v\psi(v) - \omega_{\alpha}(v) \ge \int_1^v \psi(u) du \ge \int_{1}^{v^{*}} M du + \int_{v^{*}}^{v} \psi(u)du ,
\end{align*}
which indicates that
\begin{align*}
     % \Rightarrow &v\psi(v) - \omega_{\alpha}(v) \ge (v^{*}-1) \cdot M + \int_{v^{*}}^{v} \psi(u)du, \ v\in [v^{*},\theta^{K}]\\
     \psi(v)  \geq  \frac{\omega_{\alpha}(v)}{v} + (v^{*}-1) \frac{M}{v} + \frac{1}{v}\int_{v^{*}}^{v} \psi(u)du.
\end{align*}
Like the previous case, we use the Gronwall's inequality and therefore:
\begin{align*}
    & \psi(v) \geq \frac{\omega_{\alpha}(v)}{v} + (v^{*}-1) \frac{M}{v} + \int_{v^{*}}^{v} \left[\frac{\omega_{\alpha}(u)}{u} + (v^{*}-1) \frac{M}{u}\right] \frac{1}{u}du.
\end{align*}
For $v\in [\theta_{j-1}, \theta_j]$ such that $v \geq v^*$ we have:
\begin{align*}
    &\psi(v) \geq \frac{C_j}{\alpha} + \frac{D_j}{\alpha v} + (v^{*}-1) \frac{M}{v} + \int_{v^{*}}^{v} \left[\frac{\omega_{\alpha}(u)}{u} + (v^{*}-1) \frac{M}{u}\right] \frac{1}{u}du.
\end{align*}
This can be expanded as
\begin{align*}
    &\psi(v) \geq \frac{C_j}{\alpha} + \frac{D_j}{\alpha v} + (v^{*}-1) \frac{M}{v} + \int_{v^{*}}^{\theta_{j^*}} \left[\frac{\omega_{\alpha}(u)}{u} + (v^{*}-1) \frac{M}{u}\right] \frac{1}{u}du\\
    & \hspace{4.3cm}+ \sum_{i=j^* +1}^{j-1}\int_{\theta_i-1}^{\theta_{i}} \left[\frac{\omega_{\alpha}(u)}{u} + (v^{*}-1) \frac{M}{u}\right] \frac{1}{u}du \\
    & \hspace{4.3cm}+ \int_{\theta_j-1}^{v} \left[\frac{\omega_{\alpha}(u)}{u} + (v^{*}-1) \frac{M}{u}\right] \frac{1}{u}du.
\end{align*}
Replacing $\omega_\alpha(\cdot)$ by its definition we have
\begin{align*}
    &\psi(v) \geq \frac{C_j}{\alpha} + \frac{D_j}{\alpha v} + (v^{*}-1) \frac{M}{v} + \frac{C_{j^*}}{\alpha}\ln\left(\frac{\theta_{j^*}}{v^*}\right) - \left[\frac{D_{j^*}}{\alpha} + (v^* -1)M\right]\left(\frac{1}{\theta_{j^*}} - \frac{1}{v^*}\right)\\
    & \hspace{4.3cm}+ \sum_{i=j^* +1}^{j-1}\frac{C_i}{\alpha}\ln\left(\frac{\theta_{i}}{\theta_{i-1}}\right) - \left[\frac{D_{i}}{\alpha} + (v^* -1) M \right]\left(\frac{1}{\theta_{i}} - \frac{1}{\theta_{i-1}}\right) \\
    & \hspace{4.3cm}+ \frac{C_j}{\alpha}\ln\left(\frac{v}{\theta_{j-1}}\right) - \left[\frac{D_{j}}{\alpha} + (v^* -1)M\right]\left(\frac{1}{v} - \frac{1}{\theta_{j-1}}\right).
\end{align*}
By further simplifying it we have
\begin{align*}
    &\psi(v)\geq  \frac{C_j}{\alpha} + \sum_{i=j^*}^{j-1} \frac{1}{\theta_i \alpha}\left(D_{i+1} - D_i\right) + \frac{D_{j^*}}{v^*\alpha} + (v^{*}-1) \frac{M}{v^*} \\
    &\hspace{2cm}+ \frac{C_{j^*}}{\alpha}\ln\left(\frac{\theta_{j^*}}{v^*}\right) + \sum_{i=j^* +1}^{j-1}\frac{C_i}{\alpha}\ln\left(\frac{\theta_{i}}{\theta_{i-1}}\right) + \frac{C_j}{\alpha}\ln\left(\frac{v}{\theta_{j-1}}\right).
\end{align*}
Here by using the telescopic summation we will further get
\begin{align*}
    &\psi(v)\geq  \frac{C_{j^*}}{\alpha} + \frac{D_{j^*}}{v^*\alpha} + (v^{*}-1) \frac{M}{v^*} + \frac{C_{j^*}}{\alpha}\ln\left(\frac{\theta_{j^*}}{v^*}\right) + \sum_{i=j^* +1}^{j-1}\frac{C_i}{\alpha}\ln\left(\frac{\theta_{i}}{\theta_{i-1}}\right) + \frac{C_j}{\alpha}\ln\left(\frac{v}{\theta_{j-1}}\right).
\end{align*}
Now by the definition of $\omega_\alpha(v^*)$ we know: 
\begin{align*}
    \omega_\alpha(v^*) = v^*\frac{C_{j^*}}{\alpha} + \frac{D_{j^*}}{\alpha}.
\end{align*}
And also by the definition of of $v^* = \omega_\alpha^{-1}(M)$ we have:
\begin{align*}
    &\psi(v)\geq  \frac{M}{v^*} + (v^{*}-1) \frac{M}{v^*} + \frac{C_{j^*}}{\alpha}\ln\left(\frac{\theta_{j^*}}{v^*}\right) + \sum_{i=j^* +1}^{j-1}\frac{C_i}{\alpha}\ln\left(\frac{\theta_{i}}{\theta_{i-1}}\right) + \frac{C_j}{\alpha}\ln\left(\frac{v}{\theta_{j-1}}\right).\nonumber
\end{align*}
This can be simplified to
\begin{align*}
    &\psi(v)\geq  M + \frac{C_{j^*}}{\alpha}\ln\left(\frac{\theta_{j^*}}{v^*}\right) + \sum_{i=j^* +1}^{j-1}\frac{C_i}{\alpha}\ln\left(\frac{\theta_{i}}{\theta_{i-1}}\right) + \frac{C_j}{\alpha}\ln\left(\frac{v}{\theta_{j-1}}\right).
\end{align*}
If we set $v=\theta_K$ we can see that:
\begin{align*}
    &B\geq \psi(\theta_K)\geq  M + \frac{C_{j^*}}{\alpha}\ln\left(\frac{\theta_{j^*}}{v^*}\right) + \sum_{i=j^* +1}^{K}\frac{C_i}{\alpha}\ln\left(\frac{\theta_{i}}{\theta_{i-1}}\right).
\end{align*}
It is easy to verify
\begin{align*}
    \frac{C_{j^*}}{\alpha}\ln\left(\frac{\theta_{j^*}}{v^*}\right) + \sum_{i=j^* +1}^{K}\frac{C_i}{\alpha}\ln\left(\frac{\theta_{i}}{\theta_{i-1}}\right) = \frac{C_{j^*}}{\alpha}\ln\left(\frac{\theta_K}{v^*}\right) - \frac{1}{\alpha} \sum_{i=j^*}^{K-1} m_i \ln\left(\frac{\theta_K}{\theta_i}\right).
\end{align*}
Therefore,
\begin{align*}
    &B \geq M + \frac{C_{j^*}}{\alpha}\ln\left(\frac{\theta_K}{v^*}\right) - \frac{1}{\alpha} \sum_{i=j^*}^{K-1} m_i \ln\left(\frac{\theta_K}{\theta_i}\right).
\end{align*}
This implies
\begin{align*}
    &\alpha \geq \frac{1}{B - M}\Biggl[C_{j^*}\ln\left(\frac{\theta_K}{v^*}\right) - \underbrace{\sum_{i=j^*}^{K-1} m_i \ln\left(\frac{\theta_K}{\theta_i}\right)}_{X} \Biggr].
\end{align*}
Now by replacing $v^* = \frac{\alpha M - D_{j^*}}{C_{j^*}}$ we get
\begin{align*}
    &\alpha \geq \frac{1}{B - M}\left[C_{j^*}\ln\left(\frac{\theta_{K} C_{j^*}}{\alpha M - D_{j^*}}\right) - X\right].
\end{align*}
Therefore,
\begin{align*}
    \left(\alpha M - D_{j^*}\right) \exp\left(\frac{\alpha (B-M)}{C_{j^*}}\right) \geq \theta_{K}C_{j^*}\exp\left(-\frac{X}{C_{j^*}}\right).
\end{align*}
Let us define $P = \frac{\alpha (B-M)}{C_{j^*}}$. Then:
\begin{align*}
    &\Biggl(\underbrace{P-\frac{D_{j^*}(B-M)}{C_{j^* }\cdot M}}_{S}\Biggr)e^P  \geq  \frac{\theta_{K}(B-M)}{M} \exp\left(-\frac{X}{C_{j^*}}\right).
\end{align*}
As a result
\begin{align*}
    &S \cdot e^S  \geq \frac{\theta_{K}(B-M)}{M} \exp\left(-\frac{X}{C_{j^*}}\right)\exp\left(-\frac{D_{j^*}(B-M)}{C_{j^*} \cdot M}\right).
\end{align*}
Using the definition of Lambert W-function we have:
\begin{align*}
    &S \geq  W\left(\frac{\theta_{K}(B-M)}{M} \exp\left(-\frac{X}{C_{j^*}}\right)\exp\left(-\frac{D_{j^*}(B-M)}{C_{j^*} \cdot M}\right)\right).
\end{align*}
Now by replacing $S$ and $P$ we will get
\begin{align*}
    &\alpha \geq \frac{D_{j^*}}{M}+ \frac{C_{j^*}}{B-M}W\left(\frac{\theta_{K}(B-M)}{M} \exp\left(-\frac{X}{C_{j^*}}\right)\exp\left(-\frac{D_{j^*}(B-M)}{C_{j^*} \cdot M}\right)\right).
\end{align*}
This completes the proof of Theorem \ref{theorem-gfq-lowerbound} and shows the optimality of the design presented in Theorem \ref{theorem-LowerBound-multi-Class} in all cases.

\subsection{Case Study $K =2$.}\label{appendix-gfq-case-study}
Let us consider the case that the arrivals could be from only two different classes. Based on the results of Theorem \ref{theorem-LowerBound-multi-Class}, we can see that depending on the $M$, there could be three different possibilities for the threshold function. Let us go through these three cases in details:

$\bullet$ If $M \leq  \frac{M}{\alpha_0^*}$, where $\alpha_0^*$ is the competitive ratio and it is defined as follows:
\begin{align*}
     \alpha_{0}^* \coloneqq 1 + \ln\theta_2 - \frac{m_1}{B} \ln\frac{\theta_2}{\theta_1}, 
\end{align*}
then the $\phi$ function is given by 
 \begin{align*}
    \phi(u) = \begin{cases}
        1 & u\in [0, \frac{B}{\alpha_{0}^*} - M],\\
        e^{\left(\frac{\alpha_0^* (u+M)}{B}-1\right)} & u \in (\frac{B}{\alpha_{0}^*} - M,\frac{B}{\alpha_0^*} - M +\frac{B}{\alpha_0^*}\ln\theta_1],\\
        e^{\left(\frac{\alpha_0^* (u+M)-B - m_1 \ln\theta_1}{B-m_1}\right)} & u \in (\frac{B}{\alpha_0^*} - M + \frac{B}{\alpha_0^*}\ln\theta_1,B-M].
    \end{cases}
\end{align*}
As we can observe, the optimal competitive ratio in this case is lower than that of the \OMcC without the \GFQ guarantee, which may seem counter intuitive. This is because, although adding \GFQ makes the problem more challenging, it also imposes restrictions on the adversary, potentially reducing the performance of the optimal offline algorithm as well. Another key observation is that when $M \leq \frac{B}{\alpha_0^*}$, the only \GFQ parameter that affects the competitive ratio is $m_1$. We will later numerically analyze how changes in $m_2$ impact the competitive ratio.

$\bullet$ If $M \in (\frac{B}{\alpha_1^*},\frac{B}{\alpha_{1}^*} \theta_{1} ]$, where $\alpha_{1}^*$ is defined as
\begin{align*}
     \alpha_{1}^* \coloneqq \frac{B}{B-M}W\left(\frac{\theta_{2}(B-M)}{M} e^{\left(-\frac{m_1}{B}\ln\frac{\theta_2}{\theta_1}\right)}\right),
\end{align*}
then Algorithm \ref{alg-function-daynamic-threshold-K-class} is $ \alpha_1^* $-competitive if $ \phi $ is given by
\begin{align*}
        \phi(u) = \begin{cases}
            v^* e^{\left(\frac{\alpha_1^* u}{B}\right)} & u \in [0,\frac{B}{\alpha_1^*}\ln\frac{\theta_1^*}{v^*}],\\
            v^*e^{\left(\frac{\alpha_1^* u - m_1 \ln\frac{\theta_1}{v^*}}{B-m_1}\right)} & u \in [\frac{B}{\alpha_1^*}\ln\frac{\theta_1}{v^*},B-M],
        \end{cases}
    \end{align*}
where $v^* = (\alpha_1^*\cdot M)/B$. 
    
Here, we observe that the number of segments in the threshold function will decrease by 1. This occurs because to satisfy the \GFQ requirement, the algorithm must accept more arrivals at the beginning of the allocation process, regardless of their valuation. This puts \ALG at a disadvantage compared to the offline optimal solution \OPT. Consequently, in the optimal algorithm, to compensate for this, the threshold function begins to increase from a value $v^*$ between 1 and $\theta_1$.
    
$\bullet$ If $M \in (\frac{B}{\alpha_{2}^*}\theta_{1}, B]$, where $\alpha_{2}^*$ is defined as 
\begin{align*}
   \alpha_{2}^* \coloneqq \theta_1\frac{m_1}{M} + \frac{B-m_1}{B-M}W\left(\frac{\theta_{2}(B-M)}{M} e^{\left(-\frac{\theta_1m_1(B-M)}{(B-m_1)M}\right)}\right),
\end{align*}
then Algorithm \ref{alg-function-daynamic-threshold-K-class} is $ \alpha_2^* $-competitive if $ \phi $ is given by
\begin{align*}
    \phi(u) = v^*e^{\left(\frac{\alpha_2^* u}{B-m_1}\right)}, \qquad  \forall u \in [0,B-M].
\end{align*}
where $v^* = (\alpha_2^*M - m_1\theta_1)/(B-m_1)$.

Here we can see that when $ M $ approaches $ B $, $\alpha_2^* $ converges to \( (\theta_{1}m_{1}+ \theta_{2}(B-m_1))/B \). The intuition is that when $M=B$, due to the fairness requirement, no online algorithm can effectively reserve any portion of its resource for future agents. Thus, in the worst case, no online algorithm can perform better than $B$, while the offline optimal algorithm may achieve the maximum revenue $\theta_{1}m_{1}+ \theta_{2}(B-m_1)$, leading to the worst-case competitive ratio $ (\theta_{1}m_{1}+ \theta_{2}(B-m_1))/B$. 

At first glance, it might appear that the three intervals of $M$ that define the three cases in Theorem \ref{theorem-LowerBound-multi-Class} are not continuous and do not fully cover the range of $ [0, B] $. However, as the value of $M$ approaches the end-point of one interval (e.g., the end-point $ \frac{B}{\alpha_1^*} \theta_1 $ of the second interval), the start-point of the next interval (e.g., $ \frac{B}{\alpha_2^*} \theta_1$) also converges to the end-point of the last interval. This observation is illustrated in Figure \ref{fig:GFQ_alphas}, where we fix all the parameters except the value of $m_{2}$ and show how the competitive ratio of Algorithm \ref{alg-function-daynamic-threshold-K-class} changes with variations in $ M $. As $m_2$ increases, the competitive ratio of Algorithm \ref{alg-function-daynamic-threshold-K-class}, denoted by $ \CR^* $, continuously increases. This outcome was foreseeable, since we should allocate a larger share of resources to agents regardless of their valuations to ensure fairness guarantee. Moreover, we can observe that $ \CR^* $ switches from $ \alpha_0^* $ to $ \alpha_1^* $ and $ \alpha_2^* $ w.r.t. the increase of $ m_2 $ (or equivalently, the increase of $ M $). 

\begin{figure}
    \centering
    \includegraphics[width=0.5\linewidth]{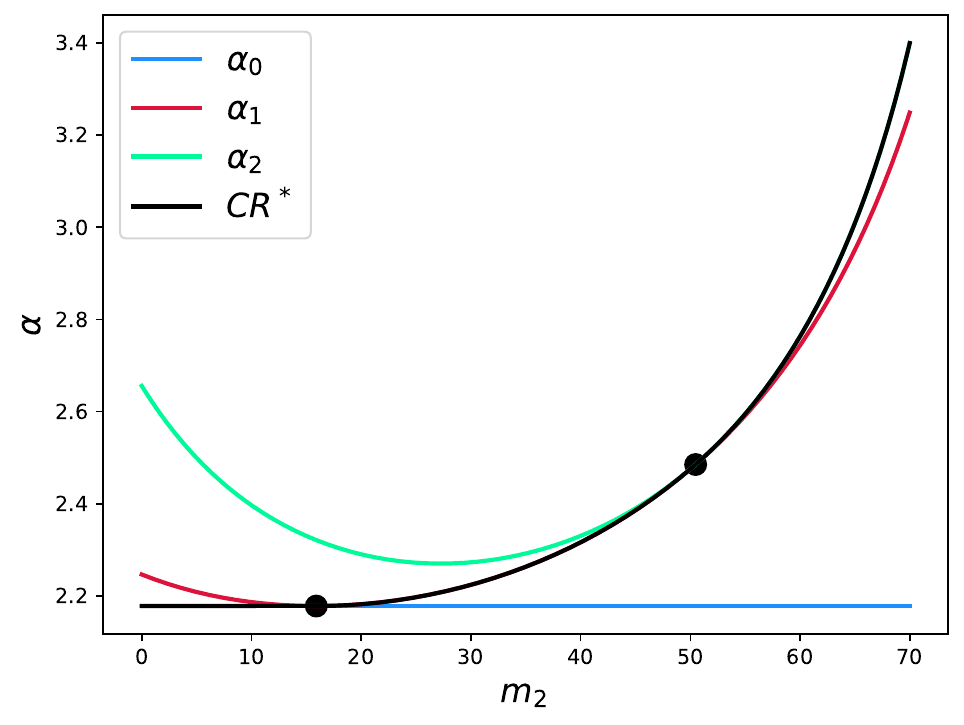}
    \caption{Illustrating the competitive ratio as a function of $m_2$. $\CR^*$ denotes the optimal competitive ratio.}
    \label{fig:GFQ_alphas}
\end{figure}

\section{Proofs of Section \ref{section-proportional-fairness}}\label{appendix-proportional-fairness}
\subsection{Proof of Theorem \ref{theorem-proportional-fairness-beta}}\label{appendix-proportional-fairness-beta}
Our goal is to determine the dynamics for the increasing threshold function $\phi_j(u_t)$ in Algorithm \ref{alg-proportional-fairness} to achieve $\beta$-\PF. For this, consider the function $\Upsilon_j(v)$, defined for any $v \in [1, \theta_j]$, as
\begin{align*}
    \Upsilon_j(v) = \argmax_{a \geq 0} \left( a \cdot v - \int_{0}^{a} \phi_j(u) \, du \right).
\end{align*}
Then for any instance $I$ with at least one arrival and set $ A$ be the set of classes that has at list one arrival. Let $\mathbf{v}$ be the vector of the maximum received valuation from each class that has at least one arrival. The utility of the classes in $j\in A$ is therefore lower-bounded as
\begin{align*}
    U_j(\mathbf{x}) \geq &\Upsilon_j(1)+ \int_{\Upsilon_j(1)}^{\Upsilon_j(v_j)} \phi_j(u) du \\
    =& \Upsilon_j(1)+ \frac{B}{K\cdot \beta}\exp\left(\frac{K\cdot\beta\cdot u}{B} -1\right)\bigg\rvert_{\Upsilon_j(1)}^{\Upsilon_j(v_j)}\\
    = & \frac{B}{\sum_i\alpha_i}+ \frac{B}{K\cdot \beta}\left(v_j-1\right)\\
    = & \frac{B}{K\cdot \beta} v_j, \quad \forall j\in A.
\end{align*}
On the other hand the utility of any allocation is upper-bounded as
\begin{align*}
    U_j(\mathbf{w}) \leq v_j\cdot w_j, \quad \forall j\in A,
\end{align*}
where $\sum_{j\in A}w_j \leq B$. Therefore:
\begin{align*}
    \frac{1}{K}\sum\nolimits_{j\in A}\frac{U_j(\mathbf{w})}{U_j(\mathbf{x})} \leq \frac{1}{K}\sum\nolimits_{j\in A}\frac{v_j\cdot w_j}{\frac{B}{K\cdot \beta} v_j} \leq \beta\sum\nolimits_{j\in A}\frac{w_j}{B} \leq \beta.
\end{align*}

\subsection{Proof of Theorem \ref{theorem-proportional-fairness-lower-bound}}\label{appendix-proportional-fairness-lower-bound-fairness}
Here we first define the hard instance that we considered.

\begin{lettereddef}[$\beta$-\PF Fairness Guarantee Hard Instance: $I^{\PF}$] \label{def_instance_I_p_pf}
    Instance $I^{\PF}$ is defined as a scenario characterized by a at most $K$ continuous, non-decreasing sequence of valuation arrivals segments. In this scenario, first there are a sequence of arrivals from class $1$, followed by the second sequence of arrivals all from class $2$ and and this continues until the arrivals of class $K$. For some value of $\epsilon$ such that $\epsilon \rightarrow 0$, instance $I^{\PF}$ can be shown as follows:
    \begin{align*}
        I^{\PF} = \Biggl\{\underbrace{(1,1), (1+\epsilon, 1), \dots , (\theta_1, 1)}_{\text{First batch of arrivals}}, \underbrace{(1, 2), (1+\epsilon,2), \dots , (\theta_2,2)}_{\text{Second batch of arrivals}}, \dots,\underbrace{(1,K), (1+\epsilon,K), \dots , (\theta_K,K)}_{K\text{-th batch of arrivals}}  \Biggr\},
    \end{align*}
    where in above $(v, j)$, $\forall j \in [k]$, corresponds to a buyer with valuation equal to $v$ from class $j$.
\end{lettereddef}
\begin{lettereddef}
     A utilization function $\psi_j(v):[1, \theta_j] \to [0, b_j],~\forall j\in[K]$  is defined as the final utilization of the effective budget $b_j$ after executing the instance $I^{\PF}$ by an online algorithm, where $\sum_{j\in[K]} b_j\leq B$.
\end{lettereddef}
\begin{letteredprop}
    \label{lemma-proportionally-fair}
    If there exists an $\beta$-proportionally fair online algorithm, there must exist a utilization function $\psi_j(u):[1,\theta_j]\to [0,b_j]$ for each class $j\in[K]$ such that $\psi_j(\cdot)$ is a non-decreasing function and satisfies
    \begin{align*}
        \begin{cases}
            \frac{1}{K} \sum_{j=1}^K \frac{v_j w_j}{\Psi_j(v_j)} \leq \beta,\\
            \sum_{j=1}^K w_j = B,\\
            \frac{B}{\beta}\leq \sum_{j=1}^K\psi_j(1),\\
            \sum_{j=1}^K \psi_j(\theta_j) \leq B,
        \end{cases}
    \end{align*}
    where $\Psi_j(v_j) = \psi_j(1) + \int_1^{v_j} u d \psi_j(u)$.
\end{letteredprop}
\begin{proof}
    Since the maximum utilization is $1$, the utilization function must satisfy the boundary condition $\sum_{j=1}^K \psi_j(\theta_j) \leq B$. Additionally, under an instance where all arrivals from each class $j\in [K]$ have the valuations of $1$ we can see that:
    \begin{align*}
        U_j(\mathbf{x}) \geq \psi_j(1)\quad\forall j \in [K],\\
        U_j(\mathbf{w}) \leq B \quad\forall j \in [K].
    \end{align*}
    Based on the definition of a $\beta$-proportional fair algorithms we must have:
    \begin{align*}
        \frac{1}{K} \sum_{j=1}^K \frac{U_j(\mathbf{w})}{U_j(\mathbf{x})} \leq \frac{1}{K}\cdot \frac{B}{\psi_j(1)} \leq \beta.
    \end{align*}
    By summation over all $j$'s we can see that:
    \begin{align*}
        \frac{B}{\beta}\leq \sum_{j=1}^K \psi_j(1).
    \end{align*}
    Also the utility of each class is lower bounded by
    \begin{align*}
         &U_j(\mathbf{x}) \geq \psi_j(1) + \int_1^{v_j} u d \psi_j(u) = \Psi_j(v_j).
    \end{align*}
    Additionally, in the worst case, allocation $\mathbf{w}$ will just select the highest valuation of each class, which implies $U_j(\mathbf{w}) \leq v_j w_j$. As a result, in the worst case for any online algorithm we can see
    \begin{align*}
         &\frac{1}{K} \sum_{j=1}^K \frac{U_j(\mathbf{w})}{U_j(\mathbf{x})} \leq \frac{1}{K} \sum_{j=1}^K \frac{v_j w_j}{\Psi_j(v_j)}.
    \end{align*}
    Since the algorithm is $\beta$-proportional fair, in the worst case it also needs to be $\beta$-proportionally fair. Therefore
    \begin{align*}
         & \frac{1}{K} \sum_{j=1}^K \frac{v_j w_j}{\Psi_j(v_j)} \leq \beta.
     \end{align*}
     We thus complete the proof of Proposition \ref{lemma-proportionally-fair}. 
\end{proof}
Let us first consider the case where the valuation \( v \) in instance \( I^{\PF} \) belongs to class 1. Since all arrivals are from this class, the proportional-fairness condition simplifies to:
\begin{align*}
    &\frac{1}{K}\frac{v\cdot w_1}{\Psi_1(v)} \leq \frac{1}{K}\frac{v\cdot B}{\Psi_1(v)}\leq \beta.
\end{align*}
Therefore
\begin{align*}
    & \frac{v\cdot B}{\beta\cdot K } \leq \psi_1(1) + \int_1^{v} u d \psi_1(u).
\end{align*}
Applying Gronwall's inequality and setting \( v = \theta_1 \), we obtain
\begin{align*}
    \frac{B}{K\cdot\beta}\alpha_1 \leq \psi_1(\theta_1).
\end{align*}
Next, assume that \( v \) belongs to class 2. The \(\beta\)-\PF condition now reduces to 
\begin{align*}
    &\frac{1}{K}\left(\frac{\theta_1\cdot w_1}{\Psi_1(\theta_1)} + \frac{v\cdot w_2}{\Psi_2(v)}\right) \leq \beta.
\end{align*}
By the definition of proportional fairness in Definition \ref{def_proportional_fairness}, any allocation \( \mathbf{w} \) must satisfy \(\beta\)-\PF. In particular, if the entire resource is allocated to class 2, the condition remains valid, i.e.,
\begin{align*}
    &\frac{1}{K}\frac{v\cdot B}{\Psi_2(v)}\leq \beta,
\end{align*}
which implies that
\begin{align*}
    & \frac{v\cdot B}{\beta\cdot K } \leq \psi_2(1) + \int_1^{v} u d \psi_i(u).
\end{align*}
Again applying Gronwall's inequality and substituting \( v = \theta_2 \), we derive
\begin{align*}
    \frac{B}{K\cdot\beta}\alpha_2 \leq \psi_2(\theta_2).
\end{align*}
Continuing the same approach we will get $\frac{B}{K\cdot \beta}\alpha_j \leq \psi_j(\theta_j)$. Therefore by summation over all $j\in[K]$ we will get
\begin{align*}
    &\sum\nolimits_{j\in[K]} \frac{B}{K\cdot \beta}\alpha_j \leq \sum\nolimits_{j\in[K]} \psi_j(\theta_j) \leq B.
\end{align*}
As a result
\begin{align*}
    \frac{1}{K}\sum\nolimits_{j\in[K]}\alpha_j \leq \beta.
\end{align*}
This completes the proof of Theorem \ref{theorem-proportional-fairness-lower-bound} and shows the optimality of the algorithm design presented in Theorem \ref{theorem-proportional-fairness-beta}.

\subsection{Proof of Theorem \ref{theorem-proportional-fairness-trade-off-design}} \label{appendix-theorem-proportional-fairness-trade-off-design}
The design of increasing threshold functions $\phi_j(\cdot)$, $j \in [K]$, follows the same design as  Theorem \ref{theorem-proportional-fairness-beta}. Consequently, we can argue that Algorithm \ref{theorem-proportional-fairness-beta} is  $\beta$-\PF for a given $\beta \in [\frac{1}{K} \sum_{j \in [K]} \alpha_{j},\infty]$. Next, let us find the competitive ratio of the Algorithm \ref{alg-proportional-fairness-tradeoff-local-globol} following the threshold design done in Theorem \ref{theorem-proportional-fairness-trade-off-design}. Consider an input instance $I$ such that from each class $j \in [k]$ the agent with maximum valuation has value $v_{i}$. Then, $\OPT(I)$, the optimal offline objective on input instance $I$, is upper-bounded as follows:
\begin{align*}
    \OPT(I) \leq \max_{j\in [K]} \{B\cdot v_j\}.
\end{align*}
Let $\ALG(I)$ denote the objective of the Algorithm \ref{alg-proportional-fairness-tradeoff-local-globol} on the input instance $I$. It follows that: 
\begin{align*}
    \ALG(I) \geq \left( \sum_{j \in [K]}  \psi_{j}(1) + \int_{\eta = 1}^{v_{j}} \eta \cdot d\psi_{j}(\eta) \right) + \psi^{G}(1) + \int_{\eta=1}^{\max_{j \in [K]}{\{v_{j}\}}} \eta \cdot d\psi^{G}(\eta),
\end{align*}
where in above functions $\psi_{j}:[1,\theta_{j}]\rightarrow [0,\frac{B\cdot \alpha_{j}}{K \cdot \beta}]$, $\forall j \in [K]$, are defined as $\psi_{j}(\eta) = \max{\{x \geq 0 | \phi_{j}(x) \leq \eta\}}$ and function $\psi^{G}:[1,\theta_{K}]\to [0,B]$  defined as $\psi_{G}(\eta) = \max{\{x \geq 0 | \phi^{G}(x) \leq \eta\}}$. 
Based on the design of threshold functions $\phi_{j}$ and $\phi^{G}$ done in Theorem \ref{theorem-proportional-fairness-trade-off-design}, it follows that:
\begin{align*}
    \ALG(I) & \ge   \psi_{j}(1) + \int_{\eta = 1}^{v_{j}} \eta \cdot d\psi_{j}(\eta) + \psi^{G}(1) + \int_{\eta=1}^{v_{j}} \eta \cdot d\psi^{G}(\eta) \\
    & \ge  \frac{B}{K\cdot \beta} + \int_{\eta = 1}^{v_{j}} \frac{B}{K \cdot \beta} d\eta + B\cdot \frac{1-\frac{\sum_{j \in [K]} \alpha_{j}}{K\cdot\beta}}{\alpha_{K}} + B\cdot \int_{\eta=1}^{v_j}  \frac{1-\frac{\sum_{j \in [K]} \alpha_{j}}{K\cdot\beta}}{\alpha_{K}} d\eta \\
    & = v_j \cdot B \cdot  \left( \frac{1}{K\cdot \beta} + \frac{1-\frac{\sum_{j \in [K]} \alpha_{j}}{K\cdot\beta}}{\alpha_{K}} \right) = v_j \cdot B \cdot \frac{1-\frac{\sum_{j \in [K-1]} \alpha_{j}}{K\cdot\beta}}{\alpha_{K}},
\end{align*}
where in above the second inequality follows from the definition of $\phi_{j}$ functions. Consequently, it follows that:
\begin{align*}
    \ALG(I) \ge \max_{j \in [K]}{\{B \cdot v_j\}} \cdot \frac{1-\frac{\sum_{j \in [K-1]} \alpha_{j}}{K\cdot\beta}}{\alpha_{K}} \ge \OPT(I) \cdot \frac{1-\frac{\sum_{j \in [K-1]} \alpha_{j}}{K\cdot\beta}}{\alpha_{K}}.
\end{align*}
As a result, we can see the that the design presented in Theorem \ref{theorem-proportional-fairness-trade-off-design} is $\left(\frac{1-\frac{\sum_{j \in [K-1]} \alpha_{j}}{K\cdot\beta}}{\alpha_{K}}\right)$ while it is $\beta$-\PF.

\subsection{Proof of Theorem \ref{theorem-proportional-fairness-alpha-lowerbound}}\label{appendix-proportional-fairness-alpha-lowerbound}
Here we again start with the define the hard instance that we considered.

\begin{lettereddef}[Efficiency-Fairness Trade-off for $\beta$-\PF Guarantee Hard Instance: $I^{to-\PF}$] \label{def_instance_I_p_pf_to}
    Instance $I^{to-\PF}$ is defined as a scenario characterized by a at most $K$ continuous, non-decreasing sequence of valuation arrivals segments. In this scenario, first there are a sequence of arrivals from class $K$, followed by the second sequence of arrivals all from class $K-1$ and this continues until the arrivals of class $1$. For some value of $\epsilon$ such that $\epsilon \rightarrow 0$, instance $I^{to-\PF}$ can be shown as follows:
    \begin{align*}
        I^{to-\PF} = \Biggl\{\underbrace{(1,K), (1+\epsilon, K), \dots , (\theta_K,K)}_{\text{First batch of arrivals}},& \underbrace{(1, K-1), (1+\epsilon, K-1), \dots , (\theta_{K-1}, K-1)}_{\text{Second batch of arrivals}},\\& \dots,\underbrace{(1,1), (1+\epsilon,1), \dots , (\theta_1,1)}_{K\text{-th batch of arrivals}}  \Biggr\},
    \end{align*}
    where in above $(v, j)$, $\forall j \in [K]$, corresponds to a buyer with valuation equal to $v$ from class $j$.
\end{lettereddef}
Now, based on the Proposition \ref{lemma-proportionally-fair}, in order to guarantee the $\beta$-proportional fairness, it is necessary to have:
\begin{align*}
    &b_1\geq \psi_1(\theta_1) \geq B\cdot\frac{1 + \ln \theta_1}{K\cdot \beta} = \frac{B\cdot \alpha_1}{K\cdot \beta },\\
    &b_2\geq \psi_2(\theta_2) \geq B\cdot\frac{1 + \ln \theta_2}{K\cdot \beta} = \frac{B\cdot\alpha_2}{K\cdot \beta},\\
    &\vdots \\
    &b_K\geq \psi_K(\theta_K) \geq B\cdot\frac{1 + \ln \theta_K}{K\cdot \beta} = \frac{B\cdot\alpha_K}{K\cdot \beta}.
\end{align*}
Therefore, the remaining portion of the resource is $B-\sum_{i\in[K]}\frac{B\cdot\alpha_i}{K\cdot \beta}$. Thus in order to guarantee the $\alpha$-competitiveness, there must exist a global utilization function $\psi^G(\cdot):[1,\theta_K]\to [0,B-\sum_{i\in[K]}\frac{B\cdot\alpha_i}{K\cdot \beta}]$ such that
\begin{align*}
    \frac{B}{\alpha}\max\nolimits_j\{v_j\}\leq \sum\nolimits_{i\in[K]}\Psi_i(v_i) + \Psi^G(\max\nolimits_j\{v_j\}),
\end{align*}
where $\Psi_i(v) = \psi_i(1)+\int_1^v \eta d \psi_i(\eta)$ for all $i\in[K]$ and $\Psi^G(v) = \psi^G(1)+\int_1^v \eta d \psi^G(\eta)$.
Now consider a case that all the arrivals are from class $K$. Therefore
\begin{align*}
    &\psi_K(1) + \int_1^{v_K} \eta ~d \psi_K(\eta) + \psi^G(1) + \int_1^{v_K} \eta ~d \psi^G(\eta) \geq \frac{B}{\alpha}v_K,\\
    &v_K\cdot \psi_K(v_K) - \int_1^{v_K}\psi_i(\eta)d\eta + v_K\cdot \psi^G(v_K) - \int_1^{v_K}\psi^G(\eta)d\eta \geq \frac{B}{\alpha}v_K.
\end{align*}
Again by using the Gronwall's inequality and substituting $v_K$ by $\theta_K$ we will get
\begin{align*}
    &\psi_K(\theta_K) + \psi^G(\theta_K) \geq \frac{B}{\alpha}(1+\ln\theta_K).
\end{align*}
As a result
\begin{align*}
    \frac{B\cdot\alpha_K}{K\cdot \beta } + B - \sum\nolimits_{j\in[K]} \frac{B\cdot\alpha_j}{K\cdot \beta} \geq \frac{B}{\alpha} \alpha_K.
\end{align*}
And finally we get
\begin{align*}
    \alpha\geq \frac{\alpha_K}{1 - \frac{\sum\nolimits_{j\in[K-1]} \alpha_j}{K\cdot \beta}}.
\end{align*}

\begin{letteredcor}[$\beta$-\PF Competitive Ratio Upper-Bound]\label{theorem-proportional-fariness-alpha}
    Algorithm \ref{alg-proportional-fairness} is (\(\sum_{j\in[K]} \alpha_j\))-competitive and (\(\frac{1}{K}\sum_{j\in[K]} \alpha_j\))-\PF with the design of the threshold functions described in Theorem \ref{theorem-proportional-fairness-beta}.
\end{letteredcor}
\begin{proof}
    Based on the design of increasing threshold functions $\phi_j(u_t)$ from Theorem \ref{theorem-proportional-fairness-beta}, Algorithm \ref{alg-proportional-fairness} is $\beta$-\PF. Now here we aim to find the competitive ratio corresponding this threshold function design. For this, consider the function $\Upsilon_j(v)$, defined for any $v \in [1, \theta_j]$, as
    \begin{align*}
        \Upsilon_j(v) = \argmax_{a \geq 0} \left( a \cdot v - \int_{0}^{a} \phi_j(b) \, db \right).
    \end{align*}
    Then for any instance $I$ with at least one arrival and set $ A$ be the set of classes that has at list one arrival let $\mathbf{v}$ be the vector of the maximum received valuation from each class that has at least one arrival. Then, $\OPT(I)$ can be upper-bounded as
    \begin{align*}
        \OPT(I) \leq \max_{i\in A} \{B\cdot v_i\}.
    \end{align*}
    
    Additionally, the performance of Algorithm \ref{alg-proportional-fairness}  can be lower-bounded as 
    \begin{align*}
        \ALG(I) \geq& \sum\nolimits_{j\in A} \left[\Upsilon_j(1)+ \int_{\Upsilon_j(1)}^{\Upsilon_j(v_j)} \phi(\eta) d\eta\right] \\
        =& \sum\nolimits_{j\in A} \left[\Upsilon_j(1)+ \frac{B}{K\cdot \beta}\exp\left(\frac{K\cdot \beta\cdot \eta}{B}-1\right)\bigg\rvert_{\Upsilon_j(1)}^{\Upsilon_j(v_j)}\right]\\
        = & \sum\nolimits_{j\in A} \left[\frac{B}{K\cdot \beta} v_j\right].
    \end{align*}
    Now, we can see that in the worst case, all the arrivals are from only form one class (let us say class $i$). Then $\OPT(I)\leq B\cdot v_i$ and $\ALG(I)\geq \frac{B}{K\cdot \beta }v_i$ and therefore
    \begin{align*}
        \frac{\OPT(I)}{\ALG(I)}\leq \frac{B\cdot v_i}{v_i\frac{B}{K\cdot \beta}} = K\cdot \beta = \sum\nolimits_{i\in[K]}\alpha_i.
    \end{align*}
    This concludes the proof and shows that any $\left(\frac{\sum_{i\in[K]}\alpha_i}{K}\right)-\PF$ algorithm is indeed $\left(\sum_{i\in[K]} \alpha_i\right)$-competitive which also aligns with the results of Theorem \ref{theorem-proportional-fairness-trade-off-design}.
\end{proof}

\section{Proofs of Section \ref{section-gamma-beta-fairness}}\label{appendix-gamma-beta-fairness}
\subsection{Proof of Theorem \ref{theorem-gamma-beta-fairness}}\label{appendix-theorem-gamma-beta-fairness}
Before proving Theorem \ref{theorem-gamma-beta-fairness}, we first present and prove a key proposition that will be essential for the proofs in this section.
\begin{letteredprop}\label{appendix-proposition-U(w)-upperbound}
    For any allocation $\mathbf{w}$, the following inequality always holds:
        \begin{align*}
            \frac{1}{1-\gamma}\sum\nolimits_{j\in[K]}U_j(\mathbf{w})^{1-\gamma} \leq \frac{1}{1-\gamma}\left(\sum\nolimits_{j\in[K]} (B\cdot v_j)^\frac{1-\gamma}{\gamma}\right)^\gamma.
        \end{align*}
    Moreover, the allocation and utility of each class $j$ under the maximizing allocation $\mathbf{w}$ are:
    \begin{align*}
        w_j = \frac{B \cdot v_j^\frac{1-\gamma}{\gamma}}{\sum\nolimits_{i\in[K]} v_i^\frac{1-\gamma}{\gamma}}, \qquad U_j(\mathbf{w}) = v_j\cdot w_j = v_j\cdot B \cdot\frac{v_j^\frac{1-\gamma}{\gamma}}{\sum\nolimits_{i\in[K]} v_i^\frac{1-\gamma}{\gamma}}.
    \end{align*}
\end{letteredprop}
\begin{proof}
    Let us first consider $\gamma \in (0,1)$. In this case we can see that in the worst-case allocation, $\mathbf{w}$ will only select the highest arriving valuation from each class. Therefore $U_j(\mathbf{w}) \leq v_j\cdot w_j$ for all $j\in[K]$ where $\sum\nolimits_{i\in[K]} w_j \leq B$. Now we want to find an upper bound for this allocation. Due to the Hölder's inequality 
    \begin{align*}
        \sum\nolimits_{i\in[K]} {(v_i\cdot w_i)}^{1-\gamma}\leq \left(\sum\nolimits_{i\in[K]} \left(w_i^{1-\gamma} \right)^{\frac{1}{1-\gamma }}\right)^{1-\gamma} \left(\sum\nolimits_{i\in[K]} \left(v_i^{1-\gamma}\right)^\frac{1}{\gamma}\right)^{\gamma}.
    \end{align*}
    And since $\sum_i w_i \leq B$ we can see that
    \begin{align*}
        \sum\nolimits_{i\in[K]} {(v_i\cdot w_i)}^\frac{1-\gamma}{\gamma} \leq B^{1-\gamma} \cdot  \left(\sum\nolimits_{i\in[K]} \left(v_i^{1-\gamma}\right)^\frac{1}{\gamma}\right)^{\gamma}.
    \end{align*}
    And as $\gamma\in(0,1)$, it is clear that
    \begin{align*}
         \frac{1}{1-\gamma}\sum\nolimits_{j\in[K]}U_j(\mathbf{w})^{1-\gamma} \leq \frac{1}{1-\gamma}\left(\sum\nolimits_{j\in[K]} (B\cdot v_j)^\frac{1-\gamma}{\gamma}\right)^\gamma.
    \end{align*}
    In the case $\gamma \in (1, \infty)$ we can use the same approach:
    \begin{align*}
        \sum\nolimits_{i\in[K]} U_i(\mathbf{w})^{1-\gamma} \geq \beta^{1-\gamma} \sum\nolimits_{i\in[K]} U_i(\mathbf{x})^{1-\gamma}.
    \end{align*}
    Since $1-\gamma \leq 0$, the Höilders inequality turns to:
    \begin{align*}
        & \sum\nolimits_{i\in[K]} v_i^\frac{1-\gamma}{\gamma}\leq \left(\sum\nolimits_{i\in[K]} \left((v_i w_i)^{\frac{1-\gamma}{\gamma}}\right)^\gamma\right)^\frac{1}{\gamma}\left(\sum\nolimits_{i\in[K]} \left(w_i^\frac{\gamma-1}{\gamma}\right)^{\frac{\gamma}{\gamma-1}}\right)^{\frac{\gamma-1}{\gamma}}.
    \end{align*}
    By simplifying it we will further get
    \begin{align*}
        \left(\sum\nolimits_{i\in[K]} v_i^\frac{1-\gamma}{\gamma}\right)^\gamma \leq \sum\nolimits_{i\in[K]} (v_i w_i)^{1-\gamma}\left(\sum\nolimits_{i\in[K]} w_i\right)^{\gamma-1}.
    \end{align*}
    Since $\sum_i w_i \leq B$ based on the budget constraint, we will further get
    \begin{align*}
        \left(\sum\nolimits_{i\in[K]} v_i^\frac{1-\gamma}{\gamma}\right)^\gamma \leq \sum\nolimits_{i\in[K]} (v_i w_i)^{1-\gamma}\cdot B^{\gamma-1}.
    \end{align*}
    And again since $\gamma\in(0,1)$, we will get 
    \begin{align*}
         \frac{1}{1-\gamma}\sum\nolimits_{j\in[K]}U_j(\mathbf{w})^{1-\gamma} \leq \frac{1}{1-\gamma}\left(\sum\nolimits_{j\in[K]} (B\cdot v_j)^\frac{1-\gamma}{\gamma}\right)^\gamma.
    \end{align*}
    Finally, based on the equality condition of Hölder's inequality, we can easily verify that equality holds when:
     \begin{align*}
        w_j = \frac{B \cdot v_j^\frac{1-\gamma}{\gamma}}{\sum\nolimits_{i\in[K]} v_i^\frac{1-\gamma}{\gamma}}, \qquad U_j(\mathbf{w}) = v_j\cdot w_j = v_j\cdot B \cdot \frac{v_j^\frac{1-\gamma}{\gamma}}{\sum\nolimits_{i\in[K]} v_i^\frac{1-\gamma}{\gamma}}.
    \end{align*}
\end{proof}

Next we present Proposition \ref{proposition-modified-gamma-beta-thresholds} which states that the threshold function design in Theorem \ref{theorem-gamma-beta-fairness} ensures the condition $U_i(\mathbf{w}) \leq \beta_i \cdot U_i(\mathbf{x})$, where $\beta_i$ is a constant greater than 1.
\begin{letteredprop}\label{proposition-modified-gamma-beta-thresholds}
     Based on the threshold functions designed from Theorem \ref{theorem-gamma-beta-fairness}, inequality $U_i(\mathbf{w}) \leq \beta_i \cdot U_i(\mathbf{x})$  for all $i \in [K]$ holds for any allocation $\mathbf{w}$.
\end{letteredprop}
\begin{proof}
    Let us consider an instance $I$ such that the maximum valuation arrived from each class be $v_j$ for all $i\in[K]$. Under this instance we can see that
    \begin{align*}
       U_i(\mathbf{x}) &\geq \psi_i(1) + \int_1^{v_i} \eta d\psi_i(\eta).
    \end{align*}
    Now let us start with the $\gamma>1$ case. Using the definition of $\psi_1(\cdot)$ function in Theorem \ref{theorem-gamma-beta-fairness} we can see that:
    \begin{align*}
        \psi_i(1) + &\int_1^{v_i} \eta d\psi_i(\eta) = v_i\psi_1(v_i) - \int_1^{v_i} \psi_i(\eta)d\eta\\
        &=\frac{B}{\beta_i}\Biggl[v_i\frac{1}{\sum_{j\in[i^-]} \left(
        \frac{v_i}{\theta_j}\right)^{\frac{\gamma-1}{\gamma}} + 1} + \frac{\gamma}{\gamma-1} v_i \ln \left(\frac{v_i^{\frac{\gamma-1}{\gamma}}\left(\sum_{j\in[i^-]} \left(
        \frac{1}{\theta_j}\right)^{\frac{\gamma-1}{\gamma}} + 1\right)}{\sum_{j\in[i^-]} \left(
        \frac{v_i}{\theta_j}\right)^{\frac{\gamma-1}{\gamma}} + 1}\right) - \\
        &\qquad \frac{\gamma}{\gamma-1} u \ln \left(\frac{u^{\frac{\gamma-1}{\gamma}}\left(\sum_{j\in[i^-]} \left(
        \frac{1}{\theta_j}\right)^{\frac{\gamma-1}{\gamma}} + 1\right)}{\sum_{j\in[i^-]} \left(
        \frac{u}{\theta_j}\right)^{\frac{\gamma-1}{\gamma}} + 1}\right)\bigg\vert_1^{v_i}\Biggr]\\
        &=\frac{B}{\beta_i}\left[\frac{v_i}{\sum_{j\in[i^-]} \left(
        \frac{v_i}{\theta_j}\right)^{\frac{\gamma-1}{\gamma}} + 1}\right].
    \end{align*}
    Let $\mathbf{v}$ be a $K$ dimensional vector where index $i$ is always zero. Let us consider function $g(\mathbf{v}) \coloneqq \frac{v_i}{\sum_{j\in[i^-]} \left(\frac{v_i}{\theta_j}\right)^{\frac{\gamma-1}{\gamma}} + 1}$. We can see that $\frac{\partial}{\partial v_j}g(\mathbf{v}) \geq 0$ for all $\gamma \geq 1$ and all $j\in[K]$. Therefore we have:
    \begin{align*}
    &\frac{v_i}{\sum_{j\in[i^-]} \left(
            \frac{v_i}{\theta_j}\right)^{\frac{\gamma-1}{\gamma}} + 1} \geq  \frac{v_i}{\sum_{j\in[i^-]} \left(
            \frac{v_i}{v_j}\right)^{\frac{\gamma-1}{\gamma}} + 1} = v_i\frac{v_i^\frac{1-\gamma}{\gamma}}{\sum\nolimits_{j\in[K]} v_j^\frac{1-\gamma}{\gamma}}.
    \end{align*}
    Therefore 
    \begin{align*}
        &U_i(\mathbf{x}) \geq \frac{B}{\beta_i}\left[v_i\frac{v_i^\frac{1-\gamma}{\gamma}}{\sum\nolimits_{j\in[K]} v_j^\frac{1-\gamma}{\gamma}}\right].
    \end{align*}
    As a result
    \begin{align*}
        U_i(\mathbf{x})\geq \frac{B}{\beta_i}U_i(\mathbf{w}), \quad \forall i\in[K].
    \end{align*}
    The last inequality is based on Proposition \ref{appendix-proposition-U(w)-upperbound}. In the case of \(\gamma \in (0,1)\), the proof remains exactly the same. The only difference is that \(\frac{\partial}{\partial v_i} g(\mathbf{v})\) is always negative in this case. The rest of the proof is straightforward. Also following the same approach for the case $\gamma=1$ based on the definition of $\psi$ function we see that
    \begin{align*}
        U_i(\mathbf{x})\geq \Psi_i(v_i) \geq \frac{B}{\beta_j}\cdot\frac{v_i}{K} \geq U_i(\mathbf{w}).
    \end{align*}
    This completes the proof of Proposition \ref{proposition-modified-gamma-beta-thresholds} for all choices of $\gamma$.
\end{proof}

Next, we turn to Proposition \ref{proposition-gamma-beta-min-max}, which outlines the design of the parameters $\beta_i$ for each $i \in [K]$. These parameters are essential for balancing utility and fairness in the algorithm, ensuring that the fairness constraints hold under varying conditions. The proposition presents a method for determining the appropriate values of $\beta_i$. 

\begin{letteredprop}\label{proposition-gamma-beta-min-max}
    Given a vector $\left(\beta_{i}\right)_{i \in [K]} \in [1,\infty]^{K}$, the inequality
    \begin{align*}
    \sum\nolimits_{i\in[K]}f\left(\frac{U_i(\mathbf{w})}{\bar{\beta}}\right)\leq  \sum\nolimits_{i\in[K]} f\left(\frac{U_i(\mathbf{w})}{\beta_i}\right)
    \end{align*}
    holds for all the following values of $\hat{\beta}$:
    \begin{align}
        \bar{\beta} \ge \begin{cases}
            \max\limits_{v_j \in \{1, \theta_j\}, \forall j} \left\{\left(\frac{\sum\nolimits_{j\in[K]} \beta_j^{\gamma-1}\cdot v_j^\frac{{1-\gamma}}{\gamma}}{\sum\nolimits_{j\in[K]} v_j^\frac{{1-\gamma}}{\gamma}}\right)^{\frac{1}{\gamma-1}}\right\} \quad & \text{if }~ \gamma \neq 1,  \\
            \\
            \prod\nolimits_{j\in[K]} \beta_j^{1/K}\  & \text{if } ~ \gamma =1.\end{cases}\label{eq-optimization-find-beta-values}
    \end{align}
\end{letteredprop}

\begin{proof}
    For any $\gamma\neq 1$, let $\mathbf{v}^*$ be the maximize vector of the right hand side of inequality \eqref{eq-optimization-find-beta-values}. Then we can see that 
    \begin{align*}
        \bar{\beta} \ge  \left(\frac{\sum\nolimits_{i\in[K]} \beta_i^{\gamma-1}\cdot {v^*_i}^\frac{{1-\gamma}}{\gamma}}{\sum\nolimits_{i\in[K]} {v^*_i}^\frac{{1-\gamma}}{\gamma}}\right)^{\frac{1}{\gamma-1}}.
    \end{align*}
    Then we can easily see that
    \begin{align*}
        &\frac{1}{1-\gamma}\bar{\beta}^{\gamma-1}\sum\nolimits_{i\in[K]} {v^*_i}^\frac{{1-\gamma}}{\gamma} \leq \frac{1}{1-\gamma}\sum\nolimits_{i\in[K]} \beta_i^{\gamma-1}\cdot {v^*_i}^\frac{{1-\gamma}}{\gamma}.
    \end{align*}
    Now, based on Proposition \ref{appendix-proposition-U(w)-upperbound} we know that under the worst allocation $\mathbf{w}$ we have
    \begin{align*}
        U_i(\mathbf{w}) =  v^*_i\cdot B \cdot \frac{{v^*_i}^\frac{1-\gamma}{\gamma}}{\sum\nolimits_{j\in[K]} {v^*_j}^\frac{1-\gamma}{\gamma}} = \frac{B\cdot {v^*_i}^\frac{1}{\gamma}}{\sum\nolimits_{j\in[K]} {v^*_j}^\frac{1-\gamma}{\gamma}}.
    \end{align*}
    Therefore 
    \begin{align*}
         & \frac{1}{1-\gamma}\bar{\beta}^{\gamma-1}\sum\nolimits_{i\in[K]} \left(\frac{B\cdot {v^*_i}^\frac{1}{\gamma}}{\sum\nolimits_{j\in[K]} {v^*_j}^\frac{1-\gamma}{\gamma}}\right)^{1-\gamma} \leq \frac{1}{1-\gamma}\sum\nolimits_{i\in[K]} \beta_i^{\gamma-1}\cdot \left(\frac{B \cdot {v^*_i}^\frac{1}{\gamma}}{\sum\nolimits_{j\in[K]} {v^*_j}^\frac{1-\gamma}{\gamma}}\right)^{1-\gamma}.
    \end{align*}
    As a result
    \begin{align*}
          \frac{1}{1-\gamma}\left(\sum\nolimits_{i\in[K]}\left(\frac{U_i(\mathbf{w})}{\bar{\beta}}\right)^{1-\gamma}\right) \leq  \frac{1}{\gamma-1}\cdot \left(\sum\nolimits_{i\in[K]} \left(\frac{U_i(\mathbf{w})}{\beta_i}\right)^{1-\gamma}\right).
    \end{align*}
    Additionally, when $\gamma=1$, we can see 
    \begin{align*}
        \bar{\beta} \geq  \prod\nolimits_{j\in[K]} \beta_j^{1/K}.
    \end{align*}
    Therefore, we have
    \begin{align*}
        &\prod\nolimits_{j\in[K]} \frac{B\cdot v_j}{\bar{\beta}\cdot K} \leq \prod\nolimits_{j\in[K]} \frac{B\cdot v_j}{\beta_j\cdot K},
    \end{align*}
    which implies that
    \begin{align*}
        \sum\nolimits_{i\in[K]}\log\left(\frac{U_i(\mathbf{w})}{\bar{\beta}}\right)\leq  \sum\nolimits_{i\in[K]} \log\left(\frac{U_i(\mathbf{w})}{\beta_i}\right).
    \end{align*}
    Thus, we have
    \begin{align*}                          \sum\nolimits_{i\in[K]}f\left(\frac{U_i(\mathbf{w})}{\bar{\beta}}\right)\leq  \sum\nolimits_{i\in[K]} f\left(\frac{U_i(\mathbf{w})}{\beta_i}\right).
    \end{align*}
    This completes the proof of Proposition \ref{proposition-gamma-beta-min-max} in all the possible ranges of $\gamma$ value.
\end{proof}

Since based on Theorem \ref{theorem-gamma-beta-fairness} $\beta = \min\{\hat{\beta}\}$ subject to $\sum\nolimits_{j\in[K]} \psi_j(\theta_j) \leq B$, we always can always choose $\beta_{i}$, $i \in [K]$ in a way they minimize the right-hand-side of Eq. \eqref{eq-optimization-find-beta-values}.

Now we are ready to go through the proof of theorem. 
Based on Proposition \ref{proposition-modified-gamma-beta-thresholds} and Proposition \ref{proposition-gamma-beta-min-max} for $\gamma\neq1$ we will get:
\begin{align*}
    \frac{1}{1-\gamma}\left(\sum\nolimits_{i\in[K]} U_i(\mathbf{x})^{1-\gamma}\right) &\geq \frac{1}{1-\gamma} \left(\sum\nolimits_{i\in[K]}\left(\frac{U_i(\mathbf{w})}{\beta_i}\right)^{1-\gamma}\right)\\
    &\geq \frac{1}{1-\gamma}\left(\sum\nolimits_{i\in[K]}\left(\frac{U_i(\mathbf{w})}{\beta}\right)^{1-\gamma}\right),
\end{align*}
where the first inequality is resulted from  Proposition \ref{proposition-modified-gamma-beta-thresholds} and the second one from Proposition \ref{proposition-gamma-beta-min-max}. Also for $\gamma=1$ we will further get:
\begin{align*}
    \sum\nolimits_{i\in[K]} \log(U_i(\mathbf{x})) &\geq\sum\nolimits_{i\in[K]}\log\left(\frac{U_i(\mathbf{w})}{\beta_i}\right)\\
    &\geq \sum\nolimits_{i\in[K]}\log\left(\frac{U_i(\mathbf{w})}{\beta}\right),
\end{align*}
where again the first inequality comes from  Proposition \ref{proposition-modified-gamma-beta-thresholds} and the second one is based on Proposition \ref{proposition-gamma-beta-min-max}.

\subsection{Case Study $K=2$}
\label{gamma-beta-fairness-case-study}
To demonstrate the results of Theorem \ref{theorem-gamma-beta-fairness} in a more explicit way, we consider a case study of a simplest possible setting when there are only two class of arrivals. The results of Theorem \ref{theorem-gamma-beta-fairness} can be summarized in three cases as follows:

$\bullet$ For $\gamma \in (1, \infty)$, the solution to the minimax problem of Eq. \eqref{eq:minimax-1} is obtained when $ v_1 = \theta_1$ and $ v_2 = 1$. As a result
\begin{align*}
    \beta = \left(\frac{{\beta_1}^{\gamma-1} + \theta_1^{\frac{\gamma-1}{\gamma}} {\beta_2}^{\gamma-1}}{1 + \theta_1^{\frac{\gamma-1}{\gamma}}}\right)^{\frac{1}{\gamma-1}}.
\end{align*}
To compute $ \beta_1 $ and $ \beta_2 $, recall that by Theorem \ref{theorem-gamma-beta-fairness}, in the case where  $K=2$, we have:
\begin{align*}
    & F_1(\theta_1; \gamma) = \frac{\theta_2^{\frac{\gamma-1}{\gamma}}}{\theta_1^{\frac{\gamma-1}{\gamma}} + \theta_2^{\frac{\gamma-1}{\gamma}}} + \frac{\gamma}{\gamma-1} \ln \left(\frac{\theta_1^{\frac{\gamma-1}{\gamma}}(\theta_2^{\frac{\gamma-1}{\gamma}}+1)}{\theta_1^{\frac{\gamma-1}{\gamma}} + \theta_2^{\frac{\gamma-1}{\gamma}}}\right), \\ %\quad
    & F_2(\theta_2;\gamma) = \frac{\theta_1^{\frac{\gamma-1}{\gamma}}}{\theta_2^{\frac{\gamma-1}{\gamma}} + \theta_1^{\frac{\gamma-1}{\gamma}}} + \frac{\gamma}{\gamma-1} \ln \left(\frac{\theta_2^{\frac{\gamma-1}{\gamma}}(\theta_1^{\frac{\gamma-1}{\gamma}}+1)}{\theta_2^{\frac{\gamma-1}{\gamma}} + \theta_1^{\frac{\gamma-1}{\gamma}}}\right).
\end{align*}
Thus, when $\frac{F_2(\theta_2;\gamma)}{F_1(\theta_1; \gamma)} \geq \theta_1^{\frac{\gamma-1}{\gamma}}$,  $\beta_1$ and $\beta_2$ are:
\begin{align*}
    & \beta_1 = F_2(\theta_2;\gamma) \left(\theta_1^{\frac{\gamma-1}{\gamma}} \frac{F_1(\theta_1; \gamma)}{F_2(\theta_2;\gamma)}\right)^{\frac{1}{\gamma}} + F_1(\theta_1; \gamma), \\ % \quad 
    & \beta_2 = F_1(\theta_1; \gamma)\left(\theta_1^{\frac{1-\gamma}{\gamma}} \frac{F_2(\theta_2;\gamma)}{F_1(\theta_1; \gamma)}\right)^{\frac{1}{\gamma}} + F_2(\theta_2;\gamma).
\end{align*}
Otherwise, when $\frac{F_2(\theta_2;\gamma)}{F_1(\theta_1; \gamma)} < \theta_1^{\frac{\gamma-1}{\gamma}}$, we have $\beta_1 = \beta_2 = F_1(\theta_1;\gamma) + F_2(\theta_2;\gamma)$.

$\bullet$ For $\gamma \in (0, 1)$, the solution to the minimax problem of Eq. \eqref{eq:minimax-1} is obtained when $ v_1 = 1$ and $ v_2 = \theta_2$. As a result
\begin{align*}
    \beta = \left(\frac{{\beta_1}^{\gamma-1} + \theta_2^{\frac{1-\gamma}{\gamma}} {\beta_2}^{\gamma-1}}{1+\theta_2^{\frac{1-\gamma}{\gamma}}}\right)^{\frac{1}{\gamma-1}}.
\end{align*}
Similarly, by Theorem \ref{theorem-gamma-beta-fairness} with  $K=2$, we have:
\begin{align*}
    & F_1(\theta_1;\gamma) \coloneqq \frac{1}{\theta_1^{\frac{\gamma-1}{\gamma}} + 1} + \frac{\gamma}{\gamma-1} \ln \left(\frac{2 \cdot \theta_1^{\frac{\gamma-1}{\gamma}}}{\theta_1^{\frac{\gamma-1}{\gamma}} + 1}\right), \quad F_2(\theta_2;\gamma) \coloneqq \frac{1}{\theta_2^{\frac{\gamma-1}{\gamma}} + 1} + \frac{\gamma}{\gamma-1} \ln \left(\frac{2 \cdot \theta_2^{\frac{\gamma-1}{\gamma}}}{\theta_2^{\frac{\gamma-1}{\gamma}} + 1}\right).
\end{align*}
Thus, when $\frac{F_2(\theta_2;\gamma)}{F_1(\theta_1;\gamma)}\geq \theta_2^\frac{1-\gamma}{\gamma}$, $\beta_1$ and $\beta_2$ are 
    \begin{align*}
        \beta_1 = F_2(\theta_2;\gamma) \left(\theta_2^\frac{1-\gamma}{\gamma}\cdot \frac{F_1(\theta_1;\gamma)}{F_2(\theta_2;\gamma)}\right)^\frac{1}{\gamma} + F_1(\theta_1;\gamma),\quad \beta_2 = F_1(\theta_1;\gamma) \left(\theta_2^\frac{\gamma-1}{\gamma}\cdot \frac{F_2(\theta_2;\gamma)}{F_1(\theta_1;\gamma)}\right)^\frac{1}{\gamma} + F_2(\theta_2;\gamma).
    \end{align*}
Otherwise, when $\frac{F_2(\theta_2;\gamma)}{F_1(\theta_1;\gamma)} <  \theta_2^\frac{1-\gamma}{\gamma}$, we have $\beta_1=\beta_2=F_1(\theta_1;\gamma) + F_2(\theta_2;\gamma)$.

$\bullet$ For $\gamma=1$, 
Algorithm \ref{alg-proportional-fairness} is $(1, \beta)$-fair with $\beta=\sqrt{\beta_1\beta_2}$, where $ \beta_1 $ and $ \beta_2 $ are given by: 
\begin{align*}
    \beta_1=\alpha_1, \qquad \beta_2=\alpha_2.
\end{align*}

The result above explains how $\beta_i$ should be extracted from Eqs. \eqref{eq:minimax-1} and \eqref{eq:minimax-2} for each $\gamma$ and the associated $\beta$. Note that the $\beta_i$ values are indeed necessary for designing the threshold functions.

\subsection{Proof of Theorem \ref{theorem-gamma-beta-lowerbound}}\label{appendix-gamma-beta-lowerbound}
We derive the lower bound $\beta^{*}_{\gamma}$ by studying the performance of every online algorithm on the hard instance $I^\GBF$  defined as follows.

\begin{lettereddef}[($\gamma,\beta$)-Fairness Guarantee Hard Instance: $I^\GBF$]\label{definition-instance-I-GBF}
    The instance $I^\GBF$ begins with the arrival of the first batch, which includes one agent from each of the $K$ classes, each with a valuation of 1. Following these $K$ agents, a second batch arrives from class $K$, with valuations ranging within $(1, \theta_{K}]$ and an infinitesimal increment $ \epsilon > 0$. Subsequently, a third batch arrives from the $(K-1)$-th class, with a continuum of valuations within the range $(1, \theta_{K-1}]$. This pattern continues with a continuum of agents from the $(K-2)$-th class, progressing sequentially until it reaches a continuum of agents from the first class, with valuations within the range $(1, \theta_{1}]$. For an arbitrarily small value of $\epsilon > 0$, we can formally denote $I^{\GBF}$ as follows:
    \begin{align*}
        I^{\GBF} = \Biggl\{ & \underbrace{(1,1), (1,2), \dots , (1,K)}_{\text{1st batch of arrivals}}\ , \underbrace{ (1+\epsilon,K) ,  \dots  ,  (\theta_K, K)}_{\text{2nd batch of arrivals}}\ ,\   \\  
        & \underbrace{ (1+\epsilon, K-1),\ \dots\ , (\theta_{K-1}, K-1)}_{\text{3rd batch of arrivals}}\ ,\  \dots \ , \underbrace{ (1+\epsilon, 1), \dots ,(\theta_{1},1)}_{(K+1)\text{-th batch of arrivals}}  \Biggr\},
    \end{align*}
    where $ (v,j)$ corresponds to a buyer with valuation $v$ from class $ j $, $\forall j \in [K]$.
\end{lettereddef}

The performance of any online algorithm on the hard instance $I^\GBF$ can be captured by the set of $K$ utilization functions, $\zeta_{j}:[1,\theta_j] \rightarrow [0,B]$, for all $j \in [K]$. Let $\zeta_{j}(v)$, for all $j \in [K]$ and $v \in [1,\theta_{j}]$, represent the expected amount of resources allocated to agents from the $j$-th class with valuations up to $v$. Due to the online nature of the problem, any $(\gamma, \beta)$-fair online algorithm, \ALG, must ensure the $(\gamma, \beta)$-fairness guarantee at every stage of the instance $I^\GBF$. Therefore, as agents from the $K$-th class arrive (the second batch  with increasing valuations), \ALG must allocate increasing amounts of resources to these agents to maintain efficiency. Simultaneously, \ALG must allocate resources conservatively to preserve sufficient budget for future agents from other classes to ensure fair resource distribution across all classes. Consequently, by the end of the stage in instance $I^{\GBF}$, where buyers with valuation $v$ from class $j$ have arrived, we can establish a lower bound for the expected amount of resources allocated to class $j$ agents with valuations up to $v$.

Based on the definition of hard instance we presented in the Definition \ref{definition-instance-I-GBF}, let us define class of instances $\{I^{j}_{v}\}_{v \in [1,\theta_{j}]}$, $j \in [K]$, where $I^{j}_{v}$ is the instance that consists of agents of instance $I^\GBF$ from the beginning up to the agent from the $j$-th class with valuation $v$.  Instance $I^{j}_{v}$ include all the arriving agents from the the classes $j+1$ up to $K$ that arrive in $I$ before the agent with valuation $v$ from the $j$-th class. Now, here we will present a key lemma which is based on the Proposition \ref{appendix-proposition-U(w)-upperbound}.

\begin{letteredlemma}
    Let $\mathbf{w}^{*}$ be the optimal allocation vector at the end of the instance $I^{j}_{v}$ for some $j \in [K]$ and $v \in [1,\theta_{j}]$. Then it follows that:
    \begin{align*}
    U_{i}(w^{*}) =
    \begin{cases}
    \frac{B}{j-1 + v^{\frac{1-\gamma}{\gamma}} + \sum_{l=j+1}^{K} \theta_l^{\frac{1-\gamma}{\gamma}}} & i \in [1,j-1], \\
    \frac{B\cdot v^{\frac{1}{\gamma}}}{i-1 + v^{\frac{1-\gamma}{\gamma}} + \sum_{l=j+1}^{K} \theta_l^{\frac{1-\gamma}{\gamma}}} & i= j,\\
    \frac{B\cdot \theta_i^{\frac{1}{\gamma}}}{j-1 + v^{\frac{1-\gamma}{\gamma}} + \sum_{l=j+1}^{K} \theta_l^{\frac{1-\gamma}{\gamma}}} & j \in [j+1,K],
            \end{cases}
    \end{align*}
    and we will have:
    \begin{align*}
        \sum_{i=1}^{K} U_{i}(w^{*})^{1-\gamma} = B^{1-\gamma} \cdot (j-1 + v^{\frac{\gamma-1}{\gamma}} + \sum_{l=j+1}^{K} \theta_{l}^{\frac{\gamma-1}{\gamma}})^\gamma.
    \end{align*}
\end{letteredlemma}
\begin{proof}
    The proof of this lemma is directly resulted from Proposition \ref{appendix-proposition-U(w)-upperbound} by setting $v_i = 1$ for all $i\in[1, j-1]$ and $v_i = \theta_i$ for all $i\in[j+1, K]$.
\end{proof}

\begin{letteredlemma}\label{appendix-lemma-system-of-inequalities}
    For some value of $j \in [K]$, $v \in [1,\theta_j]$ and $\gamma \ge 0$, any $(\gamma,\beta)$-fair online algorithms given the hard instance $I^\GBF$ as input, by the end of stage in $I^j_v$ where the agent from class $j$ with valuation $v$ arrives needs to satisfy following system of differential inequalities such that:
    \begin{align*}
        & \frac{1}{1-\gamma}\left[   \left(\frac{V_{j}(v)}{\beta}\right)^{1-\gamma} \right] \geq \frac{1}{1-\gamma} \cdot \left ( B^{1-\gamma}\cdot \left(j-1 + v^{\frac{\gamma-1}{\gamma}} + \sum_{l=j+1}^{K} \theta_{l}^{\frac{\gamma-1}{\gamma}}\right)^\gamma - \sum_{i=1}^{j-1} \left(\frac{V_{i}(1)}{\beta}\right)^{1-\gamma}  - \sum_{i=j+1}^{K} \left(\frac{V_{i}(\theta_{i})}{\beta}\right)^{1-\gamma} \right),
    \end{align*}
    where $V_j(v) = \zeta_{j}(1) + \int_{1}^{v} \eta \cdot d\zeta_{j}(\eta)$, $j \in [K], v \in [1,\theta_{j}]$; $\zeta_{j}(v)$ is the expected amount of resource that is allocated by the online algorithm to buyers from the class $j$ with valuation at most equal to $v$.
\end{letteredlemma}

\begin{proof}
Let $x$ be the allocation  generated by the online algorithm. Based on the definition of ($\gamma, \beta$)-fairness, we should have:
    \begin{align*}
        \frac{1}{1-\gamma}\left(\sum\nolimits_{i\in[K]}U_i(\mathbf{x})^{1-\gamma}\right)  \geq \frac{1}{1-\gamma}\left(\sum\nolimits_{i\in[K]}\left(\frac{U_i(\mathbf{w^{*}})}{\beta}\right)^{1-\gamma}\right).
    \end{align*}
Under instance $I_v^j$ which ends with valuation $v$ from class $j$, we can see that:
\begin{align*}
    \frac{1}{1-\gamma}\left(\sum\nolimits_{i\in[K]}U_i(\mathbf{x})^{1-\gamma}\right) & = \frac{1}{1-\gamma} \left[\sum_{i=1}^{j-1} V_{i}(1)^{1-\gamma} +  V_{j}(v)^{1-\gamma}  + \sum_{i=j+1}^{K} V_{i}(\theta_{i})^{1-\gamma}\right],                  
\end{align*}
where this equality follows since $U_{i}(\mathbf{x}) = V_{i}(v)$ for any value of $v$ in the instance $I^j_v$ defined above. Additionally, we also have
    \begin{align*}
        \frac{1}{1-\gamma}\beta^{\gamma-1} \cdot B^{1-\gamma}\cdot \left(j-1 + v^{\frac{\gamma-1}{\gamma}} + \sum_{l=j+1}^{K} \theta_{l}^{\frac{\gamma-1}{\gamma}}\right)^\gamma = \frac{1}{1-\gamma}\left(\sum\nolimits_{i\in[K]}\left(\frac{U_i(\mathbf{w^{*}})}{\beta}\right)^{1-\gamma}\right).
    \end{align*}
Therefore, the inequality of Lemma \ref{appendix-lemma-system-of-inequalities} is necessary to ensure that an algorithm is $(\gamma,\beta)$-fair.
\end{proof}
Now, by applying the Gronwall's inequality to each of the inequalities from the above lemma, the lemma below follows.
\begin{letteredlemma}
    For any online algorithm over the hard instance $I^{\GBF}$, let us assume that the exact value of $V_{i}(1)$, $\forall i \in [1,j-1]$, and $V_{i}(\theta_{i})$, $\forall i \in [j+1,K]$, are given. Then, the following inequality holds
    \begin{align*}
        &\zeta_{j}(v) \ge \frac{1}{v} \cdot \beta^{-1} \cdot g_j(v)^{\frac{1}{1-\gamma}} + \int_{\eta=1}^{v}  \frac{\beta^{-1}}{\eta^{2}} \cdot g_j(\eta)^{\frac{1}{1-\gamma}} \cdot d\eta, \qquad \forall v \in [1,\theta_{j}),
    \end{align*} 
    where 
    \begin{align*}
        &g_{j}(v) = B^{1-\gamma}\cdot \left(j-1 + v^{\frac{\gamma-1}{\gamma}} + \sum_{l=j+1}^{K} \theta_{l}^{\frac{\gamma-1}{\gamma}}\right)^\gamma - \sum_{i=1}^{j-1} \left(\beta \cdot V_{i}(1)\right)^{1-\gamma} - \sum_{i=j+1}^{K} (\beta \cdot V_{i}(\theta_{i}))^{1-\gamma}.
    \end{align*}
\end{letteredlemma}
\begin{proof}
    The first inequality follows directly from applying Gronwall's inequality to the system of inequalities in Lemma \ref{appendix-lemma-system-of-inequalities}.
\end{proof}

Next, we are going to compute the values of $V_{j}(1)$ and $V_{j}(\theta_j)$ corresponding to the optimal online algorithm on instance $I^{\GBF}$. Let us fix the value of  $V_j(\theta_j)$ and $\zeta_{j}(\theta_{j})$ using the dummy variable $\lambda_{j}$ as follows:
\begin{align*}
    &V_{j}(\theta_{j}) = \beta^{-1} \cdot \left(B^{1-\gamma}\cdot \left(j-1 + \theta_j^{\frac{\gamma-1}{\gamma}} + \sum_{l=j+1}^{K} \theta_{l}^{\frac{\gamma-1}{\gamma}}\right)^\gamma - \sum_{i=1}^{j-1} \left(\beta \cdot V_{i}(1)\right)^{1-\gamma} + \sum_{i=j+1}^{K} (\beta \cdot V_{i}(\theta_{i}))^{1-\gamma}\right)^{\frac{1} {1-\gamma}} + \frac{\lambda_{j}}{\beta} \cdot \theta_{j}, \\
    &\zeta_{j}(\theta_{j}) =  \frac{1}{\theta_{j}} \cdot \beta^{-1} \cdot g_i(\theta_{j})^{\frac{1}{1-\gamma}} - \int_{\eta=1}^{\theta_{j}}  \frac{\beta^{-1}}{\eta^{2}} \cdot g_j(\eta)^{\frac{1}{1-\gamma}} \cdot d\eta +  \frac{\lambda_j}{\beta}.
\end{align*}
For now let us assume the values of $V_{j}(\theta_{j})$ and $V_{j}(1)$ corresponding to the optimal online algorithm are known.
Given the budget constraint, it follows that $ \sum\nolimits_{j\in[K]}\zeta_j(\theta_j) \leq B$. We argue that:
\begin{align*}
   & \sum\nolimits_{j\in[K]}\zeta_j(\theta_j) = \sum\nolimits_{j \in [K]} \frac{1}{\theta_{j}} \cdot \beta^{-1} \cdot g_j(\theta_{j})^{\frac{1}{1-\gamma}} + \int_{\eta=1}^{\theta_{j}}  \frac{\beta^{-1}}{\eta^{2}} \cdot g_j(\eta)^{\frac{1}{1-\gamma}} \cdot d\eta +  \frac{\lambda_j}{\beta} \leq B.
\end{align*}
As a result
\begin{align*}
    \beta \geq B^{-1}\left(\sum\nolimits_{j \in [K]} \frac{1}{\theta_{j}} \cdot g_j(\theta_{j})^{\frac{1}{1-\gamma}} + \int_{\eta=1}^{\theta_{j}}  \frac{1}{\eta^{2}} \cdot g_j(\eta)^{\frac{1}{1-\gamma}} \cdot d\eta + \lambda_j\right).
\end{align*}
The right-hand-side of the above inequality will give a lower bound for the $\beta$-fairness guarantee of the optimal algorithm on instance $I^\GBF$. Consequently, it provides a lower bound for the $\beta$-fairness guarantee of any online algorithm. 

In order to find the values $V_{j}(\theta_{j})$ and $V_{j}(1)$ corresponding to the optimal online algorithm, let us introduce the variables $V_{j} = \beta \cdot V_{j}(\theta_{j}) $ and $\rho_{j} = \beta \cdot V_{j}(1) $. In order to compute these values, we introduce the optimization problem in Theorem \ref{theorem-gamma-beta-lowerbound} to minimize the right-hand-side of above inequality where the decision variables for this optimization problem are $\rho_j$, $\lambda_j$ and $V_{j}$. 

\subsection{Proof of Corollary \ref{theorem-gamma-beta-order-optimal}}\label{appendix-theorem-gamma-beta-order-optimal}
In the following, we discuss the order optimality of Algorithm \ref{alg-proportional-fairness} with $K=2$ in three cases. 
\paragraph{First Case: $\gamma > 1$}
It can be proven that Algorithm \ref{alg-proportional-fairness} is in the same order of the optimal solution when $\gamma$ is larger than 1. Let us start with the lower bound of the optimal solution. We know that we can rewrite the definition of ($\gamma,\beta$)-fairness as follows:
\begin{align*}
    \left(\sum_{i=1}^K \frac{U_i^{1-\gamma}(\mathbf{w})}{K}\right)^{\frac{1}{1-\gamma}} \leq \beta \cdot \left(\sum_{i=1}^K \frac{U_i^{1-\gamma}(\mathbf{x})}{K}\right)^{\frac{1}{1-\gamma}}.
\end{align*}
This is basically the (1-$\gamma$)-mean definition of the $K$ non-negative values \cite{barman2022universal}. It is known that as $1-\gamma$ decreases (or equivalently, as $\gamma$ increases), the mean function also decreases. Therefore we can see that 
\begin{align*}
    \left(\sum_{i=1}^K \frac{U_i^{1-\gamma}(\mathbf{w})}{K}\right)^{\frac{1}{1-\gamma}} \leq \beta \cdot \left(\sum_{i=1}^K \frac{U_i^{1-\gamma}(\mathbf{x})}{K}\right)^{\frac{1}{1-\gamma}} \leq \beta \cdot \left(\sum_{i=1}^K \frac{U_i(\mathbf{x})}{K}\right).
\end{align*}
Now let us focus on the two-class case and consider an instance where the arrivals are coming from both class simultaneously and their valuations increase form 1 to some $v\leq \theta_1$ value. This instance could be considered as the following.
\begin{align*}
    I=\left\{\underbrace{(1,1), (1,2)}_{\text{first batch}}, \underbrace{(1+\epsilon,1), (1+\epsilon,2)}_{\text{second batch}}, \underbrace{(1+2\epsilon,1), (1+2\epsilon,2)}_{\text{third batch}},\dots , \underbrace{(\theta_1,1), (\theta_1,2)}_{\text{last batch}}\right\}
\end{align*}
Under this instance, we can see that the in the worst a utilization function $\zeta_i(v):[1,\theta_i]\to [0,B]$ must exist for each class such that
\begin{align*}
    &\frac{B\cdot v}{2} \leq \beta \cdot \left(\frac{Z_1(v)+Z_2(v)}{2}\right).
\end{align*}
where $Z_i(v) = \zeta_i(1) + \int_1^v \eta d \zeta(\eta)$. Let $Z(v) \coloneqq Z_1(v)+ Z_2(v)$. Therefor by applying Gronwall's inequality and substituting $v$ with $\theta_1$ we will further get
\begin{align*}
    B\cdot \alpha_1 \leq \beta \cdot \zeta(\theta_1) \leq \beta \cdot B.
\end{align*}
where $\alpha_1 = 1+\ln{\theta_1}$ and the last inequality is based on the budget constraint.

Now let us move on to find the order optimality of \U-\PRB with ($\gamma,\beta$)-fairness guarantee. As $\theta_2$ approaches to infinity, based on Theorem \ref{theorem-gamma-beta-fairness} we can see that:
\begin{align*}
    &\lim_{\theta_2\to \infty} F_1(\theta_1;\gamma) = \lim_{\theta_2\to \infty} \frac{\theta_2^{\frac{\gamma-1}{\gamma}}}{\theta_1^{\frac{\gamma-1}{\gamma}} + \theta_2^{\frac{\gamma-1}{\gamma}}} + \frac{\gamma}{\gamma-1} \ln \left(\frac{\theta_1^{\frac{\gamma-1}{\gamma}}\left(\theta_2^{\frac{\gamma-1}{\gamma}}+1\right)}{\theta_1^{\frac{\gamma-1}{\gamma}} + \theta_2^{\frac{\gamma-1}{\gamma}}}\right) = 1+ \ln{\theta_1},\\
    &\lim_{\theta_2\to \infty} F_2(\theta_2;\gamma) = \lim_{\theta_2\to \infty} \frac{\theta_1^{\frac{\gamma-1}{\gamma}}}{\theta_1^{\frac{\gamma-1}{\gamma}} + \theta_2^{\frac{\gamma-1}{\gamma}}} + \frac{\gamma}{\gamma-1} \ln \left(\frac{\theta_2^{\frac{\gamma-1}{\gamma}}\left(\theta_1^{\frac{\gamma-1}{\gamma}}+1\right)}{\theta_1^{\frac{\gamma-1}{\gamma}} + \theta_2^{\frac{\gamma-1}{\gamma}}}\right) = \frac{\gamma}{\gamma-1}\ln{\left(\theta_1^{\frac{\gamma-1}{\gamma}}+1\right)}.
\end{align*}
Now let us assume that $\beta_1\geq \beta_2$. In this case based on Theorem \ref{theorem-gamma-beta-fairness}, 
\begin{align*}
    \beta=F_1(\theta_1;\gamma)+F_2(\theta_2;\gamma) = 1+\ln{\theta_1}+\frac{\gamma}{\gamma-1}\ln{\left(\theta_1^{\frac{\gamma-1}{\gamma}}+1\right)},
\end{align*}
which is in the order of $\alpha_1$ when $\gamma$ is not approaching to 1. Also when $\beta_1 < \beta_2$, we know that $\beta\leq F_1(\theta_1;\gamma)+F_2(\theta_2;\gamma)$. Thus, $\beta$ is still in the order of $\alpha_1$ and this proofs the order optimality of the algorithm when $\gamma> 1$.

\paragraph{Second Case: $\gamma < 1$}
In order to show the order optimality of Algorithm \ref{alg-proportional-fairness} when $\gamma$ is less than 1, we begin with the lower bound of the optimal solution in this case. Let us rewrite the ($\gamma,\beta$)-fairness as:
\begin{align*}
    \left(\sum_{i=1}^K U_i^{1-\gamma}(\mathbf{w})\right)^{\frac{1}{1-\gamma}} \leq \beta \cdot \left(\sum_{i=1}^K \frac{U_i^{1-\gamma}(\mathbf{x})}{K}\right)^{\frac{1}{1-\gamma}} \cdot K^{\frac{1}{1-\gamma}}.
\end{align*}
Based on Proposition \ref{appendix-proposition-U(w)-upperbound} in the two-class case we have:
\begin{align*}
    B\cdot \left(v_1^\frac{1-\gamma}{\gamma} + v_2^\frac{1-\gamma}{\gamma}\right)^\frac{\gamma}{1-\gamma} \leq \beta \cdot \left(\frac{Z_1(v_1)^{1-\gamma} + Z_2(v_2)^{1-\gamma}}{2}\right)^{\frac{1}{1-\gamma}} \cdot 2^{\frac{1}{1-\gamma}}.
\end{align*}
The minimum of the left-hand side is $B\cdot\max\{v_1,v_2\}$ and the maximum of the right-hand side is $\beta\cdot \frac{Z_1(v_1)+Z_2(v_2)}{2}\cdot 2^\frac{1}{1-\gamma}$. Thus in order to find a lower bound of $\beta$ we can rewrite the above as:
\begin{align*}
    B\cdot \max\{v_1,v_2\} \leq \beta \cdot (Z_1(v_1)+Z_2(v_2))\cdot 2^\frac{\gamma}{1-\gamma}.
\end{align*}
Now let us consider an instance where the first arrival is from class 1 with valuation 1 (minimum valuation) and then all the arrivals are coming from class 2 and their valuations increase form 1 to some $v_2\leq \theta_2$ value. This instance could be considered as the following.
\begin{align*}
    I=\left\{\underbrace{(1,1)}_{\text{first arrival}}, \underbrace{(1,2), (1+\epsilon,2), (1+2,2), \dots, (\theta_2, 2)}_{\text{second batch}}\right\}.
\end{align*}
Therefore, under this instance we will have:
\begin{align*}
    B\cdot v_2 \leq \beta\cdot (\zeta_1(1)+Z_2(v_2))\cdot 2^\frac{\gamma}{1-\gamma}.
\end{align*}
and using the Gronwall's inequality and substituting $v_2$ with $\theta_2$ we will further obtain
\begin{align*}
    B\cdot \alpha_2 \cdot 2^\frac{\gamma}{\gamma-1} \leq \beta \cdot (\zeta_1(1)+\zeta_2(\theta_2)) \leq \beta \cdot B,
\end{align*}
where $\alpha_2 = 1+\ln{\theta_2}$ and the last inequality is based on the budget constraint.
Now let us move on to find the order optimality of our designed algorithm. We can see that in Theorem \ref{theorem-gamma-beta-fairness}, $F_1(\theta_1;\gamma)$ and $F_2(\theta_2;\gamma)$ are always decreasing with respect $\gamma$. Therefore, 
\begin{align*}
    &F_1(\theta_1;\gamma) \leq \lim_{\gamma\to 0} \frac{1}{\theta_1^{\frac{\gamma-1}{\gamma}} + 1} + \frac{\gamma}{\gamma-1} \ln \left(\frac{2 \cdot \theta_1^{\frac{\gamma-1}{\gamma}}}{\theta_1^{\frac{\gamma-1}{\gamma}} + 1}\right) = \alpha_1,\\
    &F_2(\theta_2;\gamma) \leq \lim_{\gamma\to 0} \frac{1}{\theta_2^{\frac{\gamma-1}{\gamma}} + 1} + \frac{\gamma}{\gamma-1} \ln \left(\frac{2 \cdot \theta_2^{\frac{\gamma-1}{\gamma}}}{\theta_2^{\frac{\gamma-1}{\gamma}} + 1}\right) = \alpha_2.
\end{align*}
Now let us assume that $\beta_1\geq \beta_2$. In this case based on Theorem \ref{theorem-gamma-beta-fairness}, 
\begin{align*}
    \beta=F_1(\theta_1;\gamma)+F_2(\theta_2;\gamma) = \alpha_1 + \alpha_2,
\end{align*}
which is in the order of $\alpha_2$ when $\gamma$ is not approaching to 1. Also when $\beta_1 < \beta_2$, we know that $\beta\leq F_1(\theta_1;\gamma)+F_2(\theta_2;\gamma)$. Thus, $\beta$ is still in the order of $\alpha_2$ and this proofs the order optimality of the algorithm when $\gamma> 1$.

\paragraph{Third Case: $\gamma \to 1$}
When $\gamma$ approaches to 1, we can see that based on the definition of ($\gamma,\beta$)-fairness we will have:
\begin{align*}
    \left(\prod_{i=1}^K U_i(\mathbf{w})\right)^{\frac{1}{1-\gamma}} \leq \beta \cdot  \left(\prod_{i=1}^K U_i(\mathbf{x})\right)^{\frac{1}{1-\gamma}}.
\end{align*}
This is basically the definition of Nash welfare approximation and in the worst case it is necessary that the following inequality be hold:
\begin{align*}
    \frac{B}{K}\left(\prod_{i=1}^K v_i \right)^{\frac{1}{1-\gamma}} \leq \beta \cdot  \left(\prod_{i=1}^K Z_i(v_i)\right)^{\frac{1}{1-\gamma}}.
\end{align*}
Now let us focus on the two-class case and consider an instance where the first arrival is from class 1 with valuation 1 (minimum valuation) and it is followed by arrivals from class 2 and their valuations increase from 1 to some $\theta_2$ values and Finally the last batch of arrivals start coming again from class 1 with an increasing valuations from 1 to $v_1\leq \theta_1$. This instance can be written in the follow form:
\begin{align*}
    I=\left\{\underbrace{(1,1)}_{\text{first arrival}}, \underbrace{(1,2), (1+\epsilon,2), (1+2\epsilon,2), \dots, (\theta_2,2)}_{\text{second batch}}, \underbrace{(1+\epsilon,1), (1+2\epsilon,1), (\theta_1, 1)}_{\text{third batch}}\right\}.
\end{align*}
Under this instance, we can see that:
\begin{align*}
    \frac{B}{2}\sqrt{v_2} \leq \beta \cdot \sqrt{Z_1(1)\cdot Z_2(v_2)}.
\end{align*}
Therefor we can easily by the Gronwall's inequality and setting $v_2=\theta_2$ we will further get:
\begin{align*}
    \frac{B^2}{4}(1+\ln \theta_2) \leq \beta^2 \cdot \zeta_1(1)\cdot \zeta_2(\theta_2).
    % \label{eq:geometric_2}
\end{align*}
After this, the second stream of the arrivals starts. We thus have
\begin{align*}
    &\frac{B}{2}\sqrt{\theta_2\cdot v_1} \leq \beta\cdot \sqrt{Z_1(v_1)\cdot Z_2(\theta_2)}.
\end{align*}
And by setting $v_1=\theta_1$ we will get
\begin{align*}
    \frac{B^2}{4}\theta_2\cdot (1+\ln{\theta_1}) \leq \beta^2\cdot \zeta_1(\theta_1)\cdot Z_2(\theta_2).
\end{align*}
By combining the above inequalities we will get:
\begin{align*}
   \frac{1}{2}\sqrt{\alpha_1\alpha_2} \leq \beta\cdot \sqrt{\zeta_1(\theta_1)\cdot \zeta_2(\theta_2)}.
\end{align*}
And based on AM-GM inequality we will obtain the following inequality:
\begin{align*}
    &\frac{B}{2}\sqrt{\alpha_1\alpha_2} \leq \beta\cdot \sqrt{\zeta_1(\theta_1)\cdot \zeta_2(\theta_2)} \leq \beta\frac{\zeta_1(\theta_1)+ \zeta_2(\theta_2)}{2}\leq \frac{B}{2}\beta.
\end{align*}
Therefore we can see that the lower bound of the problem when $\gamma \to 1$ is in the order of $\sqrt{\alpha_2}$. Now let us move on to find the order optimality of \U-\PRB with ($\gamma,\beta$)-fairness guarantee. As $\gamma$ approaches to 1, we can see that in Theorem \ref{theorem-gamma-beta-fairness}:
\begin{align*}
    &\lim_{\gamma\to 1} F_1(\theta_1;\gamma) = \lim_{\gamma\to 1} \frac{\theta_2^{\frac{\gamma-1}{\gamma}}}{\theta_1^{\frac{\gamma-1}{\gamma}} + \theta_2^{\frac{\gamma-1}{\gamma}}} + \frac{\gamma}{\gamma-1} \ln \left(\frac{\theta_1^{\frac{\gamma-1}{\gamma}}\left(\theta_2^{\frac{\gamma-1}{\gamma}}+1\right)}{\theta_1^{\frac{\gamma-1}{\gamma}} + \theta_2^{\frac{\gamma-1}{\gamma}}}\right) = \frac{1}{2}+ \frac{\ln{\theta_1}}{2},\\
    &\lim_{\gamma\to 1} F_2(\theta_2;\gamma) = \lim_{\gamma\to 1} \frac{\theta_1^{\frac{\gamma-1}{\gamma}}}{\theta_1^{\frac{\gamma-1}{\gamma}} + \theta_2^{\frac{\gamma-1}{\gamma}}} + \frac{\gamma}{\gamma-1} \ln \left(\frac{\theta_2^{\frac{\gamma-1}{\gamma}}\left(\theta_1^{\frac{\gamma-1}{\gamma}}+1\right)}{\theta_1^{\frac{\gamma-1}{\gamma}} + \theta_2^{\frac{\gamma-1}{\gamma}}}\right) = \frac{1}{2}+ \frac{\ln{\theta_2}}{2}.
\end{align*}
In this case we can see that the only possibility is $\beta_1\leq \beta_2$. Furthermore, as $\gamma\to 1$ we will have:
\begin{align*}
    &\lim_{\gamma\to 1} \beta_1 = \alpha_1, \qquad \lim_{\gamma\to 1} \beta_2 = \alpha_2,
\end{align*}
and then we can see that:
\begin{align*}
    \lim_{\gamma\to 1} \beta = \sqrt{\alpha_1\alpha_2},
\end{align*}
which is the same as the lower bound in this case and therefore the order optimality is proven.

\subsection{Proof of Theorem \ref{theorem-gamma-beta-fairness-trade-off-design}}\label{appendix-theorem-gamma-beta-fairness-trade-off-design}
The design of the increasing threshold functions $\phi_j(\cdot)$, $\forall j \in [K]$, is based on the inverse of the utilization function $\psi_j(\cdot)$, and follows the same structure as outlined in Theorem \ref{theorem-gamma-beta-fairness}. Therefore, we can argue that Algorithm \ref{theorem-proportional-fairness-beta} is ($\gamma,\beta$)-fair for any $\beta$ greater than the one specified in Theorem \ref{theorem-gamma-beta-fairness}. Next, we determine the competitive ratio of Algorithm \ref{alg-proportional-fairness-tradeoff-local-globol}, based on the threshold design from Theorem \ref{theorem-gamma-beta-fairness-trade-off-design}. Consider an input instance $I$ where, for each class $i \in [K]$, the agent with the highest valuation has a value of $v_i$. The optimal offline objective, $\OPT(I)$, for this input instance is then upper-bounded as follows:
\begin{align*}
    \OPT(I) \leq \max_{i\in [K]} \{B\cdot v_i\}.
\end{align*}
Let $\ALG(I)$ denote the objective of Algorithm \ref{alg-proportional-fairness-tradeoff-local-globol} on the input instance $I$. It follows that: 
\begin{align*}
    \ALG(I) \geq \sum\nolimits_{i\in[K]} \Psi_i(v_i) + \Psi^G(\max\{v_i\},
\end{align*}
where $\Psi_i(v) = \psi_i(1) + \int_1^v \eta d \psi_i(\eta)$. Based on the design of utilization functions $\psi_{i}$ and $\psi^{G}$ done in Theorem \ref{theorem-gamma-beta-fairness-trade-off-design}, it follows that:
\begin{align*}
    \ALG(I) \ge& \left( \sum_{i \in [K]}  \psi_{i}(1) + \int_{\eta = i}^{p_{i}} \eta ~ d\psi_{i}(\eta) \right) + \psi^{G}(1) + \int_{\eta=1}^{\max{\{p_{i}\}_{i \in [K]}}} \eta ~d\psi^{G}(\eta)\\
    \ge& \sum_{i\in[K]} \frac{B}{\beta_i}\cdot \frac{v_i}{\sum_{j\in[i^-]} \left(\frac{v_i}{\theta_j}\right)^\frac{\gamma-1}{\gamma} + 1} +\frac{B - B\cdot \sum_{i\in[K-1]} \frac{F_i(\theta_i) - F_i(1)}{\beta_i}}{\alpha_K}\cdot\max\{v_i\} - \\&\frac{B}{\beta_K}\cdot\frac{\max\{v_i\}}{\sum_{j\in[K^-]} \left(\frac{\max\{v_i\}}{\theta_j}\right)^\frac{\gamma-1}{\gamma} + 1} - \sum_{j\in[K-1]} \frac{B}{\beta_j}F_j(1),
\end{align*}
where in above the second inequality follows from the definition of $\psi_{i}$ functions. The minimum of the right-hand side occurs when $v_i = 1$ for all $i \in [K-1]$ and $v_K = \theta_K$. Therefore:
\begin{align*}
      \ALG(I) \ge &\sum_{i\in[K-1]} \frac{B}{\beta_i}\cdot \frac{1}{\sum_{j\in[i^-]} \left(\frac{1}{\theta_j}\right)^\frac{\gamma-1}{\gamma} + 1} + \frac{B}{\beta_K}\cdot \frac{\theta_K}{\sum_{j\in[K^-]} \left(\frac{\theta_K}{\theta_j}\right)^\frac{\gamma-1}{\gamma} + 1} + \\&\frac{B - B\cdot \sum_{i\in[K-1]} \frac{F_i(\theta_i) - F_i(1)}{\beta_i}}{\alpha_K}\cdot \theta_K - \frac{B}{\beta_K}\cdot\frac{\theta_K}{\sum_{j\in[K^-]} \left(\frac{\theta_K}{\theta_j}\right)^\frac{\gamma-1}{\gamma} + 1} - \sum_{i\in[K-1]} \frac{B}{\beta_i}F_i(1).
\end{align*}
Since $F_i(1) = \frac{1}{\beta_i}\cdot \frac{1}{\sum_{j\in[i^-]} \left(\frac{1}{\theta_j}\right)^\frac{\gamma-1}{\gamma} + 1}$, we can further simplify the above inequality to
\begin{align*}
    \ALG(I) \ge &\frac{B - B\cdot \sum_{i\in[K-1]} \frac{F_i(\theta_i) - F_i(1)}{\beta_i}}{\alpha_K}\cdot \theta_K.
\end{align*}
Consequently, it follows that:
\begin{align*}
    \ALG(I) \ge \frac{B - B\cdot \sum_{i\in[K-1]} \frac{F_i(\theta_i) - F_i(1)}{\beta_i}}{\alpha_K}\cdot \theta_K \ge \OPT(I) \cdot \frac{1 - \sum_{i\in[K-1]} \frac{F_i(\theta_i) - F_i(1)}{\beta_i}}{\alpha_K}.
\end{align*}
Therefore in order to find the minimum competitive ratio, it is sufficient to solve the optimization problem presented in Theorem $\ref{theorem-gamma-beta-fairness-trade-off-design}$ for a given $\beta$.

\end{document}